\documentclass[11pt,reqno,twoside]{article}

\usepackage{fixltx2e} 

\usepackage{cmap} 

\usepackage[T1]{fontenc}
\usepackage[utf8]{inputenc}
\usepackage{graphicx}
\usepackage{placeins}
\usepackage{enumerate}
\usepackage{algcompatible}
\usepackage{bibentry}
\usepackage{subcaption}

\usepackage{verbatim}

\usepackage{setspace}

\let\counterwithin\relax  
\usepackage{lmodern} 
\usepackage[scale=0.88]{tgheros} 

\usepackage{bm} 

\usepackage{amsmath,amsbsy,amsgen,amscd,amsthm,amsfonts,amssymb} 

\usepackage[centering,top=1.5in,bottom=1.2in,left=1in,right=1in]{geometry}

\usepackage{titling}

\graphicspath{ {./images/} }

\usepackage[sf,bf,compact]{titlesec}

\usepackage{booktabs,longtable,tabu} 
\usepackage[font=small,margin=12pt,labelfont={sf,bf},labelsep={space}]{caption}

\usepackage[usenames,dvipsnames]{xcolor}

\definecolor{dark-gray}{gray}{0.3}
\definecolor{dkgray}{rgb}{.4,.4,.4}
\definecolor{dkblue}{rgb}{0,0,.5}
\definecolor{medblue}{rgb}{0,0,.75}
\definecolor{rust}{rgb}{0.5,0.1,0.1}

\usepackage{url}
\usepackage[colorlinks=true]{hyperref}
\hypersetup{linkcolor=dkblue}    
\hypersetup{citecolor=rust}      
\hypersetup{urlcolor=rust}     

\usepackage[final]{microtype} 

\newtheoremstyle{myThm} 
    {\topsep}                    
    {\topsep}                    
    {\itshape}                   
    {}                           
    {\sffamily\bfseries}                   
    {.}                          
    {.5em}                       
    {}  

\newtheoremstyle{myRem} 
    {\topsep}                    
    {\topsep}                    
    {}                   
    {}                           
    {\sffamily}                   
    {.}                          
    {.5em}                       
    {}  

\newtheoremstyle{myDef} 
    {\topsep}                    
    {\topsep}                    
    {}                   
    {}                           
    {\sffamily\bfseries}                   
    {.}                          
    {.5em}                       
    {}  

\theoremstyle{myThm}
\newtheorem{theorem}{Theorem}[section]
\newtheorem{lemma}[theorem]{Lemma}
\newtheorem{proposition}[theorem]{Proposition}
\newtheorem{corollary}[theorem]{Corollary}

\newtheorem{assumption}[theorem]{Assumption}

\theoremstyle{myRem}
\newtheorem{remark}[theorem]{Remark}

\theoremstyle{myDef}

\usepackage{fancyhdr}
\usepackage{nopageno} 
\usepackage{algorithm}
\usepackage{algpseudocode}

\usepackage{ulem}
\let\originalleft\left
\let\originalright\right
\renewcommand{\left}{\mathopen{}\mathclose\bgroup\originalleft}
\renewcommand{\right}{\aftergroup\egroup\originalright}

\usepackage{mathtools}
\mathtoolsset{centercolon}  

\newcommand{\eps}{\varepsilon}

\providecommand{\mathbbm}{\mathbb} 

\newcommand{\R}{\mathbbm{R}}

\newcommand{\OO}{\mathcal{O}}

\renewcommand{\L}{\mathcal{L}}

\definecolor{mygreen}{rgb}{0.1,0.75,0.2}

\newcommand{\nc}{\normalcolor}

\newcommand{\X}{\mathbb{X}}

\newcommand{\M}{\mathcal{M}}

\newcommand{\Nc}{\mathcal{N}}

\newcommand{\LT}{\text{Lip}(\tau)}
\newcommand{\LK}{\text{Lip}(\kappa)}
\usepackage[]{algorithm}

\usepackage{graphicx}

\usepackage{soul}
\usepackage{authblk}
\usepackage[square,numbers]{natbib}
\usepackage{chngcntr}
\usepackage{mathrsfs}
\counterwithin{table}{section}
\counterwithin{algorithm}{section}

\title{The SPDE Approach to Mat\'ern Fields: Graph Representations} 
\author{Daniel Sanz-Alonso and Ruiyi Yang}

\vspace{.25in} 

\date{University of Chicago}

\makeatletter\@addtoreset{section}{part}\makeatother%
\numberwithin{equation}{section}

\newcommand{\upperRomannumeral}[1]{\uppercase\expandafter{\romannumeral#1}}

\begin{document}
\maketitle 

\begin{abstract}

This paper investigates Gaussian Markov random field approximations to nonstationary Gaussian fields using graph representations of stochastic partial differential equations. 
We establish approximation error guarantees building on the theory of spectral convergence of graph Laplacians. The proposed graph representations provide a generalization of the Mat\'ern model to unstructured point clouds, and facilitate inference and sampling using linear algebra methods for sparse matrices. In addition, they bridge and unify several models in Bayesian inverse problems, spatial statistics and graph-based machine learning. We demonstrate through examples in these three disciplines that the unity revealed by graph representations facilitates the exchange of ideas across them. 
\end{abstract}

\section{Introduction}

The stochastic partial differential equation (SPDE) approach to Gaussian fields (GFs) has been one of the key developments in spatial statistics over the last decade \cite{lindgren2011explicit}.  
The main idea is to represent GFs as finite element solutions to SPDEs, reducing the computational cost of inference and sampling by invoking a Gaussian Markov random field (GMRF) approximation \cite{rue2005gaussian}. This paper investigates \emph{graph representations} of stationary and nonstationary Mat\'ern fields following the SPDE perspective, contributing to and unifying the extant theoretical, computational and methodological literature on GFs in Bayesian inverse problems, spatial statistics and graph-based machine learning.  We demonstrate through transparent mathematical reasoning that, under a manifold assumption, graph representations give GMRF approximations to the Mat\'ern model with error guarantees. In addition, we show that graph representations \emph{generalize} the Mat\'ern model to unstructured point clouds and graphs where existing finite element representations are not applicable. 


Recall that a random function $u(x),$  $x\in\R^d,$ is a GF if all finite collections $\{u(x_i)\}_{i=1}^n$ have self-consistent multivariate Gaussian distributions \cite{abrahamsen1997review,adler2010geometry}. A GF can be specified using a mean function $\mu(\cdot)$ and a covariance function $c(\cdot, \cdot),$ so that the mean vector and covariance matrix of the finite dimensional distributions are $\{\mu(x_i)\} \in \R^n$ and $\Sigma = \{c(x_i,x_j)\} \in \R^{n\times n}.$
GFs are natural models for spatial, temporal and spatio-temporal data. They have desirable analytic properties, including an explicit normalizing constant and closed formulae when conditioning  on  Gaussian data. 
However, in practice GFs have two main caveats. First, it is crucial and non-trivial to find \emph{flexible} covariance functions with few but interpretable parameters that can be learned from data. Second, inference of these parameters from Gaussian data of size $n$ ---or sampling the field at $n$ locations--- involves factorizing a kernel matrix $\Sigma \in \R^{n\times n},$  leading to a $\OO (n^3)$  computational cost and $\OO(n^2)$ memory cost unless further structure is assumed or imposed on the covariance model. For this reason, many recent works have investigated novel ways to deal with large datasets, some of which are reviewed in \cite{heaton2019case}. 

The SPDE approach tackles the big $n$ problem by replacing the GF with a GMRF approximation. A GMRF is a discretely indexed GF $u_n(i)$, $i\in \{1, \ldots, n\},$  such that the full conditional distribution at each site $1\le i \le n$ depends only on a (small) set of neighbors $\partial i$ to site $i$. This conditional independence structure is fully encoded in the precision matrix $Q$ of the multivariate Gaussian distribution of $u_n\in \R^n$: it holds that $Q_{ij} \ne 0$ iff  $i \in \partial j.$  Computationally, the main advantage comes from using numerical linear algebra techniques and Markov chain Monte Carlo algorithms that exploit, respectively, the sparsity of $Q$ and the characterization of the GMRF in terms of its full conditionals. The speed-up can be dramatic, with a typical computational cost $\OO (n),$ $\OO (n^{3/2}),$ and $\OO (n^2)$ for GMRFs in time, space, and space-time in two spatial dimensions, see \cite{rue2005gaussian}. In addition to alleviating the computational burden of GF methods, the SPDE approach  also alleviates the modeling challenges by suggesting nonstationary generalizations of Mat\'ern fields and extensions beyond Euclidean settings.

In this paper we employ graph-based discretizations of SPDEs to represent stationary and nonstationary Mat\'ern models. With few exceptions e.g. \cite{garcia2018continuum,bertozzi2018uncertainty,trillos2017consistency,harlim2020kernel}, previous work stemming from the SPDE approach considered representations based on finite element or finite difference discretizations \cite{lindgren2011explicit,bolin2019rational,bolin2020numerical, bolin2018weak,roininen2019hyperpriors,wiens2020modeling,bolin2014spatial}. Graph representations provide a way to \emph{generalize} the Mat\'ern model to discrete and unstructured point clouds, and thus to settings of practical interest in statistics and machine learning where only similarity relationships between abstract features may be available. \nc  Moreover, in contrast to finite elements, graph representations require minimal pre-processing cost: there is no need to compute triangulations and finite element basis or to define ghost domains as in \cite{lindgren2011explicit}. This is an essential advantage when interpolating manifold data living in a high dimensional ambient space, particularly so when the underlying manifold or its dimension are unknown. Finally, a wide range of problems in Bayesian inversion, spatial statistics and graph-based machine learning can be formulated as latent Gaussian models, and using graph representations of Mat\'ern fields as priors allows us to unify and contribute to the exchange of ideas across these disciplines. 

A disadvantage of the graph-based approach is that error guarantees are weaker than for finite element or finite difference representations.  Our belief is that this  is due to the generality of the graph-based approach, and also to the underdevelopment of existing theory.   Here we provide an up-to-date perspective of spectral convergence of graph Laplacians which overviews and generalizes some of the recent literature  \cite{burago2015graph,trillos2019error,ruiyilocalregularization} and we further show how these results can be used to establish the convergence of GMRFs to GFs. We view graph representations as being complementary to, rather than a replacement for, finite element representations. If the underlying domain is known and a suitable mesh can be obtained, then finite element representations would be recommended on the grounds of better error guarantees and sparsity, see Subsection \ref{sec:unstruct}. However, graph-based methods are more broadly applicable, and in particular  generalize the Mat\'ern model to unstructured point clouds as will be demonstrated in our numerical examples in Subsections \ref{sec:SS} and \ref{sec:ML}.

\subsection{Literature Review}
The ubiquity of GFs in statistics, applied mathematics and machine learning has led, unsurprisingly, to the reinvention and relabeling of many algorithms and ideas. GFs play a central role in spatial statistics \cite{gelfand2010handbook,heaton2019case}, especially in the subfield of geostatistics \cite{stein2012interpolation}, where they are used to interpolate data in a procedure called  \emph{kriging} and as a building block of modern hierarchical spatial models \cite{banerjee2014hierarchical}. In machine learning, GFs are called Gaussian processes and kriging is known as Gaussian process regression \cite{williams2006gaussian}. Gaussian processes are one of the main tools in Bayesian non-parametric inference \cite{williams1996gaussian,van2008rates,trillos2017consistency} and are an alternative to neural networks for supervised and semi-supervised regression \cite{mackay1997gaussian,trillos2017consistency}. They are also related to, or used within, other machine learning algorithms including splines, support vector machines and Bayesian optimization \cite{sollich2002bayesian,seeger2000relationships,brochu2010tutorial,frazier2018tutorial}.
GFs are standard prior models for statistical Bayesian inverse problems \cite{kaipio2006statistical,calvetti2007introduction,AS10,sanzstuarttaeb} with applications in medical imaging, remote sensing and ground prospecting \cite{bardsley2013gaussian,dunlop2017hierarchical,somersalo1992existence,dunlop2016bayesian,trillos2016bayesian}. Within Bayesian inversion, GFs are also employed as surrogates for the likelihood function \cite{stuart2018posterior}.  GFs have found numerous applications, allowing for \emph{uncertainty quantification} \cite{sullivan2015introduction} in astrophysics \cite{bardeen1985statistics},  biology \cite{taylor2007detecting,stathopoulos2014bat}, calibration of computer models \cite{kennedy2001bayesian,martin2005use}, data-driven learning of partial differential equations  \cite{raissi2018numerical,raissi2017machine}, geophysics \cite{isaac2015scalable,bui2013computational}, hydrology  \cite{sanchez2006representative}, image processing and medical imaging \cite{cohen1991classification,somersalo1992existence,roininen2014whittle}, meteorology \cite{bolin2011spatial,lindgren2011explicit} and probabilistic numerics \cite{hennig2015probabilistic,kersting2016active}, among others.

The emphasis of this paper is on  Mat\'ern models \cite{matern2013spatial} and generalizations thereof. Mat\'ern models are widely used in spatial statistics \cite{stein2012interpolation,gelfand2010handbook}, machine learning \cite{williams2006gaussian} and uncertainty quantification \cite{sullivan2015introduction}, with applications in various scientific fields \cite{guttorp2006studies,cameletti2013spatio}. The SPDE approach to construct GMRF approximations to GFs was proposed in the seminal paper \cite{lindgren2011explicit} and was further popularized through the software R-INLA \cite{bakka2018spatial}. GMRFs in statistics were pioneered by Besag \cite{besag1974spatial,besag1975statistical} and their computational benefits and applications are overviewed in the monograph \cite{rue2005gaussian}. 

In an independent line of work, the desire to define positive semi-definite kernels using only similarity relationships between features motivated the introduction of diffusion kernels \cite{kondor2002diffusion},  which can be interpreted as limiting cases of Mat\'ern models. The main idea underlying the construction of diffusion kernels is to exploit that graph Laplacians \cite{chung1997spectral,von2007tutorial} and their powers satisfy the positive semi-definiteness requirement. This observation has permeated the construction of graph-based regularizations in manifold learning and machine learning applications, as well as the design of model reduction techniques e.g. 
\cite{zhu1965semi,ng2018bayesian,li2018graph,liu2014bayesian,belkin2004regularization,belkin2005towards,belkin2004semi}.
Our work aims to demonstrate that a wide family of graph-based kernels in machine learning may be interpreted, in a rigorous way, as discrete approximations of standard GF models in spatial statistics. 

Large sample limits of graph Laplacians have been widely studied. Most results concern either pointwise convergence \cite{hein2005graphs,belkin2005towards,GK,Hei2006,singer06,THJ} or variational and spectral convergence \cite{belkin2007convergence,SinWu13,tao2020convergence,burago2015graph,ruiyilocalregularization},  with \cite{calder2019improved} reconciling both perspectives to  obtain improved rates. This paper builds on and generalizes  spectral convergence theory ---that is, convergence of eigenvalues and eigenfunctions of the graph-based operators to those defined in the continuum--- to study GMRF approximations of GFs. Unsurprisingly, we shall see that optimal transport ideas are key to linking discrete and continuum objects.

\subsection{Main Contributions and Outline}
Further to providing a unified narrative of existing literature, this paper contains some original contributions. We introduce GMRF approximations of nonstationary GFs defined on manifolds through graph representations of the corresponding SPDEs and generalize the constructions to arbitrary point clouds.   Our  main theoretical result, Theorem \ref{thm:Rate}, covers nonstationary models and, to our knowledge, is the first to give rates of convergence of graph-based representations of GFs.
 We also demonstrate through numerical examples that the mathematical unity that comes from viewing graph-based methods as discretizations of continuum ones facilitates the transfer of methodology and theory across Bayesian inverse problems, spatial statistics and graph-based machine learning.   In particular, we introduce nonstationary models for graph-based classification problems, which to our best knowledge has not been considered before, and empirically observe an improvement of performance that deserves further research.

This paper is organized as follows. 
Section \ref{sec:maternGFs} introduces the SPDE formulation of the Mat\'ern model and extends it to incorporate nonstationarity. 
Section \ref{sec:graphdiscretizations} introduces the graph-based approach and constructs graph approximations of the Mat\'ern fields. 
Section \ref{sec:theory} presents the main result on the convergence of the graph Mat\'ern model towards its continuum counterpart and discusses the ideas of the proof.
Section \ref{sec:numericalexamples} illustrates the application of graph Mat\'ern models in  Bayesian inverse problems,  spatial statistics and graph-based machine learning. 
Section \ref{sec:conclusions and Open Directions} discusses further research directions. Our aim is to provide a digestible narrative and for this reason we postpone all proofs and most of the technical material to an appendix. 

We close this section by introducing some notation. The symbol $\lesssim$ will denote less than or equal to up to a universal constant. For two real sequences $\{a_n\}$ and $\{b_n\}$, we denote (i) $a_n\ll b_n$ if $\operatorname{lim}_n (b_n/a_n)=0$; (ii) $a_n=O(b_n)$ if $\operatorname{lim\, sup}_n (b_n/a_n)\leq C$ for some positive constant $C$; and  (iii) $a_n\asymp b_n$ if  $c_1\leq \operatorname{lim\,inf}_n (a_n/b_n) \leq \operatorname{lim\,sup}_n (a_n/b_n) \leq c_2$ for some positive constants $c_1,c_2$.

\section{Mat\'ern Models and the SPDE Approach}\label{sec:maternGFs}
In this section we provide some background on GFs and the SPDE approach. We introduce the Mat\'ern family in Subsection \ref{ssec:stationary} and a nonstationary generalization in Subsection \ref{sec:nonstationary}. 
All fields will be assumed to be centered and we focus our attention on their covariance structure.

\subsection{Stationary Mat\'ern Models}\label{ssec:stationary}
Recall that a GF in $\R^m$ belongs to the Mat\'ern class if its covariance function can be written in the form 
\begin{align}\label{eq:materncovaria}
    c_{\sigma,\nu,\ell}(x,y) = \sigma^2\frac{2^{1-\nu}}{\Gamma(\nu)} \left(\frac{|x-y|}{\ell}\right)^{\nu}K_{\nu}\left(\frac{|x-y|}{\ell}\right),  \quad \quad x,y \in \R^m,
\end{align}
where $|\cdot |$ is the Euclidean distance in $\R^m,$ $\Gamma$ denotes the Gamma function and $K_{\nu}$ denotes the modified Bessel function of the second kind. 
The parameters $\sigma, \nu$ and $\ell$ control, respectively, the marginal variance (magnitude), regularity and correlation length scale of the field. While being defined in terms of three interpretable parameters, the modeling flexibility afforded by the Mat\'ern covariance \eqref{eq:materncovaria} is limited by its \emph{stationarity} (the value of the covariance function depends only on the difference between its arguments) and \emph{isotropy} (it depends only on their Euclidean distance).

An important characterization by Whittle \cite{whittle1954stationary, whittle1963stochastic} is that Mat\'ern fields can be defined as the solution to certain fractional order stochastic partial differential equation (SPDE). 
Precisely, setting
$\tau:= \ell^{-2}, s:=\nu+\frac{m}{2}$,
a Gaussian field with covariance function \eqref{eq:materncovaria} is the unique stationary solution to the SPDE
\begin{align} \label{eq:SPDE}
    (\tau I- \Delta)^{\frac{s}{2}} u(x) =\mathcal{W}(x),  \quad \quad x \in \R^m,
\end{align}
where the marginal variance of $u$ is 
\begin{align}
    \sigma^2=\frac{\Gamma(s-\frac{m}{2})}{(4\pi)^{\frac{m}{2}}\Gamma(s)\tau^{s-\frac{m}{2} }}\,. \label{eq:marginalvariance}
\end{align}
Throughout this paper, fractional power operators such as $(\tau I- \Delta)^{\frac{s}{2}} $ will be defined spectrally \cite{lischke2020fractional} and $\mathcal{W}$ denotes spatial Gaussian white noise with unit variance. 

As discussed in \cite{lindgren2011explicit}, the SPDE formulation of Mat\'ern GFs has several advantages. First, it allows to approximate the solution to \eqref{eq:SPDE} by a GMRF and thereby to  reduce the computational cost of inference and sampling \cite{rue2005gaussian,simpson2012think}.  Second, it suggests natural nonstationary and anisotropic generalizations of the Mat\'ern model by letting $\tau$ depend on the spatial variable \cite{roininen2019hyperpriors}
or by replacing the Laplacian with an elliptic operator with spatially varying coefficients \cite{fuglstad2015exploring, fuglstad2015does}. Third, it allows to \emph{define}  Mat\'ern models in manifolds and in bounded spatial, temporal and spatio-temporal domains  by using modifications of the SPDE \eqref{eq:SPDE}, possibly supplemented with appropriate boundary conditions \cite{khristenko2019analysis}.  In order to gain theoretical understanding, in subsequent sections we will work under a manifold assumption and analyze the convergence of graph representations of Mat\'ern fields defined on manifolds. This setting is motivated by manifold learning theory \cite{belkin2004regularization,ruiyilocalregularization} and will allow us to build on the rich literature on GFs on manifolds \cite{adler2010geometry}.

In more mathematical terms, the SPDE characterization shifts attention from the covariance function (or spectral density) description of Gaussian measures to the covariance (or precision) operator description \cite{bogachev1998gaussian}: keeping only the $\tau$ term in the marginal variance given by equation  \eqref{eq:marginalvariance}, we see that the law of the field $u(x)$ defined by equation \eqref{eq:SPDE} is ---up to a scaling factor independent of $\tau$ that we drop in what follows--- the Gaussian measure 
\begin{align}\label{eq:stationaryGMcontinuum}
   \mathcal{N}(0,\mathcal{C}), \quad  \mathcal{C}=\tau^{s-\frac{m}{2}} (\tau I -\Delta)^{-s}, 
\end{align}
where the factor $\tau^{s-\frac{m}{2}}$ can be interpreted as a normalizing constant. This observation motivates our definition of the nonstationary Mat\'ern field in Subsection \ref{sec:nonstationary}, which facilitates the theory.

\subsection{Nonstationary Mat\'ern Models}  \label{sec:nonstationary}
In this subsection we introduce a family of  nonstationary Mat\'ern fields by modifying the SPDE \eqref{eq:SPDE}. We consider a manifold setting which does not hinder the understanding of the modeling and will later allow us to frame the analysis in a concrete setting of applied significance. To that end, we let $\M$ be an $m$-dimensional smooth, connected, compact Riemannian  manifold without boundary that is embedded in $\mathbb{R}^d$. We will let $\tau$ depend on the spatial variable and replace the Laplacian by an elliptic operator $\nabla \cdot ( \kappa(x) \nabla)$, where differentiation is defined on $\M$. Formally,  we consider the  SPDE
\begin{align} \label{eq:SPDEnoncont}
    \Big[\tau(x) I-\nabla \cdot ( \kappa(x) \nabla)\Big]^{\frac{s}{2}} u(x) = \mathcal{W}(x), \quad \quad x \in \M,
\end{align}
where $\mathcal{W}$ is a spatial Gaussian white noise with unit variance on $\M$.
The additional $\kappa$ acts as a change of coordinate $\tilde{x}=\sqrt{\kappa(x)}x$ and introduces a factor of $\kappa(x)^{-\frac{m}{2}}$ for the marginal variance, whence the field $u(x)$ in equation \eqref{eq:SPDEnoncont} has marginal variance proportional to $\tau(x)^{\frac{m}{2}-s}\kappa(x)^{-\frac{m}{2}}$ at each location. If $\M = \R^d$, the solution $u$ to \eqref{eq:SPDEnoncont} defines a nonstationary field, and in analogy we will use the term nonstationary for fields defined by the SPDE \eqref{eq:SPDEnoncont}, or approximations thereof, in manifold and more abstract settings. In such settings, stationarity or ``shift-invariance'' is not well-defined without introducing an algebraic action, and nonstationarity should be understood as nonhomogeneity.

Following again the covariance operator viewpoint, we formally consider the Gaussian measure $\mathcal{N}(0,[\L^{\tau,\kappa}]^{-s})$, where $\L^{\tau,\kappa}:= \tau I-\nabla \cdot ( \kappa \nabla)$, with a proper normalization to be made precise below. Assuming sufficient regularity, $\L^{\tau,\kappa}$ is self-adjoint with respect to the $L^2(
\M)$ inner product and admits a spectral decomposition.
Therefore we shall \emph{define} our nonstationary Mat\'ern field through the following Karhunen-Lo\'eve expansion,
\begin{align}\label{eq:KLcontinuumNonstationary}
    u(x)
    &:= \tau(x)^{\frac{s}{2}-\frac{m}{4}}\kappa(x)^{\frac{m}{4}} \sum_{i=1}^{\infty} \left[\lambda^{(i)}\right]^{-\frac{s}{2}} \xi^{(i)} \psi^{(i)}(x),
\end{align}
where $\{\xi^{(i)}\}_{i=1}^{\infty}$ is a sequence of independent standard normal random variables and $\{(\lambda^{(i)},\psi^{(i)})\}_{i=1}^{\infty}$ are the eigenpairs of $\L^{\tau,\kappa}:= \tau I-\nabla \cdot ( \kappa \nabla) $. The factor $\tau(x)^{\frac{s}{2}-\frac{m}{4}}\kappa(x)^{\frac{m}{4}}$ serves as a normalizing constant for the marginal variance at each point.  For the theory outlined in Section \ref{sec:theory} we will assume that $\tau$ is Lipschitz, $\kappa$ is continuously differentiable and both are bounded from below by a positive constant,  whence Weyl's law \cite{canzani2013analysis}[Theorem 72] implies that $\lambda^{(i)}\asymp i^{\frac{2}{m}}$. Therefore by setting $s>\frac{m}{2}$, we have $\mathbb{E}\|u\|^2_{L^2(\M)}<\infty$ and the series \eqref{eq:KLcontinuumNonstationary} converges in $L^2(\M)$ almost surely. The idea of viewing the functions $\tau$ or $\kappa$ as hyperparameters and learning them from data has been investigated in \cite{fuglstad2015exploring,roininen2019hyperpriors,monterrubio2020posterior,fuglstad2015does,wiens2020modeling} and has motivated the need to penalize the complexity of priors \cite{fuglstad2019constructing}. We note that other approaches to introduce nonstationarity that do not stem directly from the SPDE formulation have been considered in the literature (e.g. \cite{anderes2008estimating,gramacy2008bayesian,kim2005analyzing,montagna2016computer,sampson2001advances}).

\begin{remark} \label{rmk:normalization}
The normalizing factors $\tau(x)^{\frac{s}{2}-\frac{m}{4}}\kappa(x)^{\frac{m}{4}}$ are crucial for hierarchical models in that they balance the marginal variances at different locations. To gain more intuition on the powers, consider the case where both $\tau$ and $\kappa$ are constant. Weyl's law then implies that  the eigenvalues of $\L^{\tau,\kappa}$ satisfy 
 $$\lambda^{(i)} \asymp \tau+ C\kappa  i^{\frac{2}{m}},$$
and therefore
\begin{align*}
	\mathbb{E}\left\|\sum_{i=1}^{\infty} \left[\lambda^{(i)}\right]^{-\frac{s}{2}}\xi^{(i)} \psi^{(i)}\right\|^2_{L^2(\M)} 
	& \asymp \sum_{i=1}^{\infty} \left[\tau +C\kappa i^{\frac{2}{m}} \right]^{-s} \\
	& \asymp \sum_{i: \tau \gtrsim \kappa i^{2/m}} \tau^{-s} + \sum_{i: \tau \lesssim \kappa i^{2/m}} \kappa^{-s}  i^{-\frac{2s}{m}} \\
	& \asymp \tau^{-s} \left(\frac{\tau}{\kappa}\right)^{\frac{m}{2}} + \kappa^{-s} \int_{\left(\frac{\tau}{\kappa}\right)^{\frac{m}{2}}} ^{\infty}  x^{-\frac{2s}{m}} \\
	& \asymp \tau^{\frac{m}{2}-s} \kappa^{-\frac{m}{2}}.
\end{align*}
We thus see that the normalizing factor above balances the expected norm. 
 \qed
\end{remark}

\begin{remark} \label{rmk:tau&kappa}
Both parameters $\tau$ and $\kappa$ control the local length scales of the sample paths. To see this,  note that  when both $\tau$ and $\kappa$ are constant \eqref{eq:KLcontinuumNonstationary} simplifies to 
\begin{align*}
    u=\left(\frac{\tau}{\kappa}\right)^{\frac{s}{2}-\frac{m}{4}}\sum_{i=1}^{\infty}\left[\frac{\tau}{\kappa}+\lambda^{(i)}\right]^{-\frac{s}{2}} \xi^{(i)}\psi^{(i)},
\end{align*}
where $\{(\lambda^{(i)},\psi^{(i)})\}_{i=1}^{\infty}$ are eigenpairs of $-\Delta$. Therefore $\frac{\tau}{\kappa}$ acts as a threshold on the essential frequencies of the samples, where those frequencies with corresponding eigenvalue on the same order of $\frac{\tau}{\kappa}$ have effective contributions. Hence a large $\tau$ (or a smaller $\kappa$) incorporates higher frequencies and gives sample paths with small length scale. Their opposite role in controlling the local length scale can be seen in Figure \ref{figure:tau&kappa},  which represent two random draws from Gaussian fields defined on the unit circle with different choices of $\kappa$ and $\tau$. \qed 
\end{remark}
\nc

\begin{figure}[!htb]
\minipage{1\textwidth}
\centering
\minipage{0.4\textwidth}
  \includegraphics[width=\linewidth]{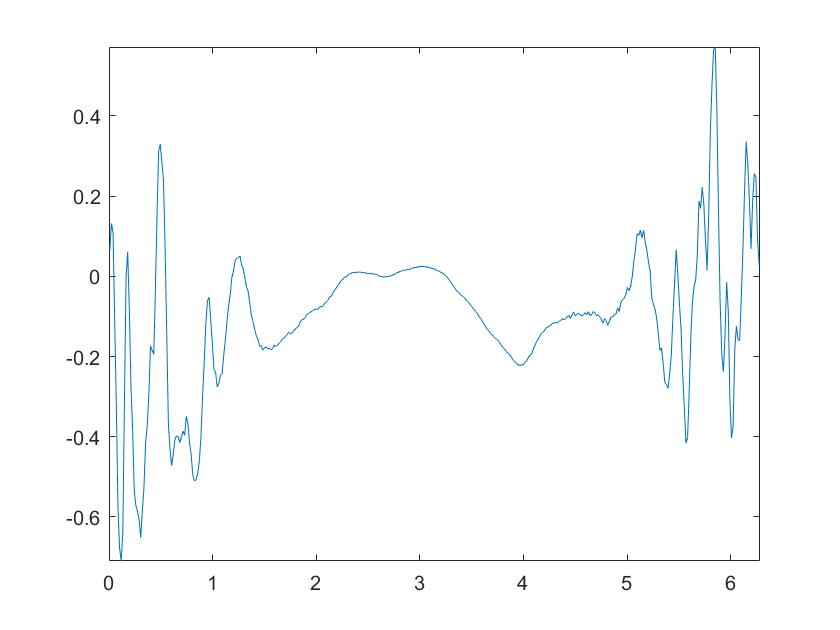}
\label{fig:awesome_image1}
\vspace{-25pt}\subcaption{}
\endminipage
\minipage{0.4\textwidth}
  \includegraphics[width=\linewidth]{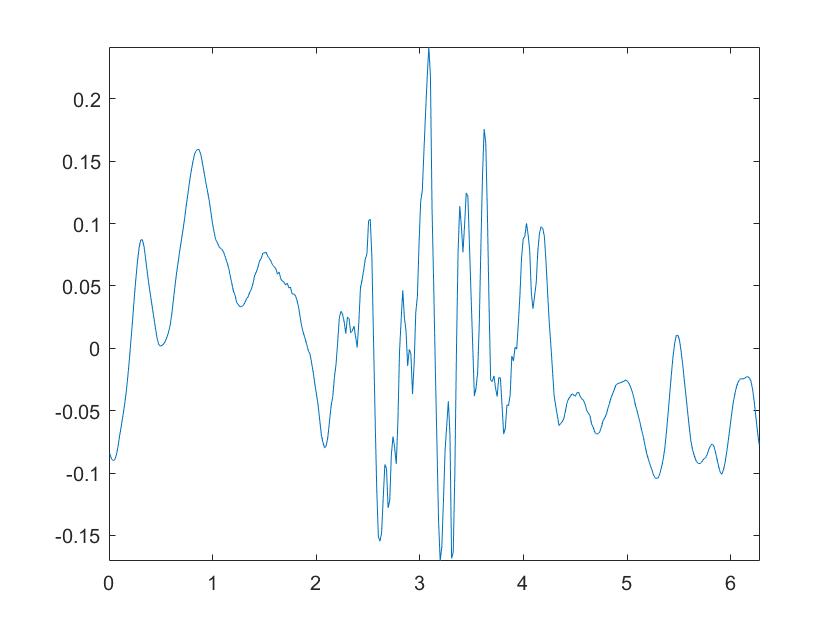}
\label{fig:awesome_image2}
\vspace{-25pt}\subcaption{}
\endminipage
\endminipage
\vspace{-10pt}
\caption{Random draws from nonstationary GFs on the unit circle; (a). $\kappa = 0.01; \tau = \exp \bigl(\cos(x) \bigr)$;  (b). $\kappa = \exp\bigl(\cos(x)\bigr), \tau = 100$.} 
\label{figure:tau&kappa}
\end{figure}

\section{GMRF Approximation with Graph Representations of SPDEs}\label{sec:graphdiscretizations}
In this section we study GMRF approximations of the Mat\'ern models introduced in Section \ref{sec:maternGFs}. Since the work \cite{lindgren2011explicit}, a burgeoning literature has been devoted to linking GFs and GMRFs, doing the modeling with the former and computations with the latter \cite{bakka2018spatial}. The main idea of \cite{lindgren2011explicit} is to introduce a stochastic weak formulation of the SPDE \eqref{eq:SPDE}: 
\begin{align*}
    \langle (\tau I -\Delta)^{\frac{s}{2}}u,\phi_i \rangle_{L^2}\overset{\mathcal{D}}{=}\langle \mathcal{W},\phi_i\rangle_{L^2},  \quad i=1,\ldots,M,
\end{align*}
where $\{\phi_i\}_{i=1}^M$ is a set of test functions and $\overset{\mathcal{D}}{=}$ denotes equal in distribution. 
Then one constructs a finite element (FEM) \emph{representation} of the solution 
$$u(x) = \sum_{k=1}^n w_{k} \varphi_{k}(x),$$
where $n$ is the number of vertices in the triangulation, $\{\varphi_{k}\}$ are interpolating piecewise linear hat functions and $\{w_{k}\}$ are Gaussian distributed weights. Importantly, these finite dimensional representations allow to obtain a GMRF precision matrix with computational cost $\OO(n)$. The convergence of the FEM representation to the GF has been studied in \cite{lindgren2011explicit} and in more generality in  \cite{bolin2019rational,bolin2020numerical, bolin2018weak}.

The FEM representation requires triangulation of the domain, possibly adding artificial nodes to obtain a suitable mesh, and in practice it is rarely implementable in dimension higher than 3. However for many applications e.g. in machine learning, interest lies in interpolating or classifying input data in high dimensional ambient space with moderate but unknown intrinsic dimension, making FEM representations of GFs impractical. Graph Laplacians, discussed next, provide a canonical way to construct GMRF approximations in the given point cloud.

\subsection{Graph Mat\'ern Models}\label{ssec:mainconstruction}
Let $\X = \{ x_1, \ldots, x_n\}$ be a given point cloud, over which we put a graph structure by considering a symmetric weight matrix $W\in \R^{n\times n}$ whose entries $W_{ij}\ge 0$ prescribe the closeness between points.  In applications including classification and regression, each $x_i$ will represent either a feature or an auxiliary point used to improve the accuracy of the GMRF approximations described in this subsection.  The graph structure encodes the geometry of the point cloud and can be exploited through the graph Laplacian. 

Several definitions of graph Laplacians co-exist in the literature. Defining $D := \text{diag}(d_1, \ldots, d_n)$ the degree matrix with $d_i := \sum_{j=1}^n W_{ij}$, three popular choices are unnormalized $\Delta_n^{\text {un}}: = D-W,$ symmetric $\Delta_n^{\text{sym}} := D^{-1/2} \Delta_n^{\text{un}} D^{-1/2}$ and random-walk $\Delta_n^{\text {rw}}:= D^{-1} \Delta_n^{\text{un}}$ graph Laplacians, see \cite{von2007tutorial}. To streamline the presentation, we use $\Delta_n\in \R^{n\times n}$ as placeholder for \emph{a} graph Laplacian with $n$ data points; its choice will be made explicit whenever it is relevant to the problem at hand.  

To gain some intuition, let us consider the unnormalized graph Laplacian, whose positive semi-definiteness is verified by the relation  
\begin{align}
	u_n^T(D-W) u_n = \frac{1}{2} \sum_{i=1}^n\sum_{j=1}^n W_{ij} |u_n(i)-u_n(j)|^2 \geq 0.  \label{eq:utLu}
\end{align}
Here $u_n = [u_n(1), \ldots, u_n(n)]^T \in \R^n$ is an arbitrary vector in $\R^n,$ interpreted as a function on $\X$ with the identification $u_n(x_i)  \equiv  u_n(i) .$
Note that $\Delta_n^{\text {un}} = D-W$  annihilates constant vectors (in agreement with the intuition that the Laplacian annihilates constant functions) and 0 is always an eigenvalue. For a fully connected graph, one can see that the eigenvalue 0 has multiplicity 1, with the constant vectors as its only eigenspace. If we consider $\mathcal{N}(0,\Delta_n^{-1})$ (with $\Delta_n^{-1}$ representing the Moore-Penrose inverse) as a degenerate Gaussian distribution in $\R^n$ with support on the orthogonal complement of the constant vectors, then \eqref{eq:utLu} is the negative log-density of this distribution (up to an additive constant), which suggests that functions  that take similar values on close nodes are favored, with closeness quantified by the weight matrix $W$. Moreover, it can be shown that the second eigenvector $\psi_n^{(2)}$ of $\Delta_n$ solves a relaxed graph cut problem \cite{von2007tutorial}, so that $\psi_n^{(2)}$ encodes crucial information about partition of the points $x_i$'s. Hence $\mathcal{N}(0,\Delta_n^{-1})$ naturally serves as a prior for clustering and classification \cite{bertozzi2018uncertainty}. Various choices of the weight matrix have been considered in the literature, including $\eps$-graphs and $k$-NN graphs,  which set $W_{ij}$ to be zero if $d(x_i,x_j)>\eps$ and if $x_i$ is not among the $k$-nearest neighbors of $x_j$ (or vice versa) respectively for some distance function $d$.  Both of them introduce sparsity in the weight matrix, which is inherited by the graph Laplacian. Under such circumstances, the graph Laplacian can be viewed as a sparse precision matrix, which gives rise to a GMRF.

\emph{ It is important to note that the preceding discussion makes no assumption on the points $x_i$ or how their closeness is defined.} Therefore, the graph-based viewpoint allows to \emph{generalize} the Mat\'ern model to  unstructured point clouds, and thus to settings of practical interest in statistics and machine learning where only similarity relationships between abstract features may be available. For instance, the points may represent books and their closeness may be based on a reader's perception of similarity between them.   However, an important example in which we will frame our theoretical investigations arises from making a \emph{manifold assumption}.
\begin{assumption}[Manifold Assumption]\label{manifoldassumption}
The points $x_i$ are independently sampled from the uniform distribution $\gamma$ on an $m$-dimensional smooth, connected, compact manifold $\M$ without boundary that is embedded in Euclidean space $\R^d$,  with bounded sectional curvature and Riemannian metric inherited from $\mathbb{R}^d$.   Assume further that $\M$ is normalized so that vol$(\M)$=1.
\end{assumption}

To emphasize the stronger structure imposed by the manifold assumption we denote the point cloud by  $\X \equiv \M_n=\{x_1,\ldots,x_n\} \subset \M.$ For many applications, the manifold assumption is an idealization of the fact that the point cloud has low dimensional structure despite living in a high dimensional ambient space, e.g., the MNIST dataset that we will study in Subsection \ref{sec:ML}.
From a theoretical viewpoint, Assumption \ref{manifoldassumption} allows us to establish  a precise link between graph Laplacians and their continuum counterparts, as we now describe heuristically. Define the weight matrix $W$ on $\M_n$   by  
\begin{align} \label{eq:weight}
	 W_{ij}& := \frac{2(m+2)}{n\nu_m h_n^{m+2}} \mathbf{1} \bigl\{|x_i-x_j|<h_n \bigr\},
\end{align}
where $| \cdot |$ is the Euclidean distance in $\R^d,$ $h_n$ is the graph connectivity and $\nu_m$ is the volume of the $m$-dimensional unit ball. Then the unnormalized graph Laplacian $\Delta_n^{\text {un}} $ is a discrete approximation of the Laplace-Beltrami operator $-\Delta$ on $\M$. 

Indeed, for a smooth function $f:\mathbb{R}^m \rightarrow \mathbb{R}$ we have by Taylor expansion  of $f$ around $X$
\begin{align*}
    \int_{B_{h_n}(X)} [f(Y)   - f(X)] dY &\approx \int_{B_{h_n}(X)} \nabla f(X)^T(Y-X) dY + \frac{1}{2} \int_{B_{h_n}(X)} (Y-X)^T \nabla^2 f(X) (Y-X) dY,
\end{align*}
where $B_{h_n}(X)$ is the Euclidean ball centered at $X$ with radius $h_n$. 
By symmetry, the first integral is zero and the second integral reduces to (after a change of variable $Z=Y-X$)
\begin{align*}
    \frac{1}{2} \sum_{i=1}^m \frac{\partial^2 f}{\partial X_i^2}(X)\int_{B_{h_n}(0)}  Z_i^2  dZ =  \frac{\nu_m h_n^{m+2}}{2(m+2)} \Delta f(X), 
\end{align*}
where $X_i$ and $Z_i$ represent the $i$-th coordinates of $X$ and $Z$. 
This gives 
\begin{align}
    -\Delta f(x_i) \approx  \frac{2(m+2)}{\nu_m h_n^{m+2}} \int_{B_{h_n}(x_i)} [f(x_i)-f(Y)]dY \approx \sum_{j=1}^n W_{ij} [f(x_i)-f(x_j)],  \label{eq:motivation} 
\end{align}
which is exactly the way $\Delta_n$ is defined. Since $\M$ is locally homeomorphic to $\mathbb{R}^m$ and the geodesic distance between any two points is well approximated by the Euclidean distance, the heuristic argument above can be formalized to show \emph{point-wise} convergence of $\Delta_n$ to $-\Delta$ in the manifold case. A rigorous result on \emph{spectral} convergence will be given in more generality in Section \ref{sec:theory}.

The previous discussion suggests to introduce the following discrete analog to the Gaussian measure \eqref{eq:stationaryGMcontinuum}
\begin{align*}
    \mathcal{N}(0,\mathcal{C}_n), \quad \mathcal{C}_n=\tau^{s-\frac{m}{2}}(\tau I_n+\Delta_n)^{-s}, 
\end{align*}
whose samples admit a Karhunen-Lo\`eve expansion
\begin{align} \label{eq:KLdiscreteStationary}
    u_n=\tau^{\frac{s}{2}-\frac{m}{4}}\sum_{i=1}^n \left[\tau+\lambda_n^{(i)}\right]^{-\frac{s}{2}} \xi^{(i)} \psi_n^{(i)},
\end{align}
where $\{\xi^{(i)}\}_{i=1}^n$ are independent standard normal random variables and $\bigl\{(\lambda_n^{(i)},\psi_n^{(i)}) \bigr\}_{i=1}^n$ are eigenpairs of $\Delta_n$. This will be our definition of the stationary graph Mat\'ern field. 

\begin{remark}
We note once again that the model 
\eqref{eq:KLdiscreteStationary} can be used in wide generality: it only presupposes that the practitioner is given a weight matrix associated with an abstract point cloud, and it only requires to specify a graph-Laplacian. We will show that \eqref{eq:KLdiscreteStationary} generalizes the stationary Mat\'ern model in the sense that if the point cloud is sampled from a manifold, the weights are defined through an appropriate $\eps$-graph, and an unnormalized graph-Laplacian is used, then the graph-based model approximates the Mat\'ern model on the manifold. Similar convergence results could be established with $k$-nearest neighbor graphs and other choices of graph Laplacian. Our numerical examples in Section \ref{sec:numericalexamples} will illustrate the application of graph-based Mat\'ern models both in manifold and abstract settings, and using a variety of graph Laplacians.  \qed
\end{remark}

\begin{remark} \label{rmk:density}
The above construction can be adapted when the points $x_i$ are distributed according to a Lipschitz density $q$ that is bounded below and above by positive constants. In this case, \eqref{eq:motivation} should take the form
\begin{align*}
	-\Delta f(x_i) \approx \sum_{j=1}^n W_{ij} [f(x_i)-f(x_j)] q(x_j)^{-1} \approx \frac12\sum_{j=1}^n W_{ij} [f(x_i)-f(x_j)] [q(x_i)^{-1}+q(x_j)^{-1}],
\end{align*}
where the last step follows from the Lipschitzness of $q$ and the fact that $q^{-1}$ is bounded away from 0 and is needed to ensure symmetry of the new weights.  Setting $f=q$ in \eqref{eq:motivation} we have 
\begin{align*}
	q(x_i) \approx \frac{1}{\nu_mh_n^m} \int_{B_{h_n}(x_i)} q(y) dy -h_n^2 \Delta q(x_i) \approx  \frac{1}{n\nu_mh_n^m} \sum_{j=1}^n \mathbf{1}\{|x_i-x_j|<h_n\}:=q_{h_n}(x_i)  
\end{align*}
where we have dropped $h_n^2\Delta q$ since it is of lower order. Hence the new weights should be adjusted as 
\begin{align*}
	W_{ij}= \frac{m+2}{n\nu_m h_n^{m+2}} \mathbf{1}\{|x_i-x_j|<h_n\} [q_{h_n}(x_i)^{-1}+q_{h_n}(x_j)^{-1}] .
\end{align*} \qed
\end{remark}

\subsection{Nonstationary Graph Mat\'ern Models}\label{ssec:modeling}
Now we are ready to construct nonstationary graph Mat\'ern fields that approximate the nonstationary  Mat\'ern field in Section \ref{sec:nonstationary}. In analogy with the previous subsection, the crucial step is to obtain a graph discretization of the operator \ $\L^{\tau,\kappa}=\tau I - \nabla \cdot (\kappa \nabla)$ with spatially varying $\tau$ and $\kappa$. Notice that we have 
\begin{align*}
    \nabla \cdot (\kappa\nabla f) = \sqrt{\kappa} [\Delta(\sqrt{\kappa}f)-f\Delta\sqrt{\kappa}]. 
\end{align*}
Applying \eqref{eq:motivation} to $\Delta(\sqrt{\kappa}f)$ and $f\Delta\sqrt{\kappa}$ gives
\begin{align*}
    -\nabla\cdot(\kappa\nabla f) \approx \int_{B_{h_n}(x)} \sqrt{\kappa(x)\kappa(y)} \, [f(x)-f(y)] \approx \sum_{i=1}^n W_{ij} \sqrt{\kappa(x_i)\kappa(x_j)} \, [f(x_i)-f(x_j)].
\end{align*}
Hence  $-\nabla\cdot(\kappa\nabla\cdot)$ can be approximated by $\Delta_n^{\kappa}=\tilde{D}-\tilde{W}$, where 
\begin{align}
    \tilde{W}_{ij}&= W_{ij} \sqrt{\kappa(x_i)\kappa(x_j)}=\frac{2(m+2)}{n\nu_m  h_n^{m+2}} \mathbf{1} \bigl\{|x_i-x_j|< h_n  \bigr\} \sqrt{\kappa(x_i)\kappa(x_j)} \label{eq:defweights} \, ,\\
    \tilde{D}_{ii}&=\sum_{j=1}^n \tilde{W}_{ij}.
\end{align}
Denoting $\tau_n=\operatorname{diag} \bigl(\tau(x_1),\ldots,\tau(x_n)\bigr)$ and   $\kappa_n=\operatorname{diag}\bigl(\kappa(x_1),\ldots,\kappa(x_n)\bigr)$, we define ---similarly as in Subsection \ref{sec:nonstationary}--- the nonstationary graph Mat\'ern field through the Karhunen-Lo\'eve expansion 
\begin{align}\label{eq:KLdiscreteNonstationary}
    u_n:= \tau_n^{\frac{s}{2}-\frac{m}{4}}\kappa_n^{\frac{m}{4}}\sum_{i=1}^n \left[\lambda_n^{(i)}\right]^{-\frac{s}{2}} \xi^{(i)} \psi_n^{(i)},
\end{align}
where $\{\xi^{(i)}\}_{i=1}^n$ are independent standard normal random variables and $\bigl\{(\lambda_n^{(i)},\psi_n^{(i)}) \bigr\}_{i=1}^n$ are eigenpairs of $L^{\tau,\kappa}_n:=\tau_n+\Delta_n^{\kappa}$. Equation \eqref{eq:KLdiscreteNonstationary} is a natural finite dimensional approximation of \eqref{eq:KLcontinuumNonstationary} and one should expect that spectral convergence of $\Delta_n$ towards $-\Delta$ will translate into convergence of \eqref{eq:KLdiscreteNonstationary} towards \eqref{eq:KLcontinuumNonstationary} in the large $n$ limit. This will be rigorously shown in Section \ref{sec:theory}.

In the covariance operator view, $u_n$ follows a Gaussian distribution  $\Nc(0,C_ n^{\tau,\kappa})$ with
\begin{equation}\label{eq:covarianceanysotropicdiscrete}
  C_n^{\tau,\kappa} := \tau_n^{\frac{s}{2}-\frac{m}{4}}\kappa_n^{\frac{m}{4}} [L^{\tau,\kappa}_n]^{-s}\kappa_n^{\frac{m}{4}}\tau_n^{\frac{s}{2}-\frac{m}{4}}.
\end{equation}  
Therefore samples can be generated by solving 
\begin{align}
    [L_n^{\tau,\kappa}]^{\frac{s}{2}}u_n=\xi_n,\quad \quad \xi_n\sim\mathcal{N}(0,I_n) \label{eq:generate sample}
\end{align}
and then multiplying with the diagonal matrix $\tau_n^{s/2-m/4}\kappa_n^{m/4}$. For $0<s<2$, \eqref{eq:generate sample} can be solved with a sparse approximation as in \cite{harizanov2018optimal,bolin2019rational} and for $s\geq 2$ an iterative scheme can be employed. We remark that \eqref{eq:generate sample} can be solved exactly by performing a spectral decomposition of $L_n^{\tau,\kappa}$, which is computationally more expensive and not recommended for large $n$'s.  

\begin{remark} \label{rmk:graphnormalization}
As discussed in Remark \ref{rmk:tau&kappa}, both $\tau_n$ and $\kappa_n$ control the local length scale.  For many applications e.g. in machine learning, we shall focus on the modeling choice with $\tau_n$ only,  because the operator $\nabla \cdot (\kappa \nabla) $ is less motivated for general $x_i$'s  that do not come from a manifold. In such cases, one can construct a nonstationary Mat\'ern field similarly as above, by using a graph Laplacian built with the $x_i$'s, e.g. with a $k$-NN graph. Indeed, the only key step is to normalize properly the marginal variances, which are largely determined by the growth of the spectrum as in Remark \ref{rmk:normalization}. Hence one can find an integer $m$ so that the first several $\lambda_n^{(i)}$'s grow roughly as $i^{\frac{2}{m}}$ and use $m$ as an effective dimension of the problem for normalization. Moreover, both the $k$-NN and $\eps$-graphs result in sparsity in $\Delta_n$, and numerical linear algebra techniques for sparse systems can be employed to attain speed-up.  \qed
\end{remark}

\subsection{A Simulation Study}

In this subsection we perform a simulation study on the unit sphere to demonstrate the graph approximation of Mat\'ern fields and its sparsity. Let $\mathcal{M}$ be the two-dimensional unit sphere embedded in $\mathbb{R}^3$ and formally consider  the Mat\'ern model specified by the SPDE 
\begin{align*}
    (I-\Delta)^{-\frac{s}{2}} u(x) =  \mathcal{W}(x),
\end{align*}
where $\Delta$ is the Laplace-Beltrami operator on $\mathcal{M}$ and $\mathcal{W}$ is spatial white noise with unit variance. More precisely, we consider the Mat\'ern field defined  by  the  Karhunen-Lo\`eve expansion 

\begin{align*}
    u(x)=\sum_{i=1}^{\infty}\left[1+\lambda^{(i)}\right]^{-\frac{s}{2}}\xi^{(i)} \psi^{(i)}, \quad \xi^{(i)}\overset{i.i.d.}{\sim} \mathcal{N}(0,1),
\end{align*}
where $\{(\lambda^{(i)},\psi^{(i)})\}_{i=1}^{\infty}$ are the eigenpairs of $-\Delta$. It is known that the eigenvalues are $\ell(\ell-1)$ with multiplicity $2\ell-1$ for $\ell=1,2,\ldots,$ and the eigenfunctions are the spherical harmonics. The covariance function associated with this field is 
\begin{align}
    c(x,y)=\sum_{i=1}^{\infty}\left[1+\lambda^{(i)}\right]^{-s}\psi^{(i)}(x)\psi^{(i)}(y). \label{eq:covariance}
\end{align}
We shall investigate the approximation of this covariance function by the covariance function of a graph Mat\'ern field. In Subsection \ref{sec:unstruct} we consider the case where only unstructured samples from the sphere are available to demonstrate the generality of the graph-based method and in Subsection \ref{sec:struct} we restrict ourselves to triangulations of the sphere for comparison with the FEM-based approximation. 
\subsubsection{Unstructured Grids} \label{sec:unstruct}
In this subsection we consider ``pseudo-unstructured'' point clouds generated as follows. The idea is to parametize points on the sphere in polar coordinates $(\theta,\phi)\in [0,\pi]\times [0,2\pi]$. so that the uniform distribution on the sphere can be generated with the formula
\begin{align*}
    \theta=\arccos(1-2U), \quad \phi=2\pi V,
\end{align*}
where $U,V$ are independent unif$(0,1)$ random variables. Now instead of generating $n$ i.i.d. pairs of $(U,V)$, we will partition the domain $[0,\pi]\times [0,2\pi]$ uniformly into subgirds of size $\frac{\pi}{M}$ by $\frac{2\pi}{M}$ for an integer $M$ and then pick one point from each subgrid randomly and uniformly. The reason is that due to the rotational symmetry of the spherical harmonics, the computed graph eigenfunctions may be out-of-phase versions of the true eigenfunctions and hence we need some structure from the point cloud in order to make them well aligned. Therefore the generated point cloud is only close to being unstructured.

Let $\{x_i\}_{i=1}^{n=M^2}$ denote the generated point cloud. We construct an $\eps$-graph over the $x_i$'s by setting the weights as in \eqref{eq:weight}. Precisely, we define
\begin{align}
    W_{ij}=\frac{2(m+2)\text{vol}(\M)}{n\nu_mh^{m+2}}\mathbf{1}\{|x_i-x_j|<h_n\}=\frac{32}{nh_n^4}\mathbf{1}\{|x_i-x_j|<h_n\},\label{eq:weight sphere}
\end{align}
where $m=2$ in this case and the additional factor vol$(\M)$ is needed to account for the fact that vol$(\M)\neq 1$.  
Let $\Delta_n=D-W$ be the unnormalized graph Laplacian. Then \eqref{eq:covariance} is approximated by 
\begin{align}
    c_n(x_j,x_k)=\sum_{i=1}^n \left[1+\lambda^{(i)}_{n}\right]^{-s} \psi_{n}^{(i)}(x_j)\psi_{n}^{(i)}(x_k), \label{eq:approx covariance}
\end{align}
where the $\psi^{(i)}_{n}$'s are suitably normalized eigenfunctions of $\Delta_n$. 

\begin{figure}[!htb]
\minipage{1\textwidth}
\centering
\minipage{0.249\textwidth}
  \includegraphics[width=\linewidth]{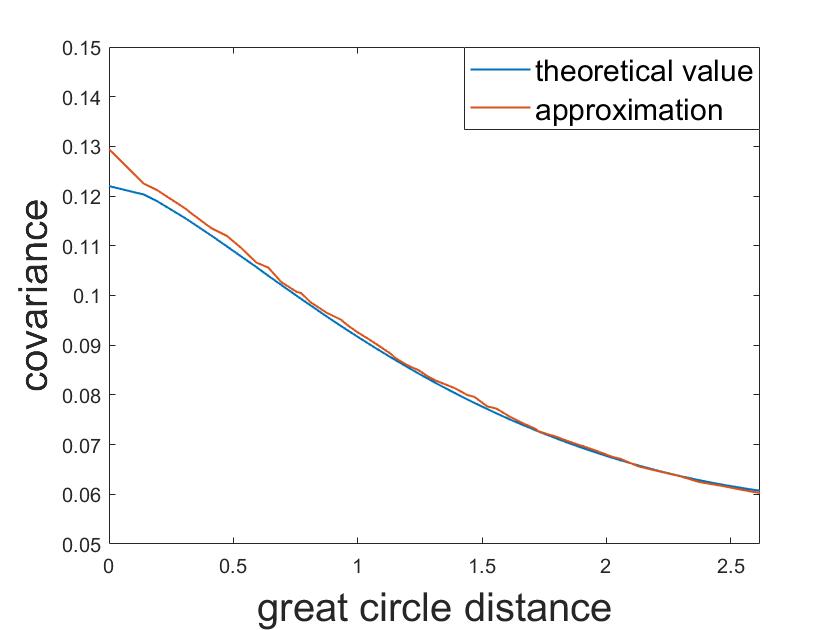}
\label{fig:covariance_full}
\endminipage
\minipage{0.249\textwidth}
  \includegraphics[width=\linewidth]{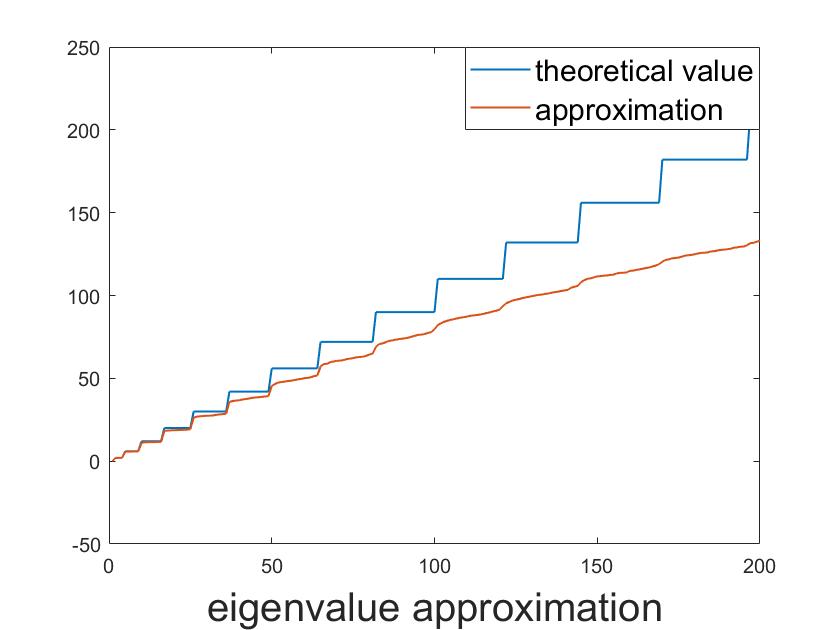}
\label{fig:spec}
\endminipage
\minipage{0.249\textwidth}
  \includegraphics[width=\linewidth]{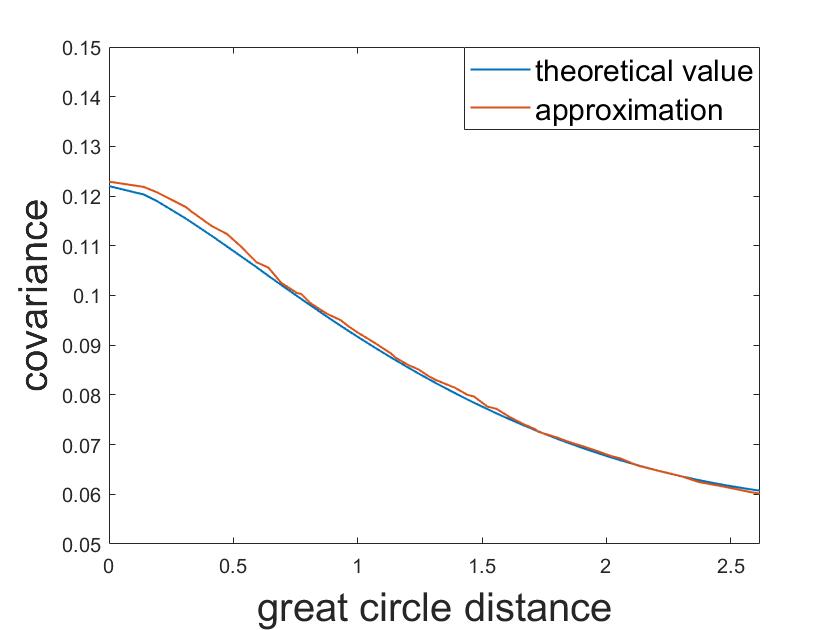}
\label{fig:covariance}
\endminipage
\minipage{0.249\textwidth}
  \includegraphics[width=\linewidth]{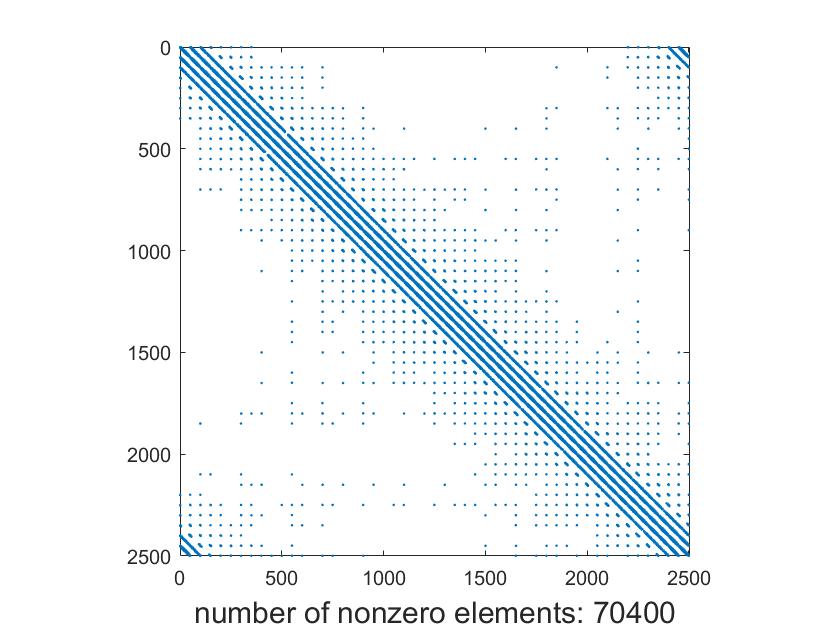}
\label{fig:sparsity}
\endminipage
\endminipage\hfill
\minipage{1\textwidth}
\centering
\minipage{0.249\textwidth}
  \includegraphics[width=\linewidth]{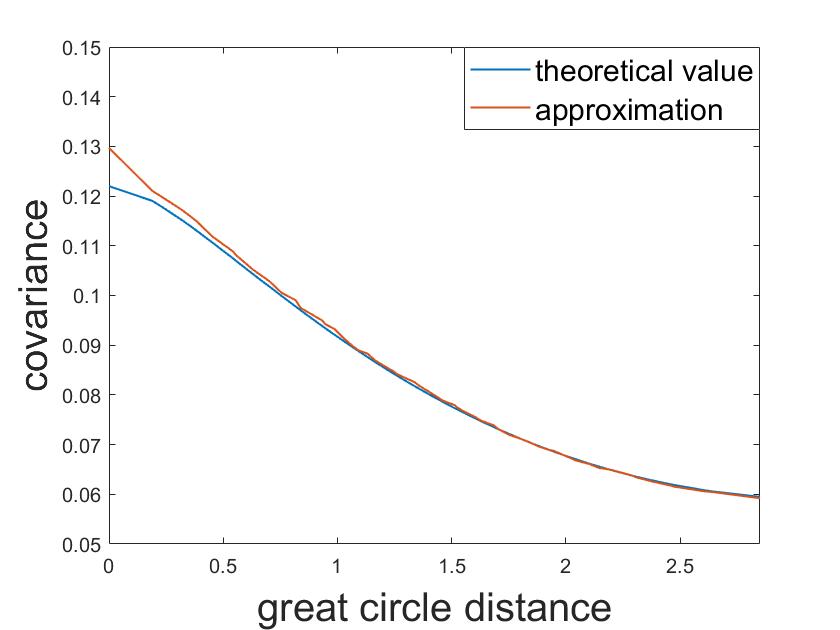}
\label{fig:covariance_full}
\endminipage
\minipage{0.249\textwidth}
  \includegraphics[width=\linewidth]{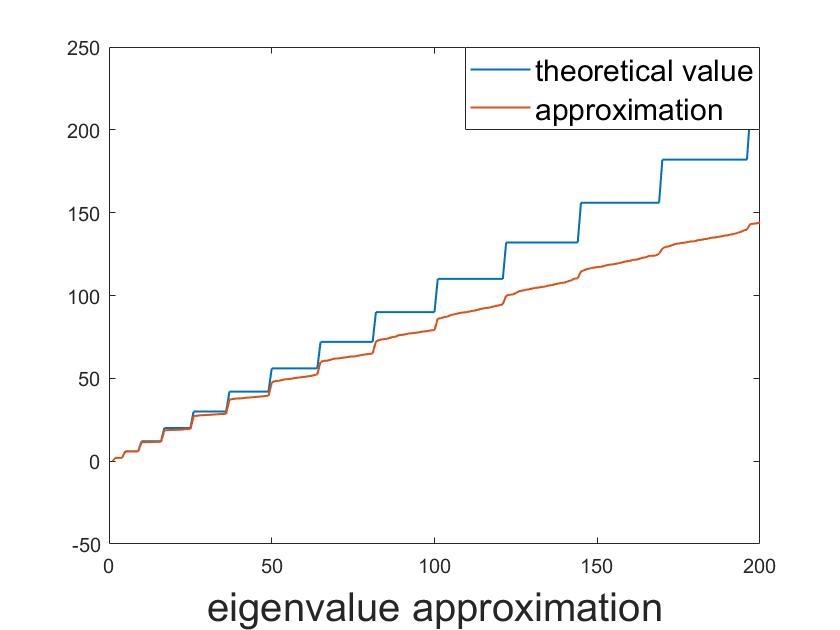}
\label{fig:spec}
\endminipage
\minipage{0.249\textwidth}
  \includegraphics[width=\linewidth]{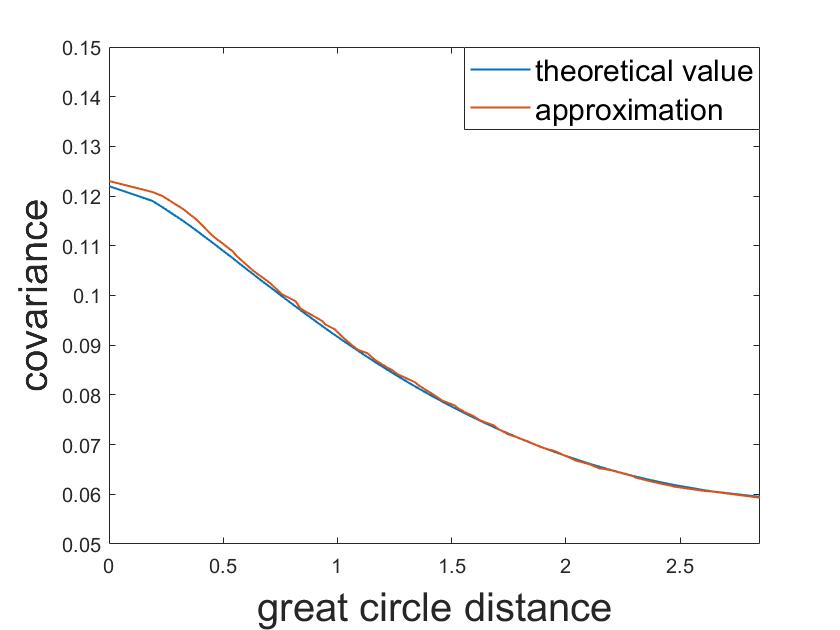}
\label{fig:covariance}
\endminipage
\minipage{0.249\textwidth}
  \includegraphics[width=\linewidth]{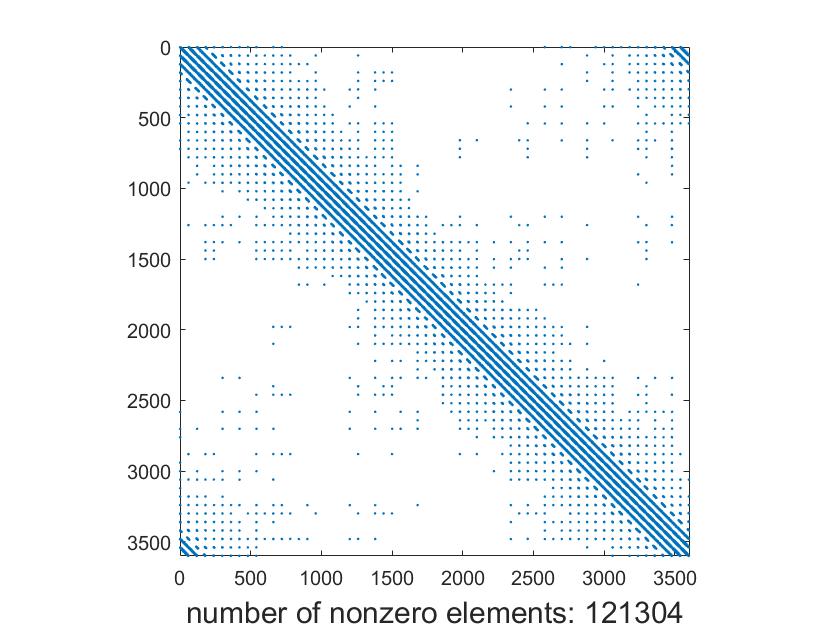}
\label{fig:sparsity}
\endminipage
\endminipage\hfill
\minipage{1\textwidth}
\centering
\minipage{0.249\textwidth}
  \includegraphics[width=\linewidth]{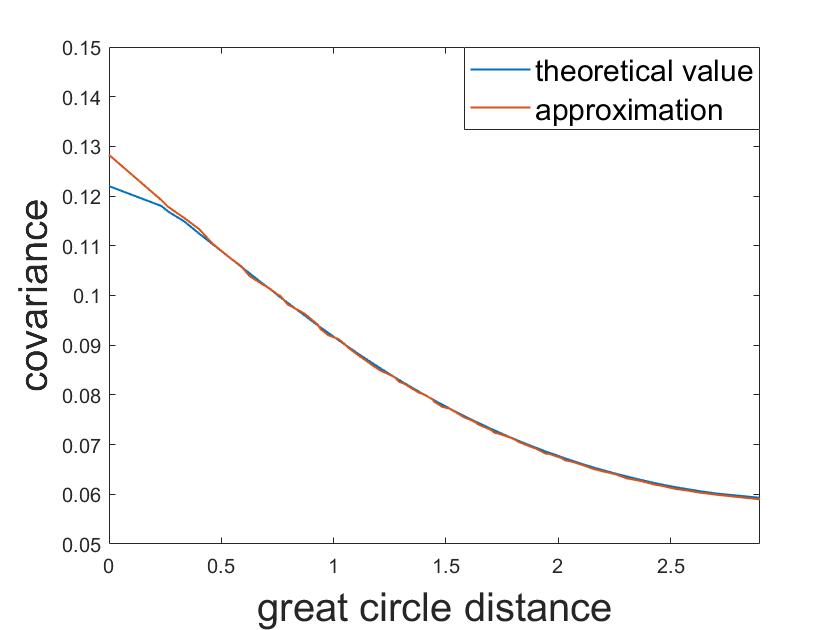}
\label{fig:covariance_full}
\endminipage
\minipage{0.249\textwidth}
  \includegraphics[width=\linewidth]{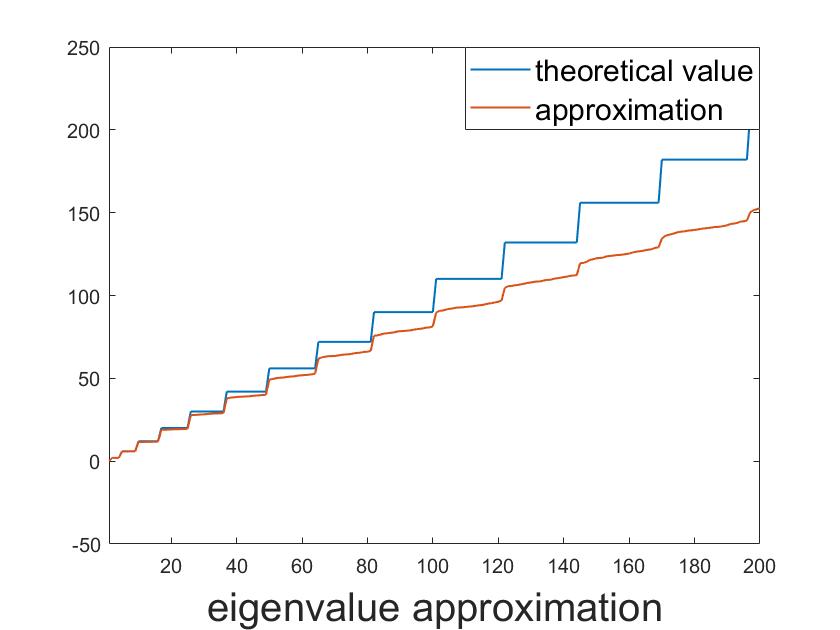}
\label{fig:spec}
\endminipage
\minipage{0.249\textwidth}
  \includegraphics[width=\linewidth]{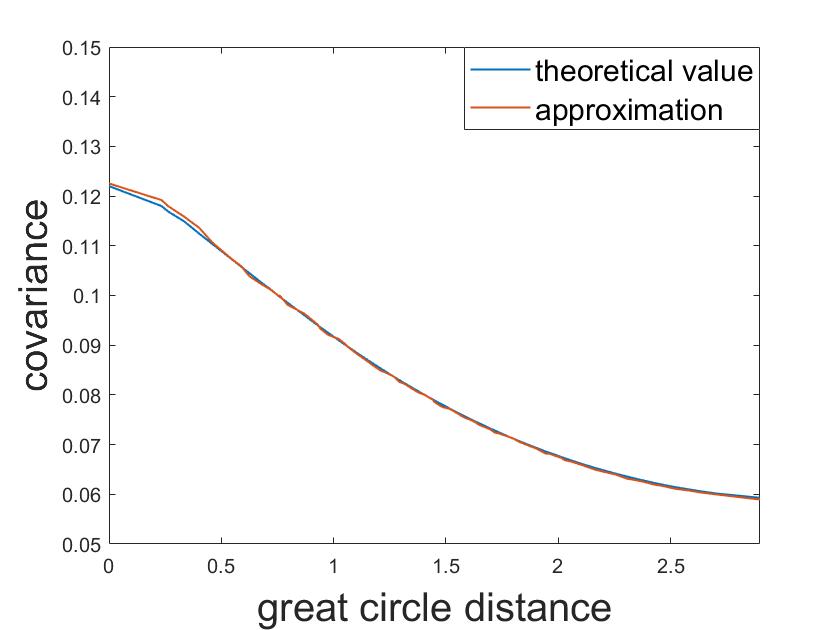}
\label{fig:covariance}
\endminipage
\minipage{0.249\textwidth}
  \includegraphics[width=\linewidth]{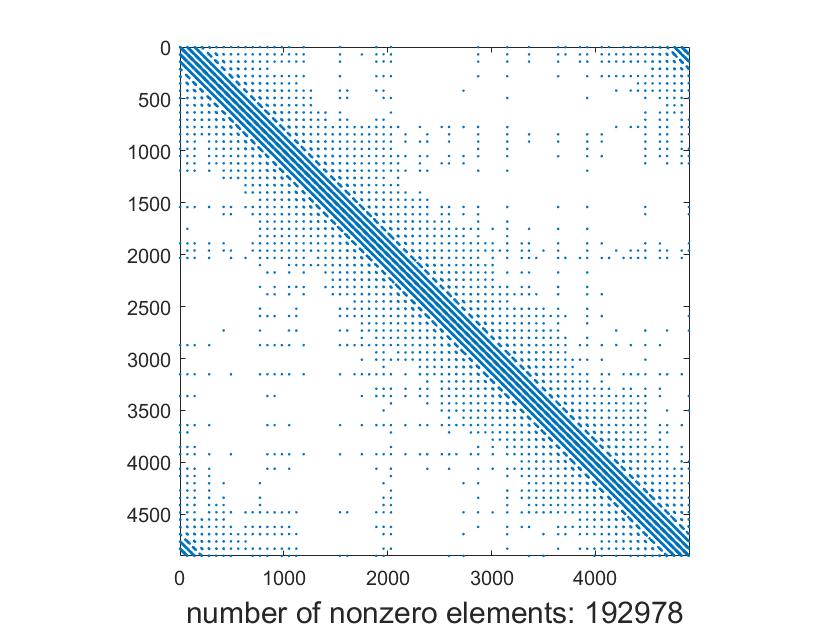}
\label{fig:sparsity}
\endminipage
\endminipage\hfill
\caption{The three rows represent simulations for $n=2500$, $3600$, $4900$ respectively. Within each row, the plots represent (1) Covariances: theoretical (blue) vs. approximation (red).  (2) Spectra of $-\Delta$ (blue) and $\Delta_n$ (red).  (3). Covariances (truncated): theoretical (blue) vs. approximation (red). (4). Sparsity pattern of $\Delta_n$.}
\label{fig:simulation}
\end{figure}

To demonstrate the approximation, we pick $M$ points $x_1,\ldots,x_{M}$ with similar longitude ranging from the north pole to the south pole.  That way, the distance between $x_1$ and $x_k$ is increasing, with $x_{M}$ being the farthest and we then compare the values of $c_n(x_1,x_k)$ and $c(x_1,x_k)$ for $k=1,\ldots,M.$ We shall focus on the $s=2$ case and set the connectivity as $h_n=1.5\times n^{-1/4}$, which is motivated by \eqref{eq:hnscaling} below. Figure \ref{fig:simulation} shows three simulations with $M=50,60,70$ respectively in each row.  The three plots in the first column show that we get reasonable approximations except for the first entry. The reason lies in the poor spectral approximation after certain threshold (Theorems \ref{thm:evalRate} and \ref{thm:efunRate}) as demonstrated in the second column. The spectrum of $\Delta_n$ becomes almost flat after some threshold and hence the tails of \eqref{eq:approx covariance} have a nonnegligible  contribution that worsens the approximation. Such effect is most prominent for $c_n(x_1,x_1)$, since for $j\neq k$ the vectors $\{\psi_n^{(i)}(x_j)\}_{i=1}^n$ and $\{\psi_n^{(i)}(x_k)\}_{i=1}^n$ are less correlated (and in fact orthogonal because they are the rows of the eigenvector matrix from singular value decomposition), so that cancellations reduce the contribution of their tails in \eqref{eq:approx covariance} even if the spectrum gets flat. In the third column of Figure \ref{fig:simulation} we compute the truncated version of $\eqref{eq:approx covariance}$ by keeping only the first $\sqrt{n}$ terms, and we see that $c_n(x_1,x_1)$ is improved substantially with the rest being almost the same as before. The truncation level $\sqrt{n}$ is motivated by Theorem \ref{thm:evalRate} so that the error $h_n\sqrt{\lambda^{(\sqrt{n})}}=O(1)$ in this case. Finally the last column of Figure \ref{fig:simulation} shows the sparsity patterns of the $\Delta_n$'s, which have decreasing percentages of nonzero entries $1.13\%$, $0.94\%$ and $0.8\%$ as $n$ increases. The sparsity indeed leads to computational speed-up as for instance the Matlab built-in function \texttt{chol} takes in our machine 0.1713s to factorize $(I_n+\Delta_n)^s$ but 0.6082s to factorize a random positive semi-definite  $n$ by $n$ dense matrix when $n=4900$. 

\subsubsection{Structured Grids}\label{sec:struct}
The unstructured point clouds in the previous subsection do not necessarily form triangulations of the sphere and FEM-based approximation would require additional nodes. Therefore to compare the two approximations, we assume instead to be given a triangulation as shown in the first column of Figure \ref{fig:fem vs graph}. The three triangulations are generated from the R-INLA function \texttt{inla.mesh.create(globe=k)} with $k=17,19,22$, which consist of $n=2892$, $3612$, $4842$ points respectively. Graph-based approximations are constructed as in \eqref{eq:weight sphere} with $h_n=2\times n^{-1/4}$ and FEM-based approximations are computed using the R-INLA package. The second column of Figure \ref{fig:fem vs graph} shows that the FEM-based approximations are almost indistinguishable from the truth and the truncated graph-based approximations (as in Subsection \ref{sec:unstruct} with truncation level $\sqrt{n}$) are also reasonably accurate. The last two columns of Figures \ref{fig:fem vs graph} demonstrate the corresponding sparsity patterns of the precision matrices, i.e. the matrices $(I_n+\Delta_n)^{s}$ for the graph-based case. The comparisons suggest that FEM-based approximations outperform the graph-based ones, which is not unexpected since the graph-based approach does not fully exploit the structure of the triangulation. Therefore the graph-based approximation is especially valuable when a triangulation of the domain is unavailable and difficult to obtain, as will be demonstrated in our numerical examples in Section \ref{sec:numericalexamples}. 
\begin{figure}[!htb]
\minipage{1\textwidth}
\centering
\minipage{0.249\textwidth}
  \includegraphics[width=\linewidth]{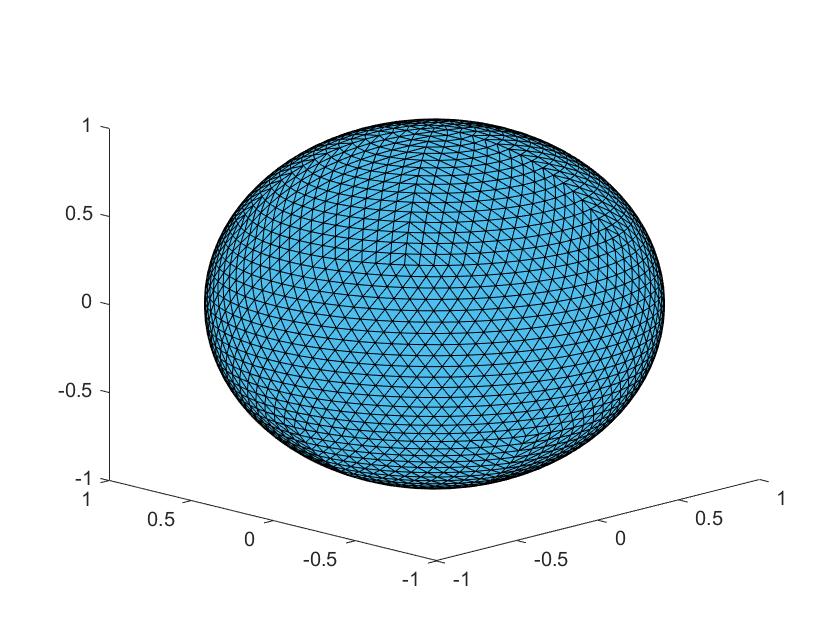}
  \vspace{-0.5cm}
\label{fig:triangulation}
\endminipage
\minipage{0.249\textwidth}
  \includegraphics[width=\linewidth]{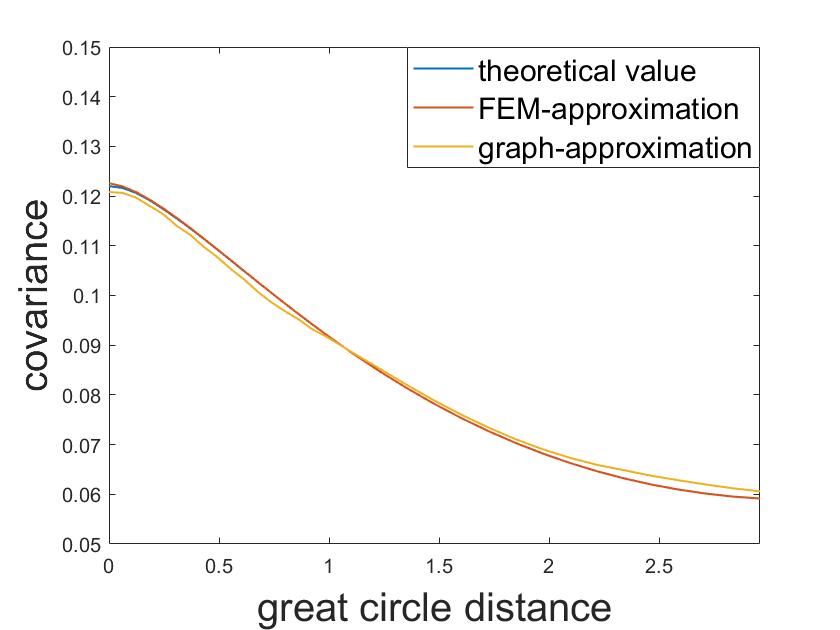}
\label{fig:covariances}
\endminipage
\minipage{0.249\textwidth}
  \includegraphics[width=\linewidth]{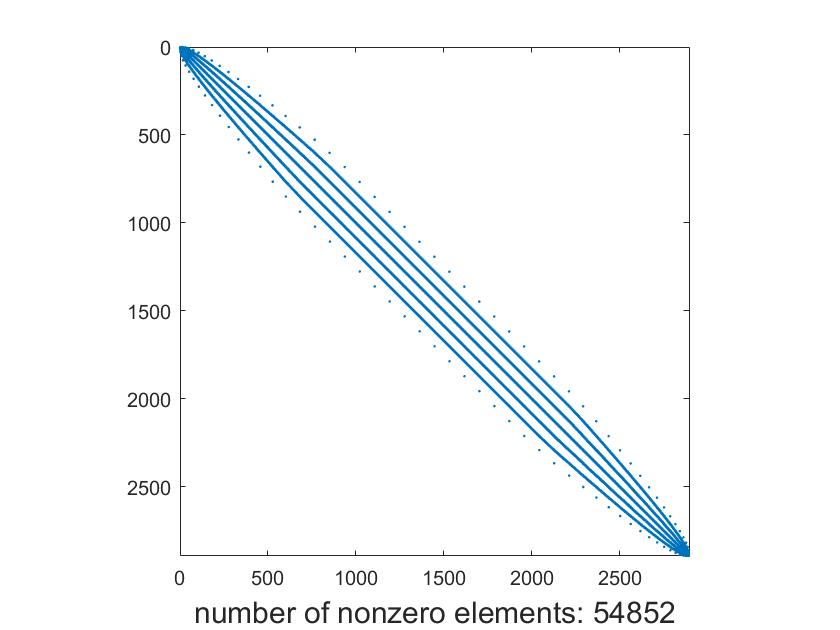}
\label{fig:spy_fem}
\endminipage
\minipage{0.249\textwidth}
  \includegraphics[width=\linewidth]{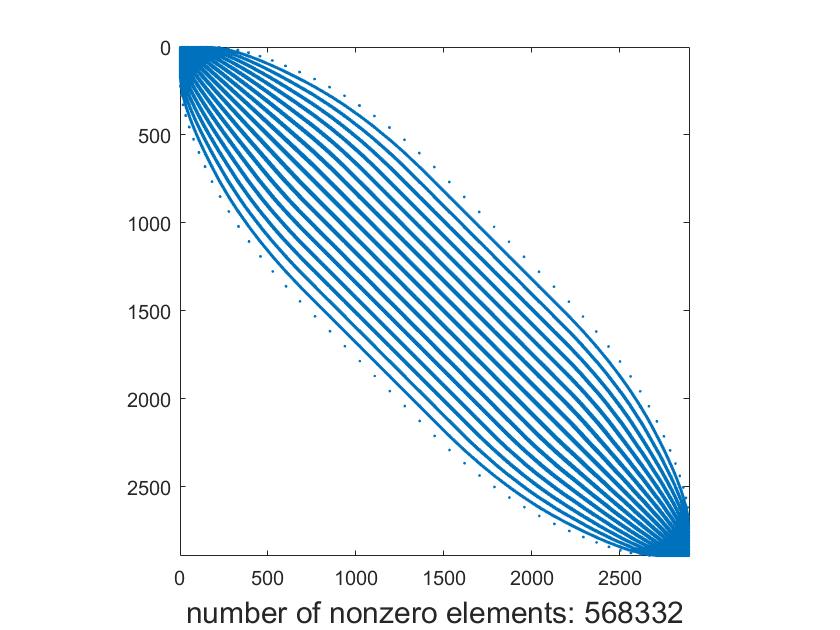}
\label{fig:spy_graph}
\endminipage
\endminipage\hfill
\minipage{1\textwidth}
\centering
\minipage{0.249\textwidth}
  \includegraphics[width=\linewidth]{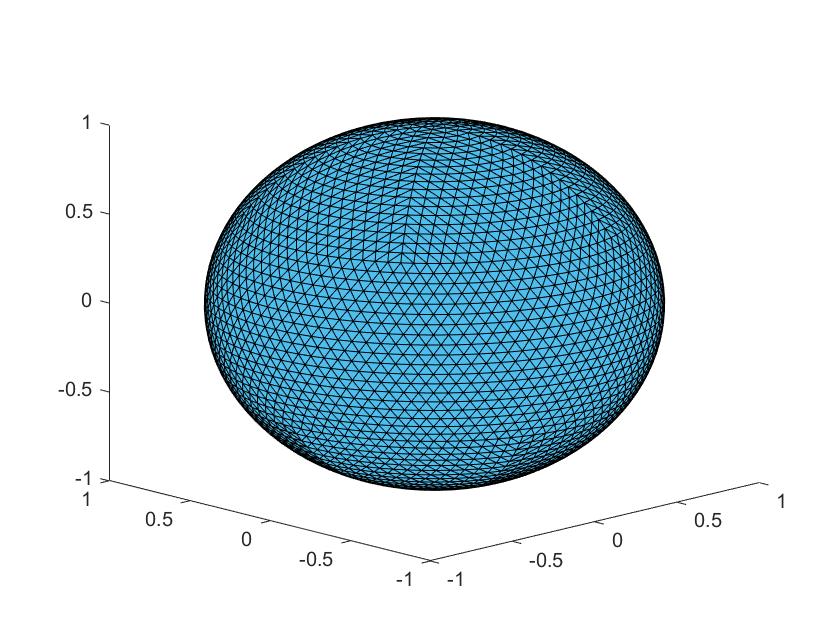}
  \vspace{-0.5cm}
\label{fig:triangulation}
\endminipage
\minipage{0.249\textwidth}
  \includegraphics[width=\linewidth]{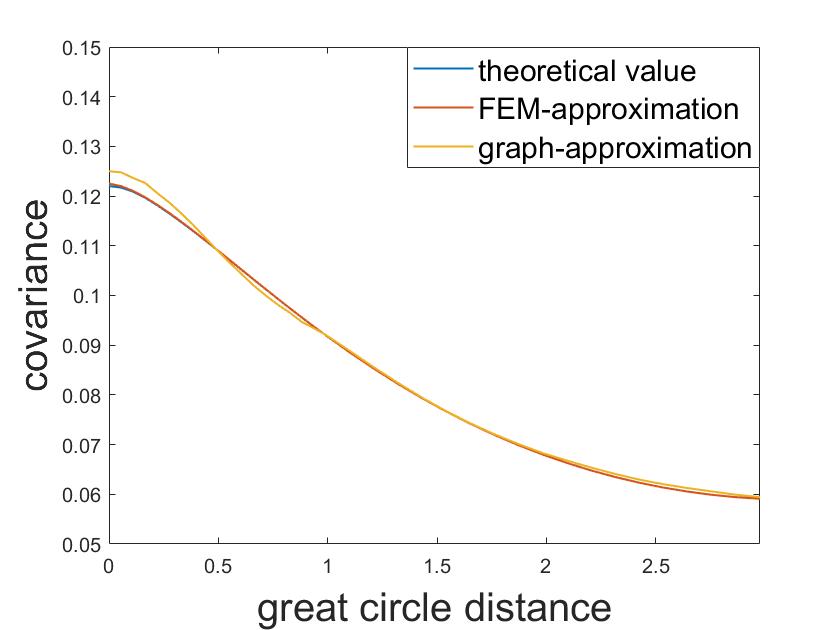}
\label{fig:covariances}
\endminipage
\minipage{0.249\textwidth}
  \includegraphics[width=\linewidth]{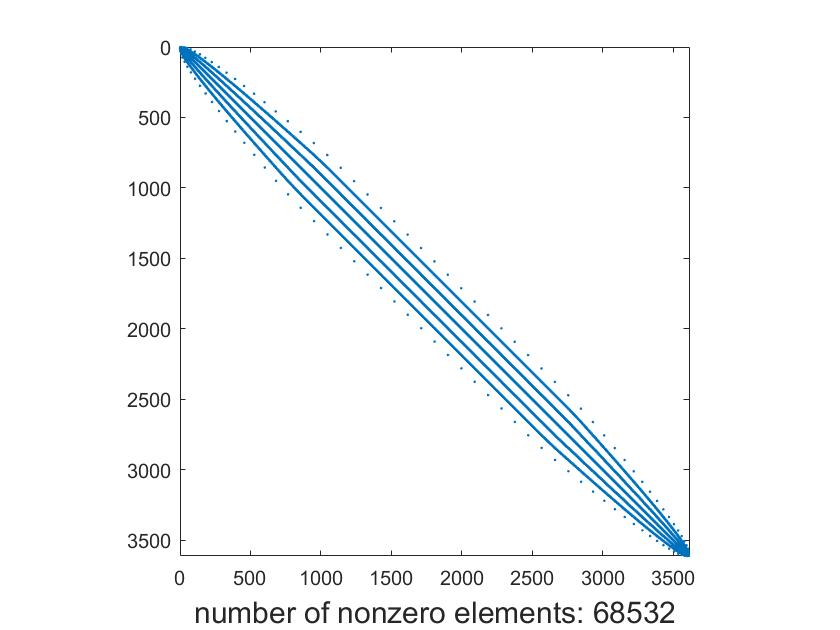}
\label{fig:spy_fem}
\endminipage
\minipage{0.249\textwidth}
  \includegraphics[width=\linewidth]{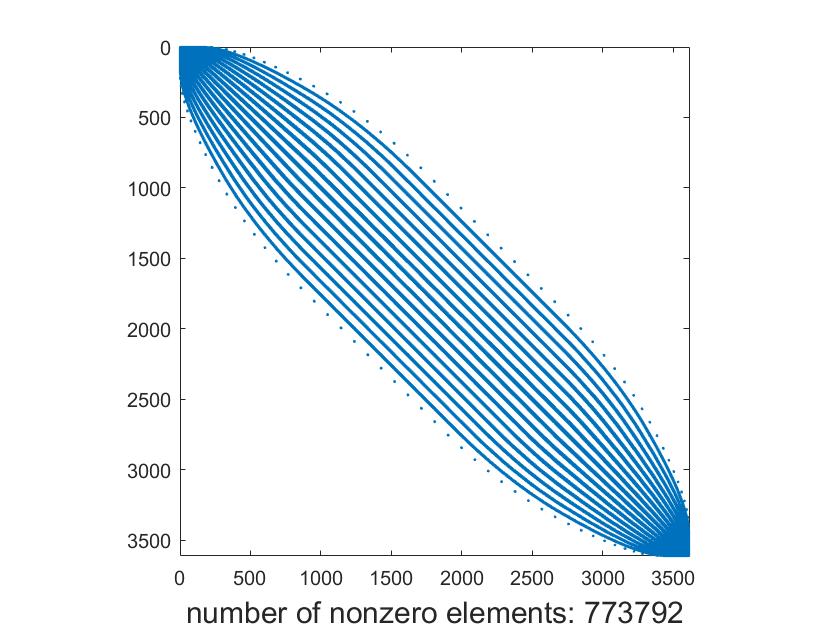}
\label{fig:spy_graph}
\endminipage
\endminipage\hfill
\minipage{1\textwidth}
\centering
\minipage{0.249\textwidth}
  \includegraphics[width=\linewidth]{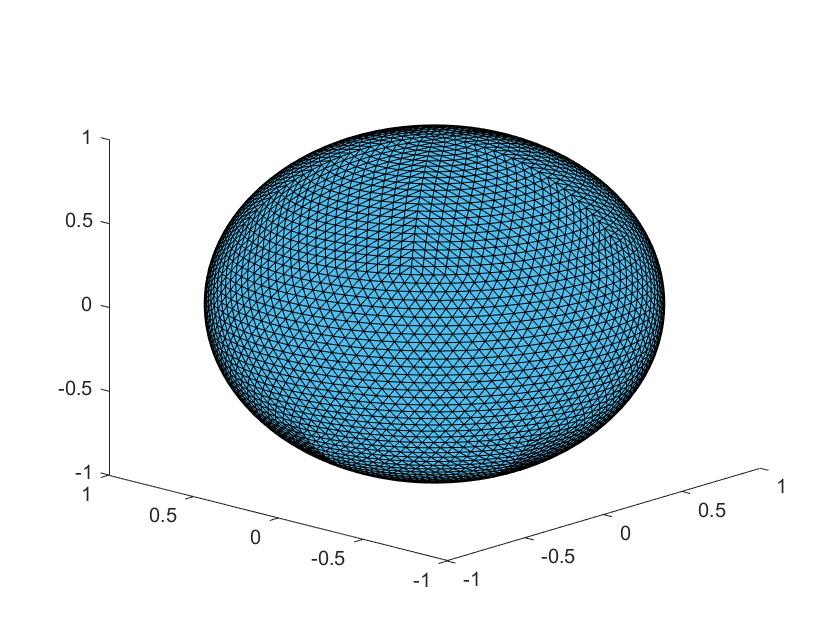}
  \vspace{-0.5cm}
\label{fig:triangulation}
\endminipage
\minipage{0.249\textwidth}
  \includegraphics[width=\linewidth]{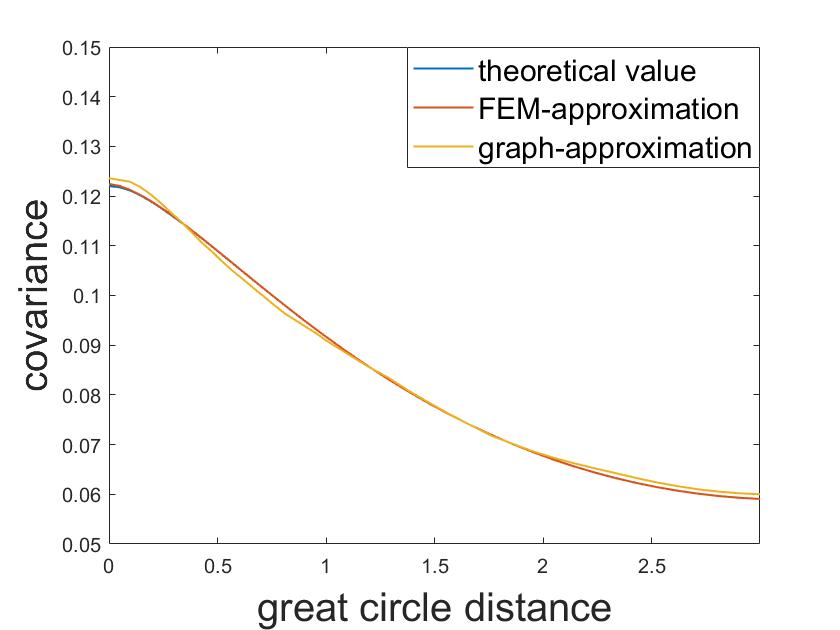}
\label{fig:covariances}
\endminipage
\minipage{0.249\textwidth}
  \includegraphics[width=\linewidth]{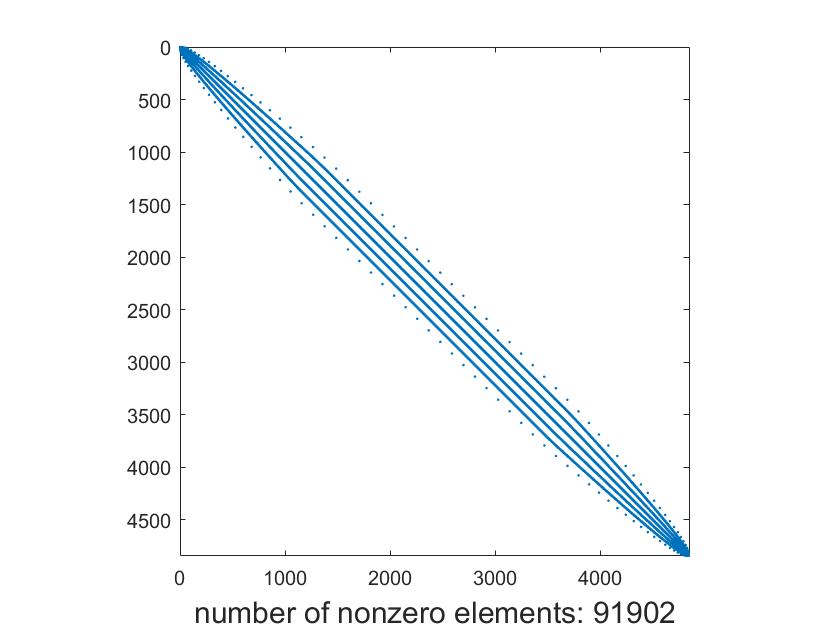}
\label{fig:spy_fem}
\endminipage
\minipage{0.249\textwidth}
  \includegraphics[width=\linewidth]{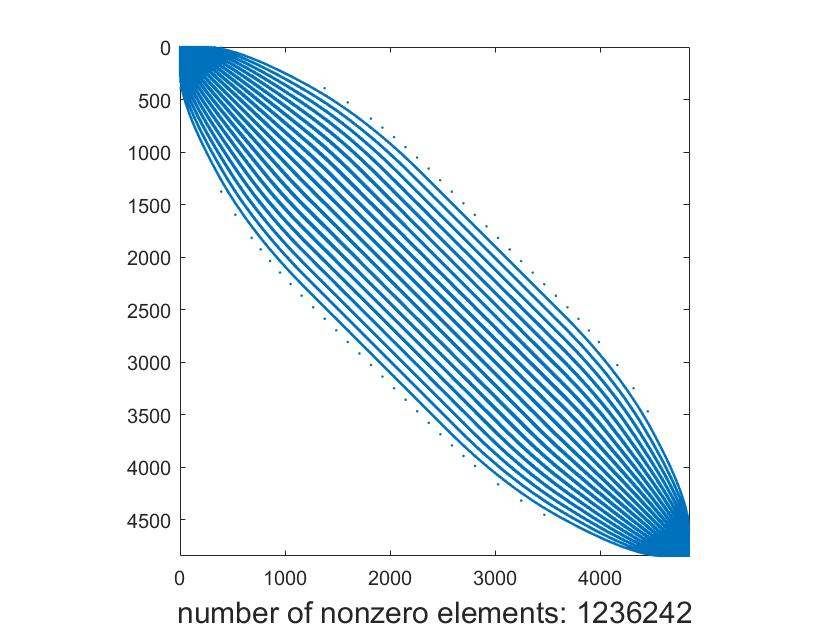}
\label{fig:spy_graph}
\endminipage
\endminipage\hfill
\caption{The three rows represent simulations for $n=2892$, $3612$, $4842$ respectively. Within each row, the plots represent (1) Visualization of the triangulation.  (2) Covariances: theoretical (blue), FEM-based (red), and graph-based (yellow).     (3). Sparsity pattern of FEM-based precision. (4). Sparsity pattern of graph-based precision.}
\label{fig:fem vs graph}
\end{figure}


\section{Convergence of Graph Representations of Mat\'ern Models}\label{sec:theory} 
In this section we study the convergence of graph representations of GFs under  the manifold  Assumption \ref{manifoldassumption}. The analysis will generalize existing literature to cover the  nonstationary models introduced in Subsection \ref{ssec:modeling}, obtaining new rates of convergence.

\subsection{Setup and Main Result}
Recall from \eqref{eq:KLcontinuumNonstationary} and \eqref{eq:KLdiscreteNonstationary} that draws from the continuum and graph Mat\'ern fields are defined by the series
\begin{align}
u &= \tau^{\frac{s}{2}-\frac{m}{4}}\kappa^{\frac{m}{2}} \sum_{i = 1}^ \infty \left[\lambda^{(i)}\right]^{\frac{s}{2}} \,\xi^{(i)} \, \psi^{(i)}, \label{eq:KLu}\\
u_n &=\tau_n^{\frac{s}{2}-\frac{m}{4}}\kappa_n^{\frac{m}{2}} \sum_{i = 1}^n \left[\lambda_n^{(i)}\right]^{\frac{s}{2}} \, \xi^{(i)} \, \psi_n^{(i)}. \label{eq:KLun}
\end{align} 
We seek to establish convergence of $u_n$ towards $u$.  
But notice that $u_n$ is only defined on the point cloud $\M_n$, while $u$ is defined on $\M$. Therefore a natural route is to introduce an interpolation scheme, using ideas from optimal transport theory. 

Recall the manifold assumption \ref{manifoldassumption} that $\{x_i\}_{i=1}^{\infty}$ is a sequence of independent   samples from the uniform distribution $\gamma$ on $\M$ and denote bu $\gamma_n=\frac{1}{n}\sum_{i=1}^n \delta_{x_i}$  the empirical measure of $\{x_i\}_{i=1}^n$. It should be intuitively clear that for the graph Mat\'ern fields to approximate well their continuum counterparts, the point cloud $\{x_i\}_{i=1}^n$ needs to first approximate $\M$ well. The following result (\cite{trillos2019error}[Theorem 2] together with Borel-Cantelli) measures such approximation quantitatively. 
\begin{proposition} \label{prop:transmap}
There is a constant $C$ such that, with probability one, there exists a sequence of transport maps $T_n:\M\rightarrow \M_n$ so that $\gamma_n=T_{n\sharp}\gamma$ and  
\begin{align}
    \underset{n \rightarrow \infty}{\operatorname{lim \,sup}} \, \frac{n^{1/m} \operatorname{sup}_{x\in\M}d_{\M} \bigl(x,T_n(x) \bigr)}{(\log n)^{c_m}} \leq C, \label{eq:transmap}
\end{align}
where $c_m = 3/4$ if $m=2$ and $c_m=1/m$ otherwise. 
\end{proposition}
Here $d_{\M}$ denotes the geodesic distance on $\M$ and the notation $\gamma_n=T_{n\sharp}\gamma$ means that $\gamma \bigl(T_n^{-1}(U)\bigr)=\gamma_n(U)$ for all measurable $U$, so that $T_n$ is a measure preserving map.  Intuitively $T_n$ transports the mass of $\M$ to the points $\{x_i\}_{i=1}^n$, so that the preimage of each singleton gets $1/n$ of the mass, i.e., $\gamma(T_n^{-1}(\{x_i\}))=1/n$. Furthermore, the sets $U_i:=T_n^{-1}(\{x_i\})$ form a partition of $\M$ and by Proposition \ref{prop:transmap} we have $U_i\subset B_{\M}(x_i,\eps_n)$, where 
\begin{align}
    \eps_n:=\underset{x\in\M}{\operatorname{sup}} \, d_{\M}(x,T_n(x)) \lesssim \frac{(\log n)^{c_m}}{n^{1/m}}  \label{eq:OTdistance}
\end{align}
and $B_{\M}(x,r)$ refers to the geodesic ball centered at $x$ with radius $r$. Therefore each $U_i$ is ``centered around'' $x_i$ and the function 
\begin{align*}
    u_n\circ T_n(x) = \sum_{i=1}^n u_n(x_i)\mathbf{1}_{U_i}(x), \quad x\in \M,
\end{align*} 
can be thought of as a locally constant interpolation of $u_n$ to a function on $\M$. This motivates us to quantify the convergence of graph Mat\'ern fields by the expected $L^2:=L^2(\gamma)$-norm between $u_n\circ T_n$ and $u$. Recall that $h_n$ is the graph connectivity that is crucial in defining $u_n$. 
\begin{theorem}\label{thm:Rate}
Suppose $\tau$ is Lipschitz, $\kappa\in C^1(\M)$ and both are bounded below by positive constants. 
Let $s>m$ and 
\begin{align}
    \frac{(\log n)^{c_m}}{n^{1/m}} \ll h_n \ll \frac{1}{n^{1/2s}}, \label{eq:hnscaling}
\end{align}
where $c_m=3/4$ if $m=2$ and $c_m=1/m$ otherwise. 
Then, with probability one, 
\begin{align*}
    \mathbb{E}\|u_n\circ T_n-u\|_{L^2} \xrightarrow{n\rightarrow{\infty}} 0,
\end{align*}
where $\{T_n\}_{n=1}^{\infty}$ is a sequence of transport maps as in Proposition \ref{prop:transmap}.
If further $s>(5m+1)/2$ and
\begin{align}
    h_n\asymp \sqrt{\frac{(\log n)^{c_m}}{n^{1/m}}}, \label{eq:hnscaling2}
\end{align}
then, with probability one, 
\begin{align*}
    \mathbb{E}\|u_n\circ T_n-u\|_{L^2} =O\left(\sqrt{h_n}\right)= O\left(\frac{(\log n)^{c_m/4}}{n^{1/4m}}\right). 
\end{align*}
\end{theorem}

\begin{remark}
Notice that $\eps_n$ as defined in \eqref{eq:OTdistance} represents the finest scale of variations that the point cloud can resolve and hence $h_n$ needs to be much larger than $\eps_n$ to capture local geometry, which is reflected in the lower bound of \eqref{eq:hnscaling}. 
We will see below that the scaling \eqref{eq:hnscaling2} gives the optimal convergence rate. In such cases the geodesic ball of radius $h_n$ has volume $O(n^{-1/2})$ up to logarithmic factors and hence the average degree of the graph is  $O(n^{1/2})$ up to logarithmic factors as the $x_i$'s are uniformly distributed. Therefore, the number of nonzero elements in the weight matrix (and hence the graph Laplacian) is $O(n^{3/2})$ up to logarithmic factors. 

Although not as sparse as the stiffness matrix in FEM approach which has $O(n)$ nonzero entries, the graph Laplacian still gives computational speed-up. Since we only assume a random design for the point cloud $x_i$'s, we are essentially analyzing the worst case scenario. Therefore we expect that the convergence rate and the resulting sparsity of the graph Laplacian can be improved if the point cloud is structured. However this will require a different analysis which is beyond the scope of this paper.   \qed
\end{remark}

In previous related work, the $T_n$'s are taken to be the optimal transport maps, whose computation can be quite challenging, especially in high dimensions \cite{peyre2019computational}. Therefore we consider an alternate interpolation map below that can be computed efficiently. Consider the Voronoi cells
\begin{align*}
    V_i=\{x\in\M: |x-x_i|=\underset{j=1,\ldots,n}{\operatorname{min}} |x-x_j|\}.
\end{align*}
Up to a set of ambiguity of $\gamma$-measure 0, the $V_i$'s form a partition of $\M$ and we shall assume they are disjoint by assigning points in their intersections to only one of them. We then define a map $\mathcal{T}_n:\M\rightarrow \M_n$ by $\mathcal{T}_n(x)=x_i$ for $x\in V_i$ and consider the nearest-neighbor interpolation 
\begin{align*}
    u_n\circ \mathcal{T}_n(x) = \sum_{i=1}^nu_n(x_i)\mathbf{1}_{V_i}(x), \quad x\in\M.
\end{align*}
Note its resemblance with $u_n\circ T_n$. Indeed, the sets $V_i$ and $U_i$'s are comparable and we have a similar result as Theorem \ref{thm:Rate}. 
\begin{theorem}\label{thm:Rate2}
Suppose $\tau$ is Lipschitz, $\kappa\in C^1(\M)$ and both are bounded below by positive constants. 
Let $s>(5m+1)/2$ and
\begin{align*}
    h_n\asymp \sqrt{\frac{(\log n)^{c_m}}{n^{1/m}}}.
\end{align*}
Then, with probability one, 
\begin{align*}
    \mathbb{E}\|u_n\circ \mathcal{T}_n-u\|_{L^2} = O\left(\frac{(\log n)^{(2m+1)c_m/4}}{n^{1/4m}}\right). 
\end{align*}
\end{theorem}
\begin{remark}
We see that the rate has an additional factor $(\log n)^{mc_m/2}$, which comes from the bound for eigenfunction approximation using the map $\mathcal{T}_n$. In fact, due to this additional logarithmic term, we were unable to show convergence for general $s>m$ using the same strategy as before. \qed
\end{remark}

\subsection{Outline of Proof}
From the Karhunen-Lo\`eve expansions \eqref{eq:KLu} and \eqref{eq:KLun} we can see that the fundamental step in analyzing the convergence of the graph representations is to obtain bounds for both the eigenvalue and eigenfunction approximations. 

Since $L^{\tau,\kappa}_n$ and $\L^{\tau,\kappa}$ are self-adjoint with respect to the $L^2(\gamma_n)$ and $L^2(\gamma)$ inner products, respectively, and are both positive definite, their spectra are real and positive. Denote by $\{(\lambda_n^{(i)}, \psi_n^{(i)})\}_{i=1}^n$ and $\{(\lambda^{(i)},\psi^{(i)})\}_{i=1}^{\infty}$ the eigenpairs of $L^{\tau,\kappa}_n$ and $\L^{\tau,\kappa}$, with the eigenvalues in increasing order. Suppose we are in a realization where the conclusion of Proposition \ref{prop:transmap} holds, and recall the definition of $\eps_n$ in \eqref{eq:OTdistance}. We formally have the following results on spectral approximations. (Rigorous statements can be found in  Appendices \ref{sec:AevalRate} and \ref{sec:AefunRate}.) 

\begin{theorem}[Eigenvalue Approximation cf. Theorem \ref{thm:evalRate2}] \label{thm:evalRate}
Suppose  $\eps_n \ll h_n$ and $h_n\sqrt{\lambda^{(k_n)}}\ll 1$. Then
\begin{align*}
    \frac{|\lambda^{(i)}_{n}-\lambda^{(i)}|}{\lambda^{(i)}} =O\left(\frac{\eps_n}{h_n}+h_n\sqrt{\lambda^{(i)}}\right),
\end{align*}
for $i=1,\ldots,k_n$.
\end{theorem}

\begin{theorem}[Eigenfunction Approximation cf. Theorem \ref{thm:efunRate2}] \label{thm:efunRate}
Suppose  $\eps_n \ll h_n$ and $h_n\sqrt{\lambda^{(k_n)}}\ll 1$. Then there exists orthonomralized eigenfunctions $\{\psi_n^{(i)}\}_{i=1}^n$ and $\{\psi^{(i)}\}_{i=1}^{\infty}$ so that 
\begin{align*}
    \|\psi^{(i)}_{n}\circ T_n-\psi^{(i)}\|_{L^2} &=O\left(i^{\frac{3}{2}}\sqrt{\frac{\eps_n}{h_n}+h_n\sqrt{\lambda^{(i)}}}\right),\\
    \|\psi^{(i)}_{n}\circ \mathcal{T}_n-\psi^{(i)}\|_{L^2} &=O\left((\log n)^{\frac{mc_m}{2}}i^{\frac{3}{2}}\sqrt{\frac{\eps_n}{h_n}+h_n\sqrt{\lambda^{(i)}}}\right),
\end{align*}
for $i=1,\ldots,k_n$.
\end{theorem}

Theorems \ref{thm:evalRate} and \ref{thm:efunRate} generalize existing results in \cite{burago2015graph,trillos2019error} where $\tau$ and $\kappa$ are constant. Theorem \ref{thm:Rate} and \ref{thm:Rate2} are shown in Appendix \ref{sec:APTL2} based on these two results. 

Here we briefly illustrate the main idea for $u_n\circ T_n$. The assumption that $h_n\sqrt{\lambda^{(k_n)}}\ll 1$ in  Theorems \ref{thm:evalRate} and \ref{thm:efunRate}  is crucial in that the spectral approximations are only provably accurate up to the $k_n$-th eigenvalue and eigenfunction. Therefore to bound the difference between $u_n\circ T_n$ and $u$, we need to consider the truncated series 
\begin{align} \label{eq:truncated}
     u_n^{k_n} :=\tau_n^{\frac{s}{2}-\frac{m}{4}}\kappa_n^{\frac{m}{2}} \sum_{i = 1}^{k_n} \left[\lambda_n^{(i)}\right]^{\frac{s}{2}} \, \xi^{(i)} \, \psi_n^{(i)},
\end{align}
and such truncation introduces an error of order $\sqrt{n}k_n^{-s/m}$. Therefore $k_n$ needs to satisfy $n^{m/2s}\ll k_n \ll h_n^{-m}$ and this explains the upper bound on $h_n$ in \eqref{eq:hnscaling}. 
By repeated application of the triangle inequality, we can show that $\mathbb{E}\|u_n\circ T_n-u\|_{L^2}$ is dominated by the error coming from the truncation and eigenfunction approximation:
\begin{align}
    \mathbb{E}\|u_n\circ T_n-u\|_{L^2} \lesssim \sqrt{n}k_n^{-\frac{s}{m}}+ \sum_{i=1}^{k_n} \left[\lambda^{(i)}\right]^{-\frac{s}{2}} \|\psi_n^{(i)}\circ T_n-\psi^{(i)}\|_{L^2}.  \label{eq:tradeoff}
\end{align}
If we are only interested in showing convergence without a rate, then we can first fix an $\ell\in\mathbb{N}$ and use the fact that $\|\psi_n^{(i)}\circ T_n\|^2_{L^2}=n^{-1}\sum_{i=1}^n \psi_n^{(i)}(x_i)^2=\|\psi_n^{(i)}\|^2_{L^2(\gamma_n)}=1$ to get
\begin{align*}
    \mathbb{E}\|u_n\circ T_n-u\|_{L^2} \lesssim \sqrt{n}k_n^{-\frac{s}{m}}+ \sum_{i=1}^{\ell} \left[\lambda^{(i)}\right]^{-\frac{s}{2}} \|\psi_n^{(i)}\circ T_n-\psi^{(i)}\|_{L^2}+\sum_{i=\ell+1}^{k_n}\left[\lambda^{(i)}\right]^{-\frac{s}{2}}.
\end{align*}
Then letting $n\rightarrow\infty$, we have $\|\psi_n^{(i)}\circ T_n-\psi^{(i)}\|_{L^2}\rightarrow 0$ for $i=1,\ldots,\ell$ and 
\begin{align*}
    \underset{n\rightarrow\infty}{\operatorname{lim\,sup}}\,\mathbb{E}\|u_n\circ T_n-u\|_{L^2} \lesssim \sum_{i=\ell+1}^{\infty}\left[\lambda^{(i)}\right]^{-\frac{s}{2}}.  
\end{align*}
The last expression goes to 0 if we let $\ell \rightarrow \infty$ given that $s>m$. 

However if we want to derive a rate of convergence, then $k_n$ needs to be chosen carefully in \eqref{eq:tradeoff}: $k_n$ should be small so that the spectral approximations up to level $k_n$ are sufficiently accurate, but at the same time not be too small to leave a large truncation error. In particular, by Theorem \ref{thm:efunRate} and Weyl's law an upper bound on the second term \eqref{eq:tradeoff} is  
\begin{align}
    k_n[\lambda^{(k_n)}]^{-s/2} \underset{i=1\ldots,k_n}{\operatorname{max}}\,\|\psi_n^{(k_n)}-\psi^{(k_n)}\|_{L^2} \lesssim k_n^{-\frac{s}{m}+\frac52}\sqrt{\frac{\eps_n}{h_n}+h_nk_n^{\frac{1}{m}}}. \label{eq:efunsumbound}
\end{align}
If we only have $s>m$, then this does not go to zero for any choice of $h_n$ given the constraint that $ n^{m/2s}\ll k_n$. Hence we need the additional assumption that $s>\frac{5}{2}m+\frac12$, which allows us to bound \eqref{eq:efunsumbound} by $\sqrt{\frac{\eps_n}{h_n}+h_n}$ and then
\begin{align*}
    \mathbb{E}\|u_n\circ T_n-u\|_{L^2}
    \lesssim \sqrt{n}k_n^{-\frac{s}{m}} + \sqrt{\frac{\eps_n}{h_n}+h_n}.
\end{align*}
Therefore optimal rates are obtained by setting $h_n\asymp \sqrt{\eps_n}$ and $k_n$ correspondingly, which together with \eqref{eq:OTdistance} gives the scaling \eqref{eq:hnscaling2}. 

\begin{remark}
The paper  \cite{trillos2017consistency} proposed to consider directly the truncated series \eqref{eq:truncated} as the graph field and this will allow us to obtain rates of convergence in Theorem \ref{thm:Rate} for all $s>m$. But in this way the precision matrix representation as in \eqref{eq:covarianceanysotropicdiscrete} no longer holds and one will need to perform spectral decomposition on $L_n^{\tau,\kappa}$ for sampling, which generally is more costly than Cholesky factorization.  \qed 
\end{remark}

\section{Numerical Examples}\label{sec:numericalexamples}
In this section we demonstrate the use of the graph Mat\'ern models introduced in Section \ref{sec:graphdiscretizations} by considering applications in Bayesian inverse problems, spatial statistics and graph-based machine learning. 

For the three examples we employ graph Mat\'ern models as priors within the general framework of latent Gaussian models, briefly overviewed in Subsection \ref{sec:LGM}. Subsection \ref{sec:BIP} studies a toy Bayesian inverse problem on a manifold setting. Our aim is to compare the modeling of length scale through $\tau_n$ and $\kappa_n;$ we further show that the accuracy of the reconstruction with the graph-based approach is satisfactory and that adding nonstationarity may help to overcome the poor performance of more naive hierarchical approaches in large noise regimes.
In Subsection \ref{sec:SS} we investigate the use of graph Mat\'ern fields for interpolating  U.S. county-level precipitation data, assuming to only have access to pairwise distances between counties and precipitation data for some of them.  Contrary to finite element representations, the graph-based approach is applicable in this discrete setting  without the need of performing multidimensional scaling to reconstruct the configuration of the point cloud.  In addition, the graph approach does not require to introduce any artificial nodes.  We also compare the performance of stationary and nonstationary graph Mat\'ern models. In Subsection \ref{sec:ML}, a semi-supervised classification problem in machine learning is studied, where the low dimensional structure of the data naturally motivates the graph-based approach; we further show that nonstationary models may improve the classification accuracy over stationary ones.  


\subsection{A General Framework: Latent Gaussian Models} \label{sec:LGM}
Latent Gaussian models are a flexible subclass of structured additive regression models defined in terms of a likelihood function, a latent process and hyperparameters. Let $\{x_i\}_{i=1}^n$ be a collection of features that we identify with graph nodes. The observation variable $y$ is modeled as a (possibly noisy) transformation of the latent process $u_n:= \bigl[u_n(x_1),\ldots,u_n(x_n)\bigr]^T$, which conditioned on the hyperparameters follows a Gaussian distribution. Finally,  a prior is placed on the hyperparameters. More precisely, we have 
\begin{align*}
	y|u_n,\mu & \sim  \pi(y|u_n,\mu), \\ 
	u_n|\theta   & \sim \mathcal{N}\left(0,Q(\theta)^{-1}\right),\\
	(\mu,\theta)     & \sim \pi(\mu,\theta),
\end{align*}
where $Q(\theta)$ is the precision matrix of the latent process and $\mu,\theta$ are hyperparameters. Markov Chain Monte Carlo inference methods are standard in Bayesian inverse problems with complex likelihood functions, but less computationally expensive deterministic approximations are often preferred in other applications. In particular, the integrated nested Laplace approximations proposed by \cite{rue2009approximate} and the corresponding R-INLA package has greatly facilitated inference of such models. 

The sparsity of the precision matrix $Q(\theta)$ is crucial for efficient likelihood evaluations and sampling of the latent process. For the problems that we consider, the latent process will be modeled as a graph Mat\'ern field, i.e.,
\begin{align} 
	u_n|\tau_n \sim \mathcal{N}\left(0,Q(\tau_n,s)^{-1}\right), \quad Q(\tau_n,s) = \tau_n^{\frac{m}{4}-\frac{s}{2}}(\tau_n+\Delta_n)^{s}\tau_n^{\frac{m}{4}-\frac{s}{2}},  \label{eq:latent}
\end{align}
where $\Delta_n$ is a graph Laplacian constructed with the $x_i$'s and $\tau_n$ is a diagonal matrix modeling the length scale at each node.  We note that the graph Laplacian is often sparse and its sparsity is inherited by $Q(\tau_n,s)$ for  small or moderate integer $s$. 
  
A constant length scale graph Mat\'ern field hyperprior is then placed on $\log(\tau_n)$:
\begin{align}
	\log \tau_n \sim \mathcal{N}\left(0,\nu^{s_0-\frac{m}{2}} (\nu I+ \Delta_n)^{-s_0}\right), \label{eq:hyper}
\end{align}
where $\nu$ and $s_0$ are chosen by prior belief on the length scale. However, when $n$ is large, learning $\tau_n$ as an $n$-dimensional vector is computationally demanding. We instead adopt a truncated Karhunen-Lo\'eve approximation for $\tau_n$. Recall that $\log \tau_n$ has the characterization 
\begin{align*}
	\log \tau_n = \nu^{\frac{s_0}{2}-\frac{m}{4}}\sum_{i=1}^{n} \left[\nu + \lambda_n^{(i)}\right]^{-\frac{s_0}{2}} \xi^{(i)} \psi_n^{(i)},
\end{align*}
where $\xi^{(i)} \overset{i.i.d.}{\sim} \mathcal{N}(0,1)$ and $\{(\lambda_n^{(i)},\psi_n^{(i)})\}_{i=1}^n$ are the eigenpairs of $\Delta_n$. Since the $\lambda_n^{(i)}$'s are increasing, the contribution of the higher frequencies is less significant. Hence we consider a truncated expansion and model $\log \tau_n$ as 
\begin{align*}
	\log \tau_n = \nu^{\frac{s_0}{2}-\frac{m}{4}}\sum_{i=1}^{n_0} \left[\nu + \lambda_n^{(i)}\right]^{-\frac{s_0}{2}} \theta^{(i)} \psi_n^{(i)},
\end{align*}
where $n_0\ll n$ is chosen based on the spectral growth and prior belief on $\tau_n$. Now the hyperparameters are the $\theta_{i}$'s, which are only $n_0$-dimensional, and the hyperprior for each  is naturally taken to be the standard normal. Therefore a complete model of our interest in the following subsections  can be summarized as 
\begin{align}
	y|u,\mu & \sim  \pi(y|u,\mu), \nonumber \\  
	u|\theta,s   & \sim \mathcal{N}\left(0,Q(\theta,s)^{-1}\right), \label{eq:latent}\\
	(\mu,s) \sim \pi(\mu,s),& \quad \theta \sim \mathcal{N}(0,I_{n_0}),
\end{align}
where 
\begin{align}
	Q(\theta,s) &=[\tau_n(\theta)]^{\frac{m}{4}-\frac{s}{2}}[\tau_n(\theta)+\Delta_n]^{s}[\tau_n(\theta)]^{\frac{m}{4}-\frac{s}{2}},  \label{eq:Q} \\
	\log \tau_n(\theta)& = \nu^{\frac{s_0}{2}-\frac{m}{4}}\sum_{i=1}^{n_0} \left[\nu + \lambda_n^{(i)}\right]^{-\frac{s_0}{2}} \theta^{(i)} \psi_n^{(i)}. \label{eq:hyperKL}
\end{align}

\begin{remark}
We remark that \eqref{eq:hyperKL} can be viewed as representing $\tau_n$ as a combination of the first several eigenfunctions of $\Delta_n$, which are a natural basis for functions over the point cloud. A disadvantage is that \eqref{eq:hyperKL} requires knowledge of the $\psi_n^{(i)}$'s. However since $\Delta_n$ is sparse and we only need to know the first $n_0\ll n$ eigenfunctions, the computational cost for $\tau_n$  is still better than $O(n^3)$. \qed 
\end{remark}

\begin{remark} \label{rmk:MV}
The marginal variance of the latent process $u_n$ can be tuned and fixed easily by matching the scales of $u_n$ and the data $y$, and thus we do not include a marginal variance parameter.  
Indeed, the normalizing factors $\tau_n^{m/4-s/2}$ guarantee that $\mathbb{E}|u_n|^2$ are roughly the same for different $\tau_n$'s. Hence one can for instance estimate $\mathbb{E}|u_n|^2$ by setting $\tau_n\equiv 1$ and normalize the observations $y$ by $\sqrt{\mathbb{E}|u_n|^2/|y|^2}$. 

Similarly, we need to tune for the marginal variance of the hyperparameters as in \eqref{eq:hyper}. As $\tau_n$ essentially acts as a cut off on the significant frequencies, this can be done by matching the scale of $\tau_n$ with the eigenvalues of $\Delta_n$ based on one's prior belief.     \qed
\end{remark}
Suppose for illustration that we are interested in the simple regression problem of inferring a Mat\'ern field $u(x)$ based on data $y$ comprising Gaussian measurement of $u$ at given locations/features $x_1, \ldots x_n.$ As noted in the introduction, the computation cost scales as $\mathcal{O} (n^3).$ However, by modeling $u_n$ using a graph Mat\'ern model we obtain a GMRF approximation, with sparse precision matrix, dramatically reducing the computational cost. Thus,  one could introduce further auxiliary nodes $x_{n+1}, \ldots, x_N$ with $N \gg n$ to improve the prior GMRF approximation of the original Mat\'ern model and still reduce the computational cost over formulations based on GF priors. Such ideas arise naturally in semi-supervised applications in machine learning where most features are unlabeled, but can also be of interest in applications in spatial statistics, as discussed in \cite{rue2005gaussian}[Chapter 5] and grant further investigation in Bayesian inverse problems.  

\subsection{Application in  Bayesian Inverse Problems} \label{sec:BIP}
In this subsection we investigate the use of nonstationary graph Mat\'ern models to define prior distributions in Bayesian inverse problems. 
For simplicity of exposition and to avoid distraction from our main purpose of illustrating the modeling of the nonstationarity, we consider a toy example taken from the inverse problem literature  \cite{roininen2019hyperpriors}. The ideas presented here apply immediately to Bayesian inverse problems with more involved likelihood functions, defined for instance in terms of the solution operator of a differential equation \cite{harlim2020kernel,bigoni2020data}.

We study the reconstruction of a signal function given noisy but direct point-wise observations. The domain of the problem is taken to be the unit circle, where the hidden signal $u^{\dagger}$ is parametrized by $t\in[0,2\pi)$ as   
\begin{align*}
	u^{\dagger}(t)= 
\begin{cases}
\exp\left(4-\frac{\pi^2}{t(\pi-t)}\right),\quad & t\in(0,\pi), \\
1, \quad & t\in[\pi+0.5,1.5\pi], \\
-1,\quad & t\in(1.5\pi,2\pi-0.5], \\
0,\quad & \text{otherwise}. 
\end{cases}
\end{align*}
Hence if $x=(\cos(t),\sin(t))$ for $t\in[0,2\pi)$, then $u^{\dagger}(x)$ is understood as $u^{\dagger}(t)$. Such signal is considered in \cite{roininen2019hyperpriors} for its varying length scale, where the domain is the interval $[0,10] \subset \mathbb{R}$ and a uniform grid finite difference discretization is used to define a Mat\'ern prior following the SPDE approach. Here we suppose instead to have only indirect access to the domain through $n=1000$ points $x_i$'s that are drawn independently from the uniform distribution on the circle, and use a graph Mat\'ern model.  

We assume to be given noisy observations of the signal at $J=n/2$ points:  
\begin{align} 
	y(x_i)=u^{\dagger}(x_i)+\eta_i, \quad \eta_i\overset{i.i.d.}{\sim} \Nc(0,\sigma^2), \quad i=1,\ldots,J, \label{eq:ex1_obs}
\end{align}
where $\sigma$ is set to be 0.1 and we have observations at every other node. To recover the signal function at the nodes $x_i$'s, we adopt a hierarchical Bayesian approach which we cast into the framework of latent Gaussian models. More precisely, the observation equation \eqref{eq:ex1_obs} gives the likelihood model 
\begin{align*}
	y|u_n \sim \mathcal{N}(Su_n,\sigma^2 I_J),
\end{align*}
where $S\in \mathbb{R}^{J\times n}$ is a matrix of 0's and 1's  that indicates the location of the observations. The latent process $u_n:=\bigl[ u^{\dagger}(x_1),\ldots,u^{\dagger}(x_n) \bigr]^T$ and the hyperparameters are modeled as in Section \ref{sec:LGM}, where the smoothness is fixed as $s=2$ and the other parameters are chosen as $s_0=1, \nu=10, n_0=21$. We note that by setting $s_0=1$, the hyperprior is actually an approximation of a zero-mean Gaussian process with exponential covariance function, where the sample paths can undergo sudden changes. The choice $n_0=21$ is motivated by the fact that the Laplacian on the circle has eigenvalues $\{i^2\}_{i=0}^{\infty}$, where any non-zero eigenvalue has multiplicity 2. Therefore the cutoff is at the eigenvalue 100, which is an order of magnitude larger than $\nu$, and higher frequencies are less consequential. The graph Laplacian is constructed as in Section \ref{ssec:mainconstruction} with connectivity $\eps_n=4\times n^{-1/1.8}$. 

\begin{algorithm}[!htb]
\caption{Posterior sampling $u,\tau|y$} \label{algo:ex1}
\begin{algorithmic}
\State{Initialize $u_n^0$, $\theta^0$ and set step size $\beta_{\theta}$.}
\For{$i=1,\ldots,M$}
\State{(a) Generate $\xi\sim \mathcal{N}(0,I_{J+n})$ and update \begin{align*}
	u_n^i=\begin{bmatrix} 
		\sigma^{-1}S  \\
		L(\theta^{i-1})
	\end{bmatrix} ^{\dagger}
	\left(\begin{bmatrix}
		\sigma^{-1}y \\
		0			
	\end{bmatrix}
		+\xi \right),
\end{align*} 
where $L(\theta)$ is the Cholesky factor of $Q(\theta)$ and $\dagger$ denotes the matrix pseudoinverse.}
\State{(b)}
\For{$j=1,\ldots,n_0$}
\State{(i) Generate $\xi^{(j)} \sim \mathcal{N}(0,1)$ and  propose $\tilde{\theta}^{i,(j)}=\theta^{i-1,(j)}+\beta_{\theta}\xi^{(j)}.$}
\State{(ii) Denote $\theta^{i,(-j)}=(\theta^{i,(1)}, \ldots, \theta^{i,(j-1)}, \theta^{i-1,(j+1)},\ldots,\theta^{i-1,(n_0)})$. Accept $\tilde{\theta}^{i,(j)}$ with probability 
\begin{align*}
	p = \operatorname{min} \left\{1, \frac{\pi\left(\tilde{\theta}^{i,(j)},\theta^{i,(-j)} |u_n^i,y\right)}{\pi\Big(\theta^{i-1,(j)},\theta^{i,(-j)}|u_n^i,y\Big)}\right\},
\end{align*}
where $\pi(\theta|u_n,y)$ is given in \eqref{eq:posttheta}}. 
\EndFor
\EndFor
\end{algorithmic}
\end{algorithm}

We will follow a similar MCMC sampling as in \cite{roininen2019hyperpriors} detailed in Algorithm \ref{algo:ex1} for inferring the signal function $u_n$ together with the length scale $\tau_n$. To illustrate the idea, we notice that the observation equation and the graph Mat\'ern model for $u_n$ translate into the equations 
\begin{align*}
	\sigma^{-1}Su_n&=\sigma^{-1}y + \xi_1, \\
	L(\theta)u_n & = \xi_2 ,
\end{align*}
where $L(\theta)$ is the Cholesky factor of $Q(\theta)$ and $\xi_1 \sim \mathcal{N}(0, I_J), \xi_2\sim \mathcal{N}(0,I_n)$. The above pair of equations motivate the update for $u_n$ as 
\begin{align*}
	\begin{bmatrix} 
		\sigma^{-1}S  \\
		L(\theta)
	\end{bmatrix} ^{\dagger}
	\left(\begin{bmatrix}
		\sigma^{-1}y \\
		0			
	\end{bmatrix}
		+\xi \right),
\end{align*}
where $\xi \sim \mathcal{N}(0, I_{J+n})$. The hyperparameters $\theta$ are updated with a Metropolis-within-Gibbs sampling scheme, where the full posterior $\theta |u_n,y$ has the form 
\begin{align}
	\pi(\theta|u_n,y) \propto  \sqrt{\det(Q(\theta))}\exp \left(-\frac12 u_n^TQ(\theta)u_n-\frac12 |\theta|^2\right).  \label{eq:posttheta}
\end{align}
In Figures \ref{fig:taua} and \ref{fig:taub} we plot the posterior means of $u_n$ and $\tau_n$ and 95\% credible intervals for each of their coordinates. The oscillatory paths of the reconstructions are due to the graph approximation, where the eigenfunctions $\psi_n^{(i)}$ of the graph Laplacian are in general very ragged. We notice that $\tau_n$ varies rapidly in the region where the signal is piecewise constant, indicating a change of length scale. Moreover, the sudden jump of the signal from 1 to -1 suggests a small local length scale which leads to a larger $\tau_n$, as predicted in Remark \ref{rmk:tau&kappa}. 

To further understand the effect of modeling the length scale through $\kappa_n$, we choose a different prior for the latent process by setting 
\begin{align*}
	Q(\theta) &= \kappa_n(\theta)^{-\frac{m}{4}} (I+\Delta_n^{\kappa})^{s} \kappa_n(\theta)^{-\frac{m}{4}}, \\
\log \kappa_n(\theta)& = \nu^{\frac{s_0}{2}-\frac{m}{4}}\sum_{i=1}^{n_0} \left[\nu + \lambda_n^{(i)}\right]^{-\frac{s_0}{2}} \theta^{(i)} \psi_n^{(i)},
\end{align*}
so that the length scale is controlled by $\kappa_n$ instead. Here $\Delta_{n}^{\kappa}$ is the discrete approximation of $\nabla \cdot (\kappa \nabla)$ introduced in Section \ref{ssec:modeling} and we adopt the same hyperparameter modeling as for $\tau_n$. In Figures \ref{fig:kappaa} and \ref{fig:kappab}, we plot the posterior means and 95\% credible intervals for each coordinate of $u_n$ and $\kappa_n$, as a comparison with their counterparts for $\tau_n$. The figures show that the two approaches give similar reconstructions for the signal and, in agreement with the intuition given in Remark \ref{rmk:tau&kappa}, $\tau_n$ and $\kappa_n$ are almost inversely proportional to each other. 

\begin{figure}[!htb]
\minipage{1\textwidth}
\centering
\minipage{0.33\textwidth}
  \includegraphics[width=\linewidth]{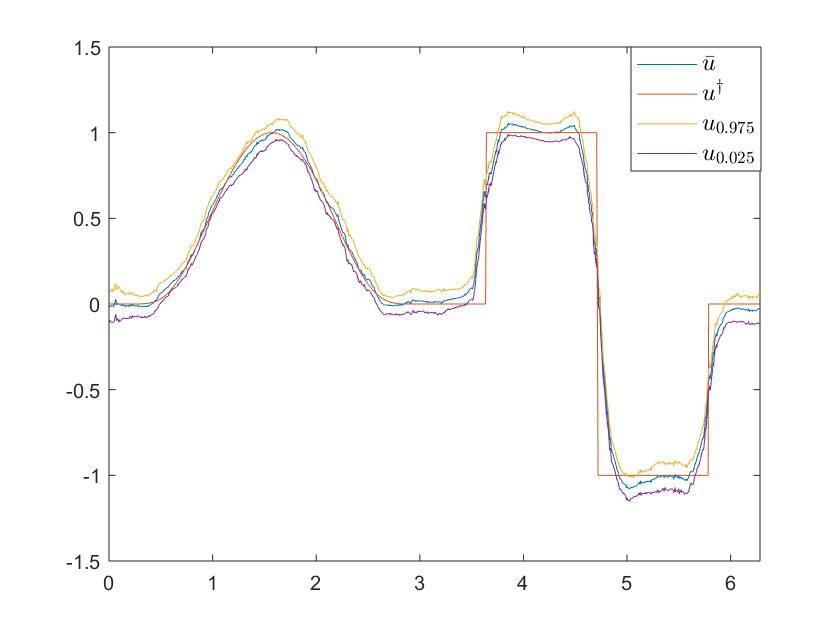}
\vspace{-20pt}
\subcaption{\label{fig:taua}}
\endminipage
\minipage{0.33\textwidth}
  \includegraphics[width=\linewidth]{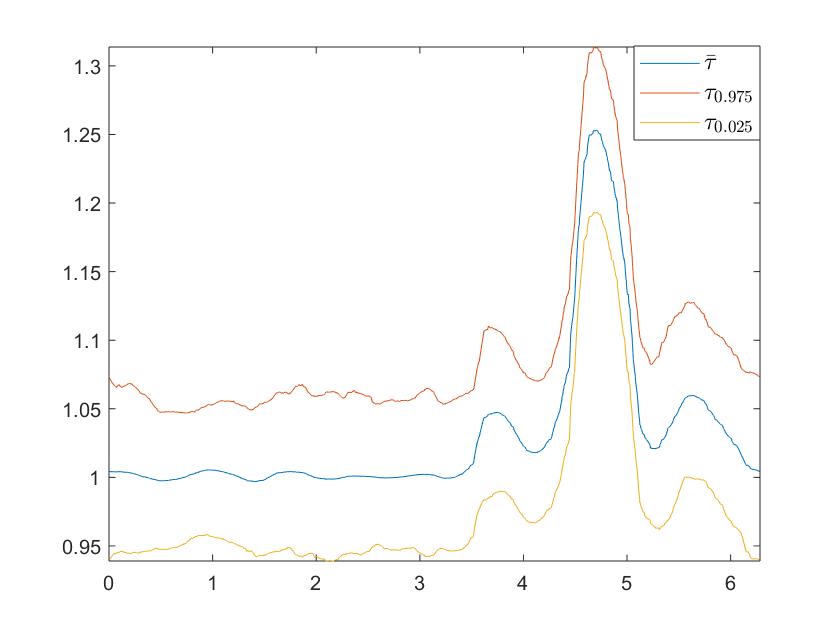}
\vspace{-20pt}
\subcaption{\label{fig:taub}}
\endminipage
\minipage{0.33\textwidth}
  \includegraphics[width=\linewidth]{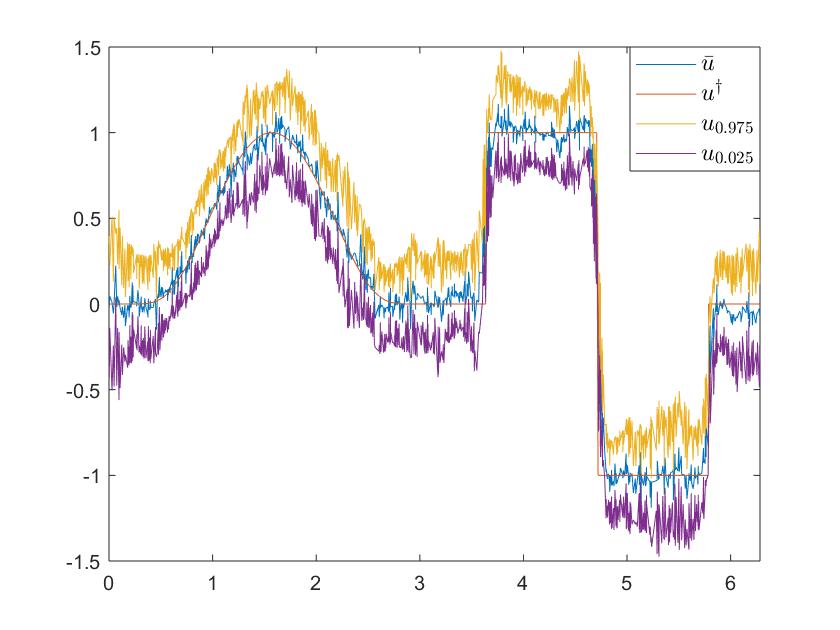}
\vspace{-20pt}
\subcaption{\label{fig:tauc}}
\endminipage
\endminipage
\caption{Posterior means and 95\% credible intervals for (a) signal from nonstationary model; (b) length scale of nonstationary model; (c) signal from stationary model (constant $\tau_n$) when length scale is modeled through $\tau_n$. }
\label{figure:ex1_tau}
\end{figure}

\begin{figure}[!htb]
\minipage{1\textwidth}
\centering
\minipage{0.33\textwidth}
  \includegraphics[width=\linewidth]{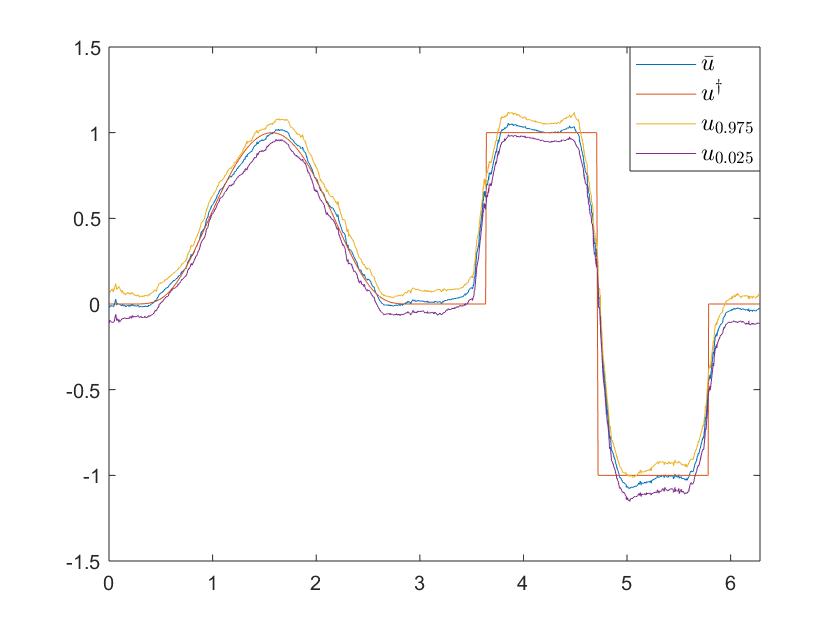}
\vspace{-10pt}\subcaption{\label{fig:kappaa}}
\endminipage
\minipage{0.33\textwidth}
  \includegraphics[width=\linewidth]{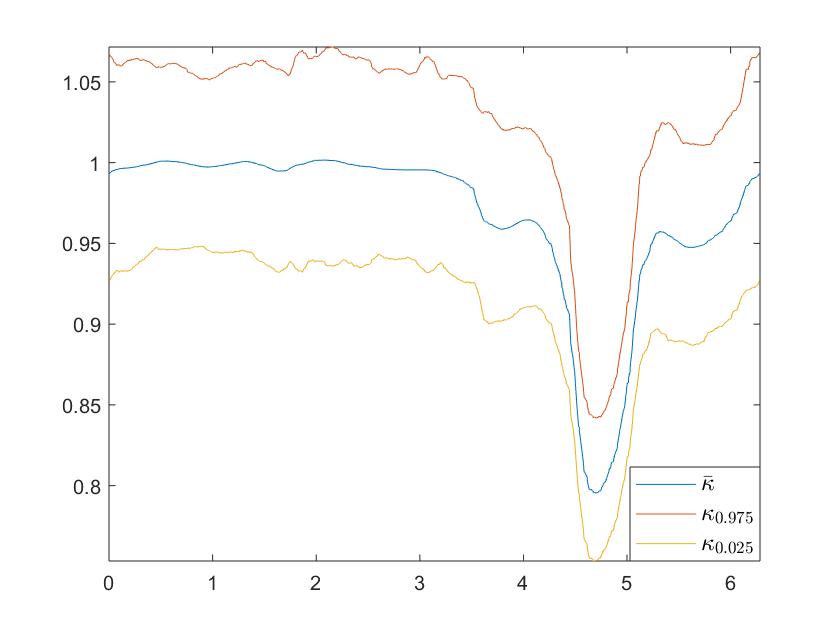}
\vspace{-10pt}\subcaption{\label{fig:kappab}}
\endminipage
\minipage{0.33\textwidth}
  \includegraphics[width=\linewidth]{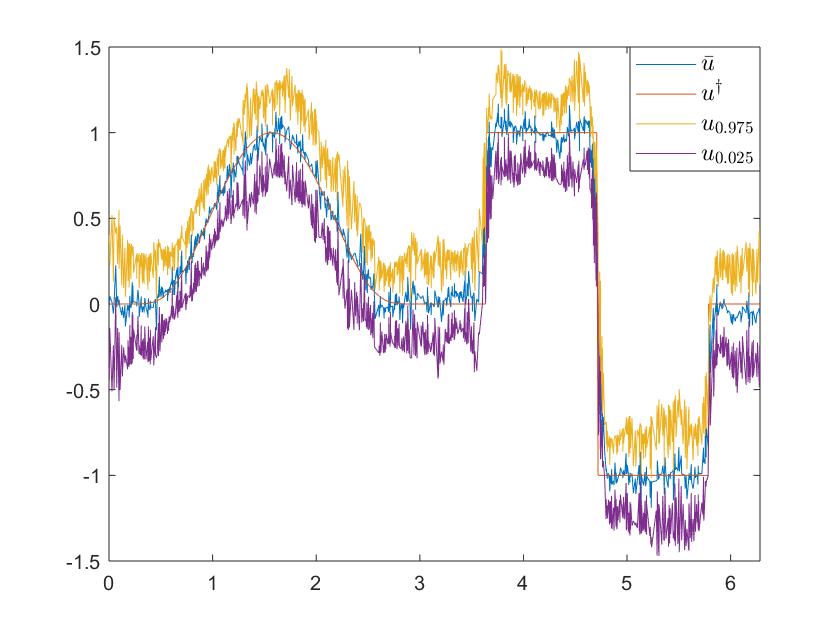}
\vspace{-10pt}\subcaption{\label{fig:kappac}}
\endminipage
\endminipage
\caption{Posterior means and 95\% credible intervals for (a) signal from nonstationary model; (b) length scale of nonstationary model; (c) signal from stationary model constant $\kappa_n$), when length scale is modeled through $\kappa_n$. }
\end{figure}
\begin{remark}
As noted in \cite{harlim2020kernel}[Section 4.5], the hierarchical approach performs poorly if the noise in the observations is large and the latent process is modeled with a constant length scale; the single length scale is blurred by the noise and the model essentially fits the noisy observations, as can be seen from the oscillatory reconstructions in Figures \ref{fig:tauc} and \ref{fig:kappac}. An important observation stemming from the above example is that adding nonstationarity into the length scale may help alleviate such issue.  \qed
\end{remark}

\subsection{Application in Spatial Statistics}    \label{sec:SS}
In this subsection we consider interpolation of county-level precipitation data in the U.S. for January 1981, available from \url{https://www.ncdc.noaa.gov/cag/county/mapping}. Similar problems have been studied in \cite{fuglstad2015does,bolin2019rational} using the SPDE formulation and finite element representations,  adding nodes for triangulation of the space.  Here we shall assume, however, that \emph{only} pairwise distances between counties are available. In particular, in contrast to \cite{fuglstad2015does,bolin2019rational}  we do not assume to have access to the spatial domain or to the locations of the observations,  and we will perform inference without adding artificial nodes. The inter-county distances are available as great circle distances from \url{https://data.nber.org/data/county-distance-database.html} and are only recorded for each pair of counties that are closer than a certain threshold distance apart, which naturally suggests a graph representation. Let $\{x_i\}_{i=1}^{n=3107}$ denote the $n$ counties (excluding Alaska, Hawaii and several other counties that we do not have precipitation data for). We model the precipitation $y$ with a latent Gaussian model, where the $J$ observations are given as noisy perturbations of the latent process $u_n$: 
\begin{align*}
	y|u,\sigma &\sim  \Nc(Su,\sigma^2 I_J),
\end{align*}  
where $S\in \mathbb{R}^{J\times n}$ is a matrix of  0's and 1's that specifies the observation locations. Notice that we have included the noise size $\sigma$ as a hyperparameter to be inferred.  
The latent process $u_n$ will be modeled in four different ways for comparison purposes as in \cite{bolin2019rational}. The idea is to consider a graph Mat\'ern prior for $u_n$, where 
\begin{align*}
	u_n|\theta,s \sim \mathcal{N}\left(0,Q(\theta,s)^{-1}\right), \quad Q(\theta,s) = \tau_n(\theta)^{\frac{m}{4}-\frac{s}{2}}(\tau_n(\theta)+\Delta_n)^{s}\tau_n(\theta)^{\frac{m}{4}-\frac{s}{2}}, \quad m=2,
\end{align*} 
and consider $s$ to be possibly a hyperparameter, while at the same time allow $\tau_n$ to be identically equal to a constant. In other words, we will model $u_n$ as a stationary/nonstationary graph Mat\'ern field with possibly fractional smoothness parameter. In the most general case, the length scale parameter $\tau_n$ is modeled as in \eqref{eq:hyperKL}, where the parameters are chosen as $s_0=2, \nu=0.1, n_0=10$ and the hyperpriors for $\sigma,s$ are chosen as    
\begin{align*}
	\log \sigma & \sim \mathcal{N}( \log(0.01), 1),\\ 
	\log s &\sim \mathcal{N}(\log 2,1). 
\end{align*} 
The marginal variance of $\tau_n$ is then tuned to be on the same order as the 10th eigenvalue of $\Delta_n$. For the stationary case, the modeling for the latent process simplifies to 
\begin{align*}
	u_n|\tau,s &\sim \mathcal{N}(0,Q(\tau,s)^{-1}), \quad Q(\tau,s)= \tau^{\frac{m}{2}-s} (\tau I_n +\Delta_n )^s,
\end{align*}
and a log-normal hyperprior is placed on $\tau$.

For this problem $\Delta_n$ is an unnormalized graph Laplacian $D-W$, with weights 
\begin{align*}
	W_{ij}= \begin{cases}
\exp\left(-\frac{d_{ij}^2}{2 \underline{d}^2}\right), \quad \text{if} \, d_{ij}\, \text{is recorded}, \\
0, \quad \text{otherwise},
	\end{cases} 
\end{align*}
where $d_{ij}$ is the distance between two counties and  $\underline{d}$ is the mean of all the pairwise distances. As mentioned above, the pairwise distances are only recorded for counties that are less than 100 miles apart, which implies that $W$ is sparse (with the percentage of nonzero entries being 1.6\%).
Instead of using an MCMC sampling scheme, we adopt an evidence maximization approach, where we first compute the optimal hyperparameter $(\sigma^*,s^*,\theta^*)$ (for the most general case) by maximizing the marginal posterior $(\sigma,s, \theta)| y$, and then compute the posterior $u_n|y,\sigma^*, s^*, \theta^*$. For the most general case, the marginal posterior of $(\sigma,s,\theta)$ is equal, up to a constant, to
\begin{align*}
	\log \pi(\sigma,s,\theta|y) &=  \log \pi(\sigma,s,\theta)-J\log \sigma + \frac{1}{2\sigma^4} y^TS\tilde{Q}(\sigma,\theta,s)^{-1}S^Ty - \frac{\|y\|_2^2}{2\sigma^2} \\
&+ \frac12 \log\left[\operatorname{det}(Q(\theta,s))-\operatorname{det}(\tilde{Q}(\sigma,\theta,s))\right],
\end{align*}
where $\tilde{Q}(\sigma,\theta,s)= \sigma^{-2}S^TS + Q(\theta,s)$ and the posterior $u_n|y,\sigma^*,s^*,\theta^*$ is a Gaussian $\Nc(\mu,\Sigma)$, where 
\begin{align}
	\mu=[\sigma^*]^{-2}\tilde{Q}(\sigma^*,\theta^*,s^*)^{-1}S^Ty, \quad \Sigma= \tilde{Q}(\sigma^*,\theta^*,s^*)^{-1}. 
\end{align}
The predictive distribution is then the restriction of  $\Nc(\mu, \Sigma)$ to the unobserved nodes, denoted by $\pi_{\text{pred}}$, and techniques for computing partial inverse of a sparse matrix can be applied. 
  We numerically optimize $\log \pi(\sigma,s,\theta|y)$ with the \texttt{fminunc}   function in Matlab. 

\begin{table}[h!]
\centering
\begin{tabular}{ |c|c|c|c| } 
 \hline
   & RMSE  & CRPS  &  LS     \\
 \hline
 Stationary \& $s=2$  & 0.0394 & 0.0199 &  -578.9    \\ 
 \hline
 Stationary \&  Inferred $s$  & 0.0399& 0.0201& -521.7    \\ 
 \hline
 Nonstationary \& $s=2$  & 0.0408 & 0.0185 & -648.5  \\ 
 \hline
 Nonstationary \& Inferred $s$ & 0.0414 & 0.0186& -644.6    \\ 
 \hline
\end{tabular}
\caption{Comparison of the four models through RMSE, CRPS and LS.}
\label{table:ex2}
\end{table}

To perform inference, we first normalize the data $y$ as described in Remark \ref{rmk:MV} so that it has mean-zero and has magnitude at the same level of $u_n$ from the graph Mat\'ern field, in which case we are only interested in the variations of $u_n$. We then adopt a pseudo-crossvalidation by randomly selecting 90\% of the data as observations and make predictions for the remaining ones. The process is repeated 20 times and we evaluate the predictions through the root mean square error (RMSE), the continuous rank probability score (CRPS), and the logarithmic scoring rule (LS) as shown in Table \ref{table:ex2}. The three criteria are considered in \cite{fuglstad2015does, bolin2019rational} for similar comparisons,  with each defined as 
\begin{align*}
\text{RMSE}:& =  \frac{\|y_{\text{test}}-\mathbb{E} \pi_{\text{pred}}\|_2}{\sqrt{n_{\text{test}}}} ,\\
\text{CRPS}:& =  \frac{1}{n_{\text{test}}} \sum_{i=1}^{n_{\text{test}}} \text{crps}(\pi^{(i)}_{\text{pred}},y^{(i)}_{\text{test}}) \\
& =\frac{1}{n_{\text{test}}} \sum_{i=1}^{n_{\text{test}}} \int_{-\infty}^{\infty} \left[\pi_{\text{pred}}^{(i)}(y^{(i)}_{\text{test}})-\mathbf{1}\{y^{(i)}_{\text{test}}\leq t\}\right]^2 dt,\\
\text{LS}  :& =  -\log \pi_{\text{pred}}(y_{\text{test}}),
\end{align*}
where $y_{\text{test}}$ is the test data with size $n_{\text{test}}$ and $\pi_{\text{pred}}^{(i)}$ is the marginal distribution for each unobserved node. Since $\pi_{\text{pred}}$ is Gaussian (and hence its marginals), the above quantities are computable, and the CRPS can be calculated with its representation for Gaussians \cite{gneiting2005calibrated} :
\begin{align*}
	\text{crps}\left(\mathcal{N}(a,b^2),y \right)=  b\left\{ \frac{y-a}{b} \left[2\Phi\left(\frac{y-a}{b}\right)-1\right] +2 \phi\left(\frac{y-a}{b}\right) -\frac{1}{\sqrt{\pi}}\right\}, 
\end{align*}
where $\Phi$ and $\phi$ are the c.d.f. and p.d.f. of the standard normal respectively. 

We notice in Table \ref{table:ex2} that the nonstationary model improves CRPS and LS but not RMSE over the stationary one, as was observed in \cite{fuglstad2015does, bolin2019rational} using finite element representations of GFs. This suggests that adding nonstationarity improves the prediction in a distribution sense. However, inferring the smoothness $s$ from data appears not to improve the predictions, in contrast to the results in \cite{bolin2019rational}. We believe this is due to the different formulations that we are taking, where \cite{bolin2019rational} adopts an SPDE approach and a rational approximation for fractional smoothness. It is also possible that this is in general a feature of the graph representation as similar observations are made in Subsection \ref{sec:ML}. We shall leave more in-depth investigations for future studies.

\begin{remark}
Prediction at new locations can be done by augmenting the graph with new nodes and refitting the model. Interpolation is not readily applicable in this example since we assume to only be given pairwise distances between counties and to not have a continuous representation of the underlying manifold (here the map of the contiguous states of the U.S.). However, in examples where the underlying manifold is known and the graph nodes form a triangulation, we could construct a finite-element basis and carry out interpolation using the basis.
\end{remark}


\subsection{Application in Machine Learning} \label{sec:ML}
In this subsection we illustrate the use of graph Mat\'ern priors in a Bayesian formulation of semi-supervised binary classification \cite{bertozzi2018uncertainty}. We seek to classify images $\{x_i\}_{i=1}^n$ of two different digits of the MNIST dataset given $J\ll n$ noisy labels. Similarly as above, the problem is cast into a latent Gaussian model, where the labels are assumed to be a probit transform of the latent process $u_n$:
\begin{align*}
	y_i &= \operatorname{sign} \big(u_n(x_i) + \eta_i\big), \quad i=1, \ldots,J,
\end{align*}     
where $\eta_i \overset{i.i.d.}{\sim} \mathcal{N}(0,\sigma^2)$. The likelihood model associated with the above equation is 
\begin{align*}
	\pi(y|u_n,\sigma) = \sum_{i=1}^J  \Phi\left(\frac{y_iu_n(x_i)}{\sigma}\right),
\end{align*}
where $\Phi$ is the c.d.f. of the standard normal. As in Subsection \ref{sec:SS}, the latent process $u_n$ will be modeled as a graph Mat\'ern field in four different ways, by considering stationary/nonstationary length scale and fixed/inferred smoothness $s$. For the most general case, the latent process is modeled as 
\begin{align*}
	u_n|\theta,s &\sim \mathcal{N}(0,Q(\theta,s)^{-1}), \quad  Q(\theta,s) = \tau_n(\theta)^{\frac{m}{4}-\frac{s}{2}}(\tau_n(\theta)+\Delta_n)^{s}\tau_n(\theta)^{\frac{m}{4}-\frac{s}{2}},
\end{align*}
and $\tau_n$ is modeled as in \eqref{eq:hyperKL}:
\begin{align*}
	\log \tau_n(\theta)& = \nu^{\frac{s_0}{2}-\frac{m}{4}}\sum_{i=1}^{n_0} \left[\nu + \lambda_n^{(i)}\right]^{-\frac{s_0}{2}} \theta^{(i)} \psi_n^{(i)} 
\end{align*}
with standard normal hyperprior on each of the $\theta_{i}$'s and a log-normal prior for $s$. As  in Remark \ref{rmk:graphnormalization}, the effective dimension $m$ is about 4 and the other parameters are chosen as $s_0=4,\nu=0.1,n_0=10$, with marginal variance of $\tau_n$ tuned empirically. For this problem $\Delta_n$ is taken to be a symmetric $k$-nearest neighbor graph Laplacian $\Delta_n=I-D^{-1/2}WD^{-1/2}$, with self-tuning weights proposed by \cite{zelnik2005self}: 
\begin{align*}
	W_{ij}= 
	\exp \left( \frac{-|x_i-x_j|^2}{2\delta(i)\delta(j)}\right),
\end{align*}
where the images $x_i$'s are viewed as vectors in $\mathbb{R}^{784}$ and $\delta(i)$ is the Euclidean distance between $x_i$ and its $k$-th nearest neighbor. The sparsity of $\Delta_n$ follows from the $k$-nearest neighbor construction. 

Similarly as in Subsection \ref{sec:SS}, we adopt an evidence maximization approach for inferring the optimal hyperparameters, which are then used to find the MAP estimator for $u_n$. However, since the likelihood is non Gaussian, there is no closed form formula for the marginal posterior  of the hyperparameters $\pi(\sigma,\theta,s|y)$, and we then apply a Laplace approximation  \cite{rue2009approximate}. More precisely, denoting all the hyperparameters by $z$, $\pi(z|y)$ is approximated by 
\begin{align}
	\pi(z|y) \propto  \frac{\pi(u_n,z,y)}{\pi(u_n|z,y)} \bigg|_{u_n=u_n^*}  \approx  \frac{\pi(u_n,z,y)}{\tilde{\pi}(u_n|z,y)} \bigg|_{u_n=u_n^*},  \label{eq:pz}
\end{align}
where $u_n^*$ is the mode of $\pi(u_n|z,y)$  and  $\tilde{\pi}(u_n|z,y)$  is its Laplace approximation at $u_n^*$. The log density for $u_n|z,y$ has form 
\begin{align}
	\log \pi(u_n|z,y) \propto \sum_{i=1}^J \Phi\left(\frac{y_iu_n(x_i)}{\sigma}\right)- \frac12 u_n^T Q(\theta,s)u_n,  \label{eq:likelihood}
\end{align}
and the mode $u_n^*$ is found numerically with the Newton's method, where the gradient and Hessian of \eqref{eq:likelihood} are available analytically. 
The logarithm of the last expression in \eqref{eq:pz} is equal up to a constant to 
\begin{align*}
	\log \pi(\sigma,\theta,s) - \frac12 [u_n^*]^T Q(\theta,s)u_n^* +\frac12  \log \left[ \operatorname{det}\left(Q(\theta,s)\right)- \operatorname{det}(\tilde{Q}\left(\sigma,\theta,s)\right) \right]+ \sum_{i=1}^J \log \Phi\left(\frac{y_iu_n^*(x_i)}{\sigma}\right),  
\end{align*}
where $\tilde{Q}(\theta,s)= Q(\sigma,\theta,s)+H$ and $H$ is diagonal with entries 
\begin{align*}
	H_{ii}=\frac{y_iu_n^*(x_i)\phi\left(\sigma^{-1}u_n^*(x_i)\right)}{\sigma^2\Phi\left(\sigma^{-1}y_iu_n^*(x_i)\right)}+ \left[\frac{\phi\left(\sigma^{-1}u_n^*(x_i)\right)}{\sigma\Phi\left(\sigma^{-1}y_iu_n^*(x_i)\right)}\right]^2, \quad i=1,\ldots,J,
\end{align*}
and zero otherwise. 
The priors on $\sigma$ and $s$ are taken to be 
\begin{align*}
	\log \sigma & \sim \mathcal{N}\left( \log(0.1), 1  \right), \\
	\log s      & \sim \mathcal{N}\left( \log(4),   1  \right).
\end{align*}
Table \ref{table:Ex3} shows the classification error rates of the four different models for four pairs of digits, with $n=1000$ and $J=20$,  where each experiment is repeated 100 times.
We see that the nonstationary model improves slightly the performance while the model with inferred smoothness does the opposite. We believe this may be due to the fact that $s=4$ is already a near optimal choice for this problem, or it may also be an intrinsic characteristic of the graph representations as mentioned in Subsection \ref{sec:SS}. Table \ref{table:Ex3_2} shows the classification error rates for the inferred $s$ case when the prior is taken to be narrower:
\begin{align*}
	\log s \sim \mathcal{N}\left( \log 4, 0.01 \right),
\end{align*} 
in which case the nonstationary model with inferred $s$ also improves the prediction. 
\begin{table}[h!]
\centering
\begin{tabular}{ |c|c|c|c|c| } 
 \hline
   & 3\&8  & 5\&8 &  4\&9 & 7\&9   \\
 \hline
 Stationary \& $s=4$           &8.90\%  &8.64\% &17.67\% &10.13\%   \\ 
 \hline
 Stationary \& Inferred $s$    &9.61\%  &9.51\% &18.33\% &11.00\%  \\ 
 \hline
 Nonstationary \& $s=4$        &8.44\%  &7.23\% &17.38\% &9.37\%\\ 
 \hline
 Nonstationary \& Inferred $s$ &8.77\%  &8.12\% &19.92\% &10.70\%   \\ 
 \hline
\end{tabular}
\caption{Classification error rates with 2\% labeled data for different pairs of digits with hyperprior $\log s \sim \mathcal{N}\left( \log 4, 1 \right)$. }
\label{table:Ex3}
\end{table}

\begin{table}[h!]
\centering
\begin{tabular}{ |c|c|c|c|c| } 
 \hline
   & 3\&8  & 5\&8 &  4\&9 & 7\&9   \\
 \hline
 Stationary \& Inferred $s$    &8.91\% &8.66\%  &17.69\% &10.14\%  \\ 
 \hline
 Nonstationary \& Inferred $s$ &8.67\% &7.06\%  &17.54\% &9.71\%   \\ 
 \hline
\end{tabular}
\caption{Classification error rates with 2\% labeled data for different pairs of digits with hyperprior $\log s \sim \mathcal{N}\left( \log 4, 0.01 \right)$.}
\label{table:Ex3_2}
\end{table}

\begin{remark}
Label prediction at new images can be done by augmenting the graph with new nodes and refitting the model, or by using a $k$-NN interpolation as described in 
\cite{trillos2017consistency}.
\end{remark}

\section{Conclusions and Open Directions}\label{sec:conclusions and Open Directions}
This paper introduces graph representations of  Mat\'ern fields motivated by the SPDE approach. We have shown through rigorous analysis that graph  Mat\'ern fields approximate the Mat\'ern model under a manifold assumption, and we have established an explicit rate of convergence. We have emphasized that  graph Mat\'ern models can be used in a wide range of settings, as they generalize the Mat\'ern model to abstract point clouds beyond Euclidean or manifold settings. In addition, graph Mat\'ern fields are GMRFs and therefore numerical linear algebra techniques can be applied to gain speed up by exploiting sparsity.  


We have illustrated  through numerical examples the application of graph Mat\'ern fields in Bayesian inverse problems, spatial statistics and graph-based machine learning, bringing these fields together and transferring ideas among them. The graph Mat\'ern models can be directly implemented on the given point cloud, without any additional pre-processing such as adding nodes for triangulation for FEM methods. We demonstrate through comparisons certain benefits of the nonstationary models, where in particular nonstationary  improves classification accuracy. We believe  that  adding nonstationarity for graph-based learning problems had not been considered before and we hope that our empirical findings will stimulate further research in this area. 


The nonstationarity introduced through $\tau$ had been well-studied in Euclidean settings,  while comparatively less has been said about $\kappa$. We hope to investigate its modeling effects beyond the role as a length scale described in Remark \ref{rmk:tau&kappa} and to consider graph representations of anisotropic models where the Laplacian is replaced by $\nabla \cdot (H(x)\nabla)$. The case where $H$  is a constant positive definite matrix  can be easily dealt with by introducing a coordinate  transformation by $H^{-1/2}$. However, the general case  where $H$ is a function of the spatial variable is more involved and further research is needed. 
 
Another direction for further research is to investigate in more detail the case where the points are distributed according to a non-uniform density. As noted in Remark \ref{rmk:density}, one can normalize the weights to remove the effects of the density, aiming at recovering the Laplacian. A more interesting question is whether the density can be incorporated as part of the continuum operator that will lead to meaningful Mat\'ern type field. Especially for machine learning applications, we  wonder if the density of the point cloud can be exploited in the construction of prior distributions.

\section*{Acknowledgement}  
   DSA is thankful for the support of NSF and NGA through the grant DMS-2027056. 
The work of DSA was also partially supported by the NSF Grant DMS-1912818/1912802.

\bibliographystyle{abbrvnat}

\bibliography{references}

\begin{thebibliography}{104}
\providecommand{\natexlab}[1]{#1}
\providecommand{\url}[1]{\texttt{#1}}
\expandafter\ifx\csname urlstyle\endcsname\relax
  \providecommand{\doi}[1]{doi: #1}\else
  \providecommand{\doi}{doi: \begingroup \urlstyle{rm}\Url}\fi

\bibitem[Abrahamsen(1997)]{abrahamsen1997review}
P.~Abrahamsen.
\newblock {A review of Gaussian random fields and correlation functions}.
\newblock \emph{Norsk Regnesentral/Norwegian Computing Center Oslo}, 1997.

\bibitem[Adler(2010)]{adler2010geometry}
R.~J. Adler.
\newblock \emph{The Geometry of Random Fields}, volume~62.
\newblock SIAM, 2010.

\bibitem[Anderes and Stein(2008)]{anderes2008estimating}
E.~B. Anderes and M.~L. Stein.
\newblock {Estimating deformations of isotropic Gaussian random fields on the
  plane}.
\newblock \emph{The Annals of Statistics}, 36\penalty0 (2):\penalty0 719--741,
  2008.

\bibitem[Bakka et~al.(2018)Bakka, Rue, Fuglstad, Bolin, Illian, Krainski,
  Simpson, and Lindgren]{bakka2018spatial}
H.~Bakka, H.~Rue, A.~Fuglstad, G.-A.and~Riebler, D.~Bolin, J.~Illian,
  E.~Krainski, D.~Simpson, and F.~Lindgren.
\newblock {Spatial modeling with R-INLA: A review}.
\newblock \emph{Wiley Interdisciplinary Reviews: Computational Statistics},
  10\penalty0 (6):\penalty0 e1443, 2018.

\bibitem[Banerjee et~al.(2014)Banerjee, Carlin, and
  Gelfand]{banerjee2014hierarchical}
S.~Banerjee, B.~P. Carlin, and A.~E. Gelfand.
\newblock \emph{{Hierarchical Modeling and Analysis for Spatial Data}}.
\newblock Chapman and Hall/CRC, 2014.

\bibitem[Bardeen et~al.(1985)Bardeen, Szalay, Kaiser, and
  Bond]{bardeen1985statistics}
J.~M. Bardeen, A.~S. Szalay, N.~Kaiser, and J.~R. Bond.
\newblock {The statistics of peaks of Gaussian random fields}.
\newblock \emph{Astrophys. J.}, 304\penalty0 (FERMILAB-PUB-85-148-A):\penalty0
  15--61, 1985.

\bibitem[Bardsley and Kaipio(2013)]{bardsley2013gaussian}
J.~M. Bardsley and J.~Kaipio.
\newblock {Gaussian Markov random field priors for inverse problems}.
\newblock \emph{Inverse Problems \& Imaging}, 7\penalty0 (2), 2013.

\bibitem[Belkin and Niyogi(2004)]{belkin2004semi}
M.~Belkin and P.~Niyogi.
\newblock {Semi-supervised learning on Riemannian manifolds}.
\newblock \emph{Machine learning}, 56\penalty0 (1-3):\penalty0 209--239, 2004.

\bibitem[Belkin and Niyogi(2005)]{belkin2005towards}
M.~Belkin and P.~Niyogi.
\newblock {Towards a theoretical foundation for Laplacian-based manifold
  methods.}
\newblock In \emph{COLT}, volume 3559, pages 486--500. Springer, 2005.

\bibitem[Belkin and Niyogi(2007)]{belkin2007convergence}
M.~Belkin and P.~Niyogi.
\newblock Convergence of {L}aplacian eigenmaps.
\newblock \emph{Advances in Neural Information Processing Systems (NIPS)},
  19:\penalty0 129, 2007.

\bibitem[Belkin et~al.(2004)Belkin, Matveeva, and
  Niyogi]{belkin2004regularization}
M.~Belkin, I.~Matveeva, and P.~Niyogi.
\newblock Regularization and semi-supervised learning on large graphs.
\newblock In \emph{International Conference on Computational Learning Theory},
  pages 624--638. Springer, 2004.

\bibitem[Bertozzi et~al.(2018)Bertozzi, Luo, Stuart, and
  Zygalakis]{bertozzi2018uncertainty}
A.~L. Bertozzi, X.~Luo, A.~M. Stuart, and K.~C. Zygalakis.
\newblock Uncertainty quantification in graph-based classification of high
  dimensional data.
\newblock \emph{SIAM/ASA Journal on Uncertainty Quantification}, 6\penalty0
  (2):\penalty0 568--595, 2018.

\bibitem[Besag(1974)]{besag1974spatial}
J.~Besag.
\newblock Spatial interaction and the statistical analysis of lattice systems.
\newblock \emph{Journal of the Royal Statistical Society: Series B
  (Methodological)}, 36\penalty0 (2):\penalty0 192--225, 1974.

\bibitem[Besag(1975)]{besag1975statistical}
J.~Besag.
\newblock Statistical analysis of non-lattice data.
\newblock \emph{Journal of the Royal Statistical Society: Series D (The
  Statistician)}, 24\penalty0 (3):\penalty0 179--195, 1975.

\bibitem[Bigoni et~al.(2020)Bigoni, Chen, Trillos, Marzouk, and
  Sanz-Alonso]{bigoni2020data}
D.~Bigoni, Y.~Chen, N.~G. Trillos, Y.~Marzouk, and D.~Sanz-Alonso.
\newblock Data-driven forward discretizations for bayesian inversion.
\newblock \emph{Inverse Problems}, 36\penalty0 (10):\penalty0 105008, 2020.

\bibitem[Bogachev(1998)]{bogachev1998gaussian}
V.~I. Bogachev.
\newblock \emph{{Gaussian Measures}}.
\newblock Number~62. American Mathematical Soc., 1998.

\bibitem[Bolin(2014)]{bolin2014spatial}
D.~Bolin.
\newblock {Spatial Mat{\'e}rn fields driven by non-Gaussian noise}.
\newblock \emph{Scandinavian Journal of Statistics}, 41\penalty0 (3):\penalty0
  557--579, 2014.

\bibitem[Bolin and Kirchner(2019)]{bolin2019rational}
D.~Bolin and K.~Kirchner.
\newblock {The rational SPDE approach for Gaussian random fields with general
  smoothness}.
\newblock \emph{Journal of Computational and Graphical Statistics}, \penalty0
  (just-accepted):\penalty0 1--27, 2019.

\bibitem[Bolin and Lindgren(2011)]{bolin2011spatial}
D.~Bolin and F.~Lindgren.
\newblock Spatial models generated by nested stochastic partial differential
  equations, with an application to global ozone mapping.
\newblock \emph{The Annals of Applied Statistics}, 5\penalty0 (1):\penalty0
  523--550, 2011.

\bibitem[Bolin et~al.(2018)Bolin, Kirchner, and Kov{\'a}cs]{bolin2018weak}
D.~Bolin, K.~Kirchner, and M.~Kov{\'a}cs.
\newblock {Weak convergence of Galerkin approximations for fractional elliptic
  stochastic PDEs with spatial white noise}.
\newblock \emph{BIT Numerical Mathematics}, 58\penalty0 (4):\penalty0 881--906,
  2018.

\bibitem[Bolin et~al.(2020)Bolin, Kirchner, and Kov{\'a}cs]{bolin2020numerical}
D.~Bolin, K.~Kirchner, and M.~Kov{\'a}cs.
\newblock Numerical solution of fractional elliptic stochastic pdes with
  spatial white noise.
\newblock \emph{IMA Journal of Numerical Analysis}, 40\penalty0 (2):\penalty0
  1051--1073, 2020.

\bibitem[Brochu et~al.(2010)Brochu, Cora, and De~Freitas]{brochu2010tutorial}
E.~Brochu, V.~M. Cora, and N.~De~Freitas.
\newblock {A tutorial on Bayesian optimization of expensive cost functions,
  with application to active user modeling and hierarchical reinforcement
  learning}.
\newblock \emph{arXiv preprint arXiv:1012.2599}, 2010.

\bibitem[Bui-Thanh et~al.(2013)Bui-Thanh, Ghattas, Martin, and
  Stadler]{bui2013computational}
T.~Bui-Thanh, O.~Ghattas, J.~Martin, and G.~Stadler.
\newblock {A computational framework for infinite-dimensional Bayesian inverse
  problems Part I: The linearized case, with application to global seismic
  inversion}.
\newblock \emph{SIAM Journal on Scientific Computing}, 35\penalty0
  (6):\penalty0 A2494--A2523, 2013.

\bibitem[Burago et~al.(2015)Burago, Ivanov, and Kurylev]{burago2015graph}
D.~Burago, S.~Ivanov, and Y.~Kurylev.
\newblock A graph discretization of the laplace--beltrami operator.
\newblock \emph{Journal of Spectral Theory}, 4\penalty0 (4):\penalty0 675--714,
  2015.

\bibitem[Calder and Garcia~Trillos(2019)]{calder2019improved}
J.~Calder and N.~Garcia~Trillos.
\newblock {Improved spectral convergence rates for graph Laplacians on
  epsilon-graphs and k-NN graphs}.
\newblock \emph{arXiv preprint arXiv:1910.13476}, 2019.

\bibitem[Calvetti and Somersalo(2007)]{calvetti2007introduction}
D.~Calvetti and E.~Somersalo.
\newblock \emph{{An Introduction to Bayesian Scientific Computing: Ten Lectures
  on Subjective Computing}}, volume~2.
\newblock Springer Science \& Business Media, 2007.

\bibitem[Cameletti et~al.(2013)Cameletti, Lindgren, Simpson, and
  Rue]{cameletti2013spatio}
M.~Cameletti, F.~Lindgren, D.~Simpson, and H.~Rue.
\newblock {Spatio-temporal modeling of particulate matter concentration through
  the SPDE approach}.
\newblock \emph{AStA Advances in Statistical Analysis}, 97\penalty0
  (2):\penalty0 109--131, 2013.

\bibitem[Canzani(2013)]{canzani2013analysis}
Y.~Canzani.
\newblock {Analysis on manifolds via the Laplacian}.
\newblock \emph{Lecture Notes available at: http://www. math. harvard.
  edu/canzani/docs/Laplacian. pdf}, 2013.

\bibitem[Chung(1997)]{chung1997spectral}
F.~R.~K. Chung.
\newblock \emph{{Spectral Graph Theory}}.
\newblock Number~92. American Mathematical Soc., 1997.

\bibitem[Cohen et~al.(1991)Cohen, Fan, and Patel]{cohen1991classification}
F.~S. Cohen, Z.~Fan, and M.~A. Patel.
\newblock {Classification of rotated and scaled textured images using Gaussian
  Markov random field models}.
\newblock \emph{IEEE Transactions on Pattern Analysis \& Machine Intelligence},
  \penalty0 (2):\penalty0 192--202, 1991.

\bibitem[Dunlop and Stuart(2016)]{dunlop2016bayesian}
M.~M. Dunlop and A.~M. Stuart.
\newblock The bayesian formulation of eit: Analysis and algorithms.
\newblock \emph{Inverse Problems and Imaging}, 10\penalty0 (4):\penalty0
  1007--1036, 2016.

\bibitem[Dunlop et~al.(2017)Dunlop, Iglesias, and
  Stuart]{dunlop2017hierarchical}
M.~M. Dunlop, M.~A. Iglesias, and A.~M. Stuart.
\newblock {Hierarchical Bayesian level set inversion}.
\newblock \emph{Statistics and Computing}, 27\penalty0 (6):\penalty0
  1555--1584, 2017.

\bibitem[Frazier(2018)]{frazier2018tutorial}
P.~I. Frazier.
\newblock A tutorial on bayesian optimization.
\newblock \emph{stat}, 1050:\penalty0 8, 2018.

\bibitem[Fuglstad et~al.(2015{\natexlab{a}})Fuglstad, Lindgren, Simpson, and
  Rue]{fuglstad2015exploring}
G.-A. Fuglstad, F.~Lindgren, D.~Simpson, and H.~Rue.
\newblock {Exploring a new class of non-stationary spatial Gaussian random
  fields with varying local anisotropy}.
\newblock \emph{Statistica Sinica}, pages 115--133, 2015{\natexlab{a}}.

\bibitem[Fuglstad et~al.(2015{\natexlab{b}})Fuglstad, Simpson, Lindgren, and
  Rue]{fuglstad2015does}
G.-A. Fuglstad, D.~Simpson, F.~Lindgren, and H.~Rue.
\newblock Does non-stationary spatial data always require non-stationary random
  fields?
\newblock \emph{Spatial Statistics}, 14:\penalty0 505--531, 2015{\natexlab{b}}.

\bibitem[Fuglstad et~al.(2019)Fuglstad, Simpson, Lindgren, and
  Rue]{fuglstad2019constructing}
G.-A. Fuglstad, D.~Simpson, F.~Lindgren, and H.~Rue.
\newblock {Constructing priors that penalize the complexity of Gaussian random
  fields}.
\newblock \emph{Journal of the American Statistical Association}, 114\penalty0
  (525):\penalty0 445--452, 2019.

\bibitem[Garcia~Trillos and Sanz-Alonso(2017)]{trillos2016bayesian}
N.~Garcia~Trillos and D.~Sanz-Alonso.
\newblock {The Bayesian formulation and well-posedness of fractional elliptic
  inverse problems}.
\newblock \emph{Inverse Problems}, 33\penalty0 (6):\penalty0 065006, 2017.

\bibitem[Garcia~Trillos and Sanz-Alonso(2018)]{garcia2018continuum}
N.~Garcia~Trillos and D.~Sanz-Alonso.
\newblock {Continuum limits of posteriors in graph Bayesian inverse problems}.
\newblock \emph{SIAM Journal on Mathematical Analysis}, 50\penalty0
  (4):\penalty0 4020--4040, 2018.

\bibitem[Garc{\'\i}a~Trillos and Slep{\v{c}}ev(2016)]{trillos2016continuum}
N.~Garc{\'\i}a~Trillos and D.~Slep{\v{c}}ev.
\newblock Continuum limit of total variation on point clouds.
\newblock \emph{Archive for rational mechanics and analysis}, 220\penalty0
  (1):\penalty0 193--241, 2016.

\bibitem[Garc{\'\i}a~Trillos et~al.(2019)Garc{\'\i}a~Trillos, Gerlach, Hein,
  and Slep{\v{c}}ev]{trillos2019error}
N.~Garc{\'\i}a~Trillos, M.~Gerlach, M.~Hein, and D.~Slep{\v{c}}ev.
\newblock Error estimates for spectral convergence of the graph laplacian on
  random geometric graphs toward the laplace--beltrami operator.
\newblock \emph{Foundations of Computational Mathematics}, pages 1--61, 2019.

\bibitem[Garcia~Trillos et~al.(2019)Garcia~Trillos, Sanz-Alonso, and
  Yang]{ruiyilocalregularization}
N.~Garcia~Trillos, D.~Sanz-Alonso, and R.~Yang.
\newblock Local regularization of noisy point clouds: Improved global geometric
  estimates and data analysis.
\newblock \emph{Journal of Machine Learning Research}, 20\penalty0
  (136):\penalty0 1--37, 2019.
\newblock URL \url{http://jmlr.org/papers/v20/19-261.html}.

\bibitem[Garcia~Trillos et~al.(2020)Garcia~Trillos, Kaplan, Samakhoana, and
  Sanz-Alonso]{trillos2017consistency}
N.~Garcia~Trillos, Z.~Kaplan, T.~Samakhoana, and D.~Sanz-Alonso.
\newblock {On the consistency of graph-based Bayesian semi-supervised learning
  and the scalability of sampling algorithms}.
\newblock \emph{Journal of Machine Learning Research}, 21\penalty0
  (28):\penalty0 1--47, 2020.

\bibitem[Gelfand et~al.(2010)Gelfand, Diggle, Guttorp, and
  Fuentes]{gelfand2010handbook}
A.~E. Gelfand, P.~Diggle, P.~Guttorp, and M.~Fuentes.
\newblock \emph{Handbook of Spatial Statistics}.
\newblock CRC press, 2010.

\bibitem[Gin{\'e} and Koltchinskii(2006)]{GK}
E.~Gin{\'e} and V.~Koltchinskii.
\newblock Empirical graph {L}aplacian approximation of {L}aplace-{B}eltrami
  operators: large sample results.
\newblock In \emph{High dimensional probability}, volume~51 of \emph{IMS
  Lecture Notes Monogr. Ser.}, pages 238--259. Inst. Math. Statist., Beachwood,
  OH, 2006.
\newblock \doi{10.1214/074921706000000888}.
\newblock URL \url{http://dx.doi.org/10.1214/074921706000000888}.

\bibitem[Gneiting et~al.(2005)Gneiting, Raftery, Westveld~III, and
  Goldman]{gneiting2005calibrated}
T.~Gneiting, A.~E. Raftery, A.~H. Westveld~III, and T.~Goldman.
\newblock {Calibrated probabilistic forecasting using ensemble model output
  statistics and minimum CRPS estimation}.
\newblock \emph{Monthly Weather Review}, 133\penalty0 (5):\penalty0 1098--1118,
  2005.

\bibitem[Gramacy and Lee(2008)]{gramacy2008bayesian}
R.~B. Gramacy and H.~K.~H. Lee.
\newblock {Bayesian treed Gaussian process models with an application to
  computer modeling}.
\newblock \emph{Journal of the American Statistical Association}, 103\penalty0
  (483):\penalty0 1119--1130, 2008.

\bibitem[Guttorp and Gneiting(2006)]{guttorp2006studies}
P.~Guttorp and T.~Gneiting.
\newblock {Studies in the history of probability and statistics XLIX: On the
  M\'atern correlation family}.
\newblock \emph{Biometrika}, 93\penalty0 (4):\penalty0 989--995, 2006.

\bibitem[Harizanov et~al.(2018)Harizanov, Lazarov, Margenov, Marinov, and
  Vutov]{harizanov2018optimal}
S.~Harizanov, R.~Lazarov, S.~Margenov, P.~Marinov, and Y.~Vutov.
\newblock Optimal solvers for linear systems with fractional powers of sparse
  spd matrices.
\newblock \emph{Numerical Linear Algebra with Applications}, 25\penalty0
  (5):\penalty0 e2167, 2018.

\bibitem[Harlim et~al.(2020)Harlim, Sanz-Alonso, and Yang]{harlim2020kernel}
J.~Harlim, D.~Sanz-Alonso, and R.~Yang.
\newblock Kernel methods for bayesian elliptic inverse problems on manifolds.
\newblock \emph{SIAM/ASA Journal on Uncertainty Quantification}, 8\penalty0
  (4):\penalty0 1414--1445, 2020.

\bibitem[Heaton et~al.(2019)Heaton, Datta, Finley, Furrer, Guinness,
  Guhaniyogi, Gramacy, Hammerling, Katzfuss, et~al.]{heaton2019case}
M.~J. Heaton, A.~Datta, A.~O. Finley, R.~Furrer, J.~Guinness, F.~Guhaniyogi,
  R.and~Gerber, R.~B. Gramacy, D.~Hammerling, M.~Katzfuss, et~al.
\newblock A case study competition among methods for analyzing large spatial
  data.
\newblock \emph{Journal of Agricultural, Biological and Environmental
  Statistics}, 24\penalty0 (3):\penalty0 398--425, 2019.

\bibitem[Hein(2006)]{Hei2006}
M.~Hein.
\newblock Uniform convergence of adaptive graph-based regularization.
\newblock In G.~Lugosi and H.~U. Simon, editors, \emph{Proc. of the 19th Annual
  Conference on Learning Theory (COLT)}, pages 50--64. Springer, 2006.

\bibitem[Hein et~al.(2005)Hein, Audibert, and Von~Luxburg]{hein2005graphs}
M.~Hein, J.-Y. Audibert, and U.~Von~Luxburg.
\newblock {From graphs to manifolds--weak and strong pointwise consistency of
  graph Laplacians}.
\newblock In \emph{International Conference on Computational Learning Theory},
  pages 470--485. Springer, 2005.

\bibitem[Hennig et~al.(2015)Hennig, Osborne, and
  Girolami]{hennig2015probabilistic}
P.~Hennig, M.~A. Osborne, and M.~Girolami.
\newblock Probabilistic numerics and uncertainty in computations.
\newblock \emph{Proceedings of the Royal Society A: Mathematical, Physical and
  Engineering Sciences}, 471\penalty0 (2179):\penalty0 20150142, 2015.

\bibitem[Isaac et~al.(2015)Isaac, Petra, Stadler, and
  Ghattas]{isaac2015scalable}
T.~Isaac, N.~Petra, G.~Stadler, and O.~Ghattas.
\newblock {Scalable and efficient algorithms for the propagation of uncertainty
  from data through inference to prediction for large-scale problems, with
  application to flow of the Antarctic ice sheet}.
\newblock \emph{Journal of Computational Physics}, 296:\penalty0 348--368,
  2015.

\bibitem[Kaipo and Somersalo(2006)]{kaipio2006statistical}
J.~Kaipo and E.~Somersalo.
\newblock {Statistical and Computational Inverse Problems}.
\newblock \emph{{Springer Science \& Business Media}}, 160, 2006.

\bibitem[Kennedy and O'Hagan(2001)]{kennedy2001bayesian}
M.~C. Kennedy and A.~O'Hagan.
\newblock Bayesian calibration of computer models.
\newblock \emph{Journal of the Royal Statistical Society: Series B (Statistical
  Methodology)}, 63\penalty0 (3):\penalty0 425--464, 2001.

\bibitem[Kersting and Hennig(2016)]{kersting2016active}
H.~Kersting and P.~Hennig.
\newblock Active uncertainty calibration in bayesian ode solvers.
\newblock In \emph{Proceedings of the Thirty-Second Conference on Uncertainty
  in Artificial Intelligence}, pages 309--318, 2016.

\bibitem[Khristenko et~al.(2019)Khristenko, Scarabosio, Swierczynski, Ullmann,
  and Wohlmuth]{khristenko2019analysis}
U.~Khristenko, L.~Scarabosio, P.~Swierczynski, E.~Ullmann, and B.~Wohlmuth.
\newblock {Analysis of boundary effects on PDE-based sampling of
  Whittle--Mat\'ern Random Fields}.
\newblock \emph{SIAM/ASA Journal on Uncertainty Quantification}, 7\penalty0
  (3):\penalty0 948--974, 2019.

\bibitem[Kim et~al.(2005)Kim, Mallick, and Holmes]{kim2005analyzing}
H.-M. Kim, B.~K. Mallick, and C.~C. Holmes.
\newblock {Analyzing nonstationary spatial data using piecewise Gaussian
  processes}.
\newblock \emph{Journal of the American Statistical Association}, 100\penalty0
  (470):\penalty0 653--668, 2005.

\bibitem[Kondor and Lafferty(2002)]{kondor2002diffusion}
R.~I. Kondor and J.~Lafferty.
\newblock Diffusion kernels on graphs and other discrete structures.
\newblock In \emph{Proceedings of the 19th international conference on machine
  learning}, volume 2002, pages 315--322, 2002.

\bibitem[Li et~al.(2018)Li, Mark, Raskutti, and Willett]{li2018graph}
Y.~Li, B.~Mark, G.~Raskutti, and R.~Willett.
\newblock Graph-based regularization for regression problems with
  highly-correlated designs.
\newblock In \emph{2018 IEEE Global Conference on Signal and Information
  Processing (GlobalSIP)}, pages 740--742. IEEE, 2018.

\bibitem[Lindgren et~al.(2011)Lindgren, Rue, and
  Lindstr{\"o}m]{lindgren2011explicit}
F.~Lindgren, H.~Rue, and J.~Lindstr{\"o}m.
\newblock {An explicit link between Gaussian fields and Gaussian Markov random
  fields: the stochastic partial differential equation approach}.
\newblock \emph{Journal of the Royal Statistical Society: Series B (Statistical
  Methodology)}, 73\penalty0 (4):\penalty0 423--498, 2011.

\bibitem[Lischke et~al.(2020)Lischke, Pang, Gulian, Song, Glusa, Zheng, Mao,
  Cai, Meerschaert, Ainsworth, et~al.]{lischke2020fractional}
A.~Lischke, G.~Pang, M.~Gulian, F.~Song, C.~Glusa, X.~Zheng, Z.~Mao, W.~Cai,
  M.~M. Meerschaert, M.~Ainsworth, et~al.
\newblock What is the fractional laplacian? a comparative review with new
  results.
\newblock \emph{Journal of Computational Physics}, 404:\penalty0 109009, 2020.

\bibitem[Liu et~al.(2014)Liu, Chakraborty, Li, Liu, and
  Lozano]{liu2014bayesian}
F.~Liu, S.~Chakraborty, F.~Li, Y.~Liu, and A.~C. Lozano.
\newblock {Bayesian regularization via graph Laplacian}.
\newblock \emph{Bayesian Analysis}, 9\penalty0 (2):\penalty0 449--474, 2014.

\bibitem[MacKay()]{mackay1997gaussian}
D.~J.~C. MacKay.
\newblock Gaussian processes-a replacement for supervised neural networks?
\newblock \emph{NIPS tutorial}.

\bibitem[Martin and Simpson(2005)]{martin2005use}
J.~D. Martin and T.~W. Simpson.
\newblock Use of kriging models to approximate deterministic computer models.
\newblock \emph{AIAA journal}, 43\penalty0 (4):\penalty0 853--863, 2005.

\bibitem[Mat{\'e}rn(2013)]{matern2013spatial}
B.~Mat{\'e}rn.
\newblock \emph{Spatial Variation}, volume~36.
\newblock Second Edition, Springer Science \& Business Media, 2013.

\bibitem[Montagna and Tokdar(2016)]{montagna2016computer}
S.~Montagna and S.~T. Tokdar.
\newblock {Computer emulation with nonstationary Gaussian processes}.
\newblock \emph{SIAM/ASA Journal on Uncertainty Quantification}, 4\penalty0
  (1):\penalty0 26--47, 2016.

\bibitem[Monterrubio-G{\'o}mez et~al.(2020)Monterrubio-G{\'o}mez, Roininen,
  Wade, Damoulas, and Girolami]{monterrubio2020posterior}
K.~Monterrubio-G{\'o}mez, L.~Roininen, S.~Wade, T.~Damoulas, and M.~Girolami.
\newblock Posterior inference for sparse hierarchical non-stationary models.
\newblock \emph{Computational Statistics \& Data Analysis}, 148:\penalty0
  106954, 2020.

\bibitem[Ng et~al.(2018)Ng, Colombo, and Silva]{ng2018bayesian}
Y.~C. Ng, N.~Colombo, and R.~Silva.
\newblock {Bayesian semi-supervised learning with graph Gaussian processes}.
\newblock In \emph{Advances in Neural Information Processing Systems}, pages
  1683--1694, 2018.

\bibitem[Nicolaescu(2020)]{nicolaescu2020lectures}
L.~I. Nicolaescu.
\newblock \emph{Lectures on the Geometry of Manifolds}.
\newblock World Scientific, 2020.

\bibitem[Peyr{\'e} et~al.(2019)Peyr{\'e}, Cuturi,
  et~al.]{peyre2019computational}
G.~Peyr{\'e}, M.~Cuturi, et~al.
\newblock Computational optimal transport: With applications to data science.
\newblock \emph{Foundations and Trends{\textregistered} in Machine Learning},
  11\penalty0 (5-6):\penalty0 355--607, 2019.

\bibitem[Raissi et~al.(2017)Raissi, Perdikaris, and
  Karniadakis]{raissi2017machine}
M.~Raissi, P.~Perdikaris, and G.~E. Karniadakis.
\newblock {Machine learning of linear differential equations using Gaussian
  processes}.
\newblock \emph{Journal of Computational Physics}, 348:\penalty0 683--693,
  2017.

\bibitem[Raissi et~al.(2018)Raissi, Perdikaris, and
  Karniadakis]{raissi2018numerical}
M.~Raissi, P.~Perdikaris, and G.~E. Karniadakis.
\newblock {Numerical Gaussian processes for time-dependent and nonlinear
  partial differential equations}.
\newblock \emph{SIAM Journal on Scientific Computing}, 40\penalty0
  (1):\penalty0 A172--A198, 2018.

\bibitem[Roininen et~al.(2014)Roininen, Huttunen, and
  Lasanen]{roininen2014whittle}
L.~Roininen, J.~Huttunen, and S.~Lasanen.
\newblock {Whittle-Mat{\'e}rn priors for Bayesian statistical inversion with
  applications in electrical impedance tomography}.
\newblock \emph{Inverse Probl. Imaging}, 8\penalty0 (2):\penalty0 561, 2014.

\bibitem[Roininen et~al.(2019)Roininen, Girolami, Lasanen, and
  Markkanen]{roininen2019hyperpriors}
L.~Roininen, M.~Girolami, S.~Lasanen, and M.~Markkanen.
\newblock {Hyperpriors for Mat{\'e}rn fields with applications in Bayesian
  inversion}.
\newblock \emph{Inverse Problems \& Imaging}, 13\penalty0 (1):\penalty0 1--29,
  2019.

\bibitem[Rue and Held(2005)]{rue2005gaussian}
H.~Rue and L.~Held.
\newblock \emph{{Gaussian Markov Random Fields: Theory and Applications}}.
\newblock Chapman and Hall/CRC, 2005.

\bibitem[Rue et~al.(2009)Rue, Martino, and Chopin]{rue2009approximate}
H.~Rue, S.~Martino, and N.~Chopin.
\newblock {Approximate Bayesian inference for latent Gaussian models by using
  integrated nested Laplace approximations}.
\newblock \emph{Journal of the royal statistical society: Series b (statistical
  methodology)}, 71\penalty0 (2):\penalty0 319--392, 2009.

\bibitem[Sampson et~al.(2001)Sampson, Damian, and Guttorp]{sampson2001advances}
P.~D. Sampson, D.~Damian, and P.~Guttorp.
\newblock {Advances in modeling and inference for environmental processes with
  nonstationary spatial covariance}.
\newblock In \emph{GeoENV III—Geostatistics for Environmental Applications},
  pages 17--32. Springer, 2001.

\bibitem[Sanchez-Vila et~al.(2006)Sanchez-Vila, Guadagnini, and
  Carrera]{sanchez2006representative}
X.~Sanchez-Vila, A.~Guadagnini, and J.~Carrera.
\newblock Representative hydraulic conductivities in saturated groundwater
  flow.
\newblock \emph{Reviews of Geophysics}, 44\penalty0 (3), 2006.

\bibitem[Sanz-Alonso et~al.(2019)Sanz-Alonso, Stuart, and Taeb]{sanzstuarttaeb}
D.~Sanz-Alonso, A.~M. Stuart, and A.~Taeb.
\newblock {Inverse Problems and Data Assimilation}.
\newblock \emph{arXiv preprint arXiv:1810.06191}, 2019.

\bibitem[Seeger(2000)]{seeger2000relationships}
M.~Seeger.
\newblock {Relationships between Gaussian processes, support vector machines
  and smoothing splines}.
\newblock \emph{Machine Learning}, 2000.

\bibitem[Simpson et~al.(2012)Simpson, Lindgren, and Rue]{simpson2012think}
D.~Simpson, F.~Lindgren, and H.~Rue.
\newblock {Think continuous: Markovian Gaussian models in spatial statistics}.
\newblock \emph{Spatial Statistics}, 1:\penalty0 16--29, 2012.

\bibitem[Singer(2006)]{singer06}
A.~Singer.
\newblock From graph to manifold {L}aplacian: The convergence rate.
\newblock \emph{Applied and Computational Harmonic Analysis}, 21\penalty0
  (1):\penalty0 128--134, 2006.

\bibitem[Singer and Wu(2017)]{SinWu13}
A.~Singer and H.-T. Wu.
\newblock {Spectral convergence of the connection Laplacian from random
  samples}.
\newblock \emph{Information and Inference: A Journal of the IMA}, 6\penalty0
  (1):\penalty0 58--123, 2017.

\bibitem[Sollich(2002)]{sollich2002bayesian}
P.~Sollich.
\newblock {Bayesian methods for support vector machines: Evidence and
  predictive class probabilities}.
\newblock \emph{Machine learning}, 46\penalty0 (1-3):\penalty0 21--52, 2002.

\bibitem[Somersalo et~al.(1992)Somersalo, Cheney, and
  Isaacson]{somersalo1992existence}
E.~Somersalo, M.~Cheney, and D.~Isaacson.
\newblock Existence and uniqueness for electrode models for electric current
  computed tomography.
\newblock \emph{SIAM Journal on Applied Mathematics}, 52\penalty0 (4):\penalty0
  1023--1040, 1992.

\bibitem[Stathopoulos et~al.(2014)Stathopoulos, Jones, and
  Girolami]{stathopoulos2014bat}
V.~Stathopoulos, V.and Zamora-Gutierrez, K.~Jones, and M.~Girolami.
\newblock {Bat call identification with Gaussian process multinomial probit
  regression and a dynamic time warping kernel}.
\newblock In \emph{Artificial intelligence and statistics}, pages 913--921,
  2014.

\bibitem[Stein(2012)]{stein2012interpolation}
M.~L. Stein.
\newblock \emph{{Interpolation of Spatial Data: Some Theory for Kriging}}.
\newblock Springer Science \& Business Media, 2012.

\bibitem[Stuart(2010)]{AS10}
A.~M. Stuart.
\newblock Inverse problems: a {B}ayesian perspective.
\newblock \emph{Acta Numerica}, 19:\penalty0 451--559, 2010.

\bibitem[Stuart and Teckentrup(2018)]{stuart2018posterior}
A.~M. Stuart and A.~Teckentrup.
\newblock {Posterior consistency for Gaussian process approximations of
  Bayesian posterior distributions}.
\newblock \emph{Mathematics of Computation}, 87\penalty0 (310):\penalty0
  721--753, 2018.

\bibitem[Sullivan(2015)]{sullivan2015introduction}
T.~J. Sullivan.
\newblock \emph{Introduction to Uncertainty Quantification}, volume~63.
\newblock Springer, 2015.

\bibitem[Tao and Shi(2020)]{tao2020convergence}
W.~Tao and Z.~Shi.
\newblock Convergence of laplacian spectra from random samples.
\newblock \emph{Journal of Computational Mathematics}, 38\penalty0
  (6):\penalty0 952--984, 2020.

\bibitem[Taylor and Worsley(2007)]{taylor2007detecting}
J.~E. Taylor and K.~J. Worsley.
\newblock Detecting sparse signals in random fields, with an application to
  brain mapping.
\newblock \emph{Journal of the American Statistical Association}, 102\penalty0
  (479):\penalty0 913--928, 2007.

\bibitem[Ting et~al.(2010)Ting, Huang, and Jordan]{THJ}
D.~Ting, L.~Huang, and M.~I. Jordan.
\newblock An analysis of the convergence of graph {L}aplacians.
\newblock In \emph{Proc. of the 27th Int. Conference on Machine Learning
  (ICML)}, 2010.

\bibitem[van~der Vaart(2008)]{van2008rates}
J.~H. van~der Vaart, A. W.and van~Zanten.
\newblock {Rates of contraction of posterior distributions based on Gaussian
  process priors}.
\newblock \emph{The Annals of Statistics}, 36\penalty0 (3):\penalty0
  1435--1463, 2008.

\bibitem[Von~Luxburg(2007)]{von2007tutorial}
U.~Von~Luxburg.
\newblock A tutorial on spectral clustering.
\newblock \emph{Statistics and Computing}, 17\penalty0 (4):\penalty0 395--416,
  2007.

\bibitem[Whittle(1954)]{whittle1954stationary}
P.~Whittle.
\newblock On stationary processes in the plane.
\newblock \emph{Biometrika}, pages 434--449, 1954.

\bibitem[Whittle(1963)]{whittle1963stochastic}
P.~Whittle.
\newblock Stochastic-processes in several dimensions.
\newblock \emph{Bulletin of the International Statistical Institute},
  40\penalty0 (2):\penalty0 974--994, 1963.

\bibitem[Wiens et~al.(2020)Wiens, Nychka, and Kleiber]{wiens2020modeling}
A.~Wiens, D.~Nychka, and W.~Kleiber.
\newblock Modeling spatial data using local likelihood estimation and a
  mat{\'e}rn to spatial autoregressive translation.
\newblock \emph{Environmetrics}, 31\penalty0 (6):\penalty0 e2652, 2020.

\bibitem[Williams and Rasmussen(1996)]{williams1996gaussian}
C.~K.~I. Williams and C.~E. Rasmussen.
\newblock Gaussian processes for regression.
\newblock In \emph{Advances in neural information processing systems}, pages
  514--520, 1996.

\bibitem[Williams and Rasmussen(2006)]{williams2006gaussian}
C.~K.~I. Williams and C.~E. Rasmussen.
\newblock \emph{{Gaussian Processes for Machine Learning}}, volume~2.
\newblock MIT press Cambridge, MA, 2006.

\bibitem[Zelnik-Manor and Perona(2005)]{zelnik2005self}
L.~Zelnik-Manor and P.~Perona.
\newblock Self-tuning spectral clustering.
\newblock In \emph{Advances in neural information processing systems}, pages
  1601--1608, 2005.

\bibitem[Zhu et~al.(2003)Zhu, Lafferty, and Ghahramani]{zhu1965semi}
X.~Zhu, J.~Lafferty, and Z.~Ghahramani.
\newblock {Semi-supervised learning: from Gaussian fields to Gaussian
  processes}.
\newblock In \emph{School of CS, CMU}. Citeseer, 2003.

\end{thebibliography}

\newpage
\begin{appendix}

In this Appendix we prove the main theorems in Section \ref{sec:theory}. Section \ref{sec:APrelim} makes precise our  setting and assumptions, defines several quantities of interest and presents other necessary preliminaries. Sections  \ref{sec:AevalRate} and \ref{sec:AefunRate} contain proofs for the spectral convergence of $L^{\tau,\kappa}_n$ towards $\L^{\tau,\kappa}$. Finally, Section \ref{sec:APTL2} gives the proof of Theorem \ref{thm:Rate}. 
\section{Preliminaries} \label{sec:APrelim}

Suppose $\M$ is an $m$-dimensional smooth, connected, compact manifold without boundary embedded in $\mathbb{R}^d$, with the absolute value of sectional curvature bounded by $K$ and Riemannian metric inherited from $\mathbb{R}^d$.  Let $\{x_n\}_{n=1}^{\infty}$ be a sequence of independent samples from the uniform distribution $\gamma$ on $\M$,  i.e. the normalized volume measure.   Denote $\M_n:=\{x_1,\ldots,x_n\}$ and let $\gamma_n:=\frac{1}{n}\sum_{i=1}^n \delta_{x_i}$ be the empirical distribution of the point cloud $\M_n$. Throughout $\tau$ and $\kappa$ will denote Lipschitz continuous functions on $\M$ with $\kappa$  being continuously differentiable.  We assume that both functions are bounded from below by positive constants, so there exist $\alpha,\beta>0$ with $\frac{1}{\beta} \leq \tau \leq \beta$, $\frac{1}{\alpha} \leq \kappa \leq  \alpha$. We will analyze  the following operators 
\begin{align*}
    \L^{\tau,\kappa}&:=\tau I- \nabla \cdot(\kappa\nabla),  \\ 
    L^{\tau,\kappa}_n&:=\tau_n+\Delta^{\kappa}_n,
\end{align*}
where $\tau_n:=$ diag$\bigl(\tau(x_1),\ldots,\tau(x_n)\bigr)$ and differentiation is defined on the manifold (see e.g. \cite{nicolaescu2020lectures}).  Here $\Delta_n^{\kappa}=D-W \in \R^{n \times n}$ and the entries of $D$ and $W$ are given by
\begin{align*}
    W_{ij}&:=\frac{2(m+2)}{n\nu_m h_n^{m+2}} \mathbf{1}\{d_{\M}(x_i,x_j)<h_n\} \sqrt{\kappa(x_i)\kappa(x_j)}\,\,, \\
    D_{ii}&:=\sum_{j=1}^n W_{ij},
\end{align*}
where $\nu_m$ is the volume of the $m-$dimensional unit ball and $d_{\M}$ is the geodesic distance on $\M$. 
\begin{remark}
In this appendix we consider the weights $W_{ij}$ to be defined through the geodesic, rather than the Euclidean distance in \eqref{eq:defweights}. Since in small neighborhoods both distances agree up to a correction term that is of a higher order than our interest \cite{trillos2019error}, our results would also hold for the weights in \eqref{eq:defweights}. We choose to work with the geodesic distance to streamline our presentation. 
\end{remark}
As discussed in Section \ref{sec:theory}, the scaling of $h_n$ will be chosen so that 
\begin{align}
    \frac{(\log n)^{c_m}}{n^{1/m}} \ll h_n \ll \frac{1}{n^{1/s}}, \label{eq:hnScaling2}
\end{align}
where $c_m=3/4$ if $m=2$ and $c_m=1/m$ otherwise. We recall that the scaling of $h_n$ in \eqref{eq:hnScaling2} implies that the $\infty$-OT distance between $\gamma_n$ and $\gamma$ satisfies  $\eps_n = d_\infty(\gamma_n,\gamma)\ll h_n.$ In what follows we assume that we are in a realization where the conclusion of Proposition \ref{prop:transmap} holds and we let $\{T_n\}_{n=1}^\infty$ be a sequence of transport maps satisfying the bound \eqref{eq:transmap}.

We will use the following inner products and  induced norms on continuum and discrete spaces 
\begin{align*}
    \langle f,g\rangle_{L^2}:=\int f(x)g(x) d\gamma(x),\quad \langle f,g\rangle_{\tau}&:= \int f(x)g(x)\tau(x)d\gamma(x),\quad \langle f,g\rangle_{\kappa}:= \int f(x)g(x)\kappa(x)d\gamma(x), \\
    \langle v,w \rangle_2 := \frac{1}{n} \sum_{i=1}^n v(x_i)w(x_i),\quad  \langle v,w\rangle_{\tau_n}&:=\frac{1}{n} \sum_{i=1}^n v(x_i)w(x_i)\tau(x_i),\quad \langle v,w\rangle_{\kappa_n}:=\frac{1}{n} \sum_{i=1}^n v(x_i)w(x_i)\kappa(x_i).
\end{align*}
Notice that $\L^{\tau,\kappa}: \mathcal{D}(\L^{\tau,\kappa}) \subset L^2(\gamma) \to L^2(\gamma)$,  where $\mathcal{D}(\L^{\tau,\kappa})$ is the domain of definition,  is self-adjoint with respect to the $\langle \cdot,\cdot\rangle_{L^2}$ inner-product and has a compact resolvent; we will denote by $\{\lambda^{(k)}\}_{k = 1}^\infty$ and $\{\psi^{(k)}\}_{k = 1}^\infty$ its eigenvalues and eigenfunctions and recall that from standard theory 
$$\mathcal{D}(\L^{\tau,\kappa}) = \Bigl\{ f\in L^2(\gamma):  \sum_{k = 1}^\infty \left[\lambda^{(k)}\right]^2 \langle f, \psi^{(k)} \rangle_{L^2}^2 <\infty   \Bigr\}.$$
We refer to \cite{nicolaescu2020lectures}[Section 10.4.2] for more details. Similarly, $L^{\tau,\kappa}_n$ is self-adjoint with respect to $\langle\cdot,\cdot\rangle_2.$ 
By the minimax principle we can characterize the $k$-th smallest eigenvalues of $\L^{\tau,\kappa}$ and $L^{\tau,\kappa}_n$ by
\begin{align*}
	\lambda^{(k)}& = \underset{\mathcal{V}:\operatorname{dim}(\mathcal{V})=k}{\operatorname{min}} \,\,\underset{f\in \mathcal{V}\backslash 0} {\operatorname{max}} \frac{\langle f,\L^{\tau,\kappa}f\rangle_{L^2}}{\langle f, f\rangle_{L^2}},\\
	\lambda_{n}^{(k)}& = \underset{V:\operatorname{dim}(V)=k}{\operatorname{min}}\,\, \underset{v\in V\backslash 0} {\operatorname{max}} 
	\frac{\langle v, L^{\tau,\kappa}_nv\rangle_2}{\langle v, v\rangle_2}.
\end{align*}
We define the continuum and the discrete Dirichlet energies
\begin{align*}
    D[f]:=& \langle f,\L^{\tau,\kappa} f \rangle_{L^2}=\int \tau(x) f(x)^2 + \int \kappa(x) |\nabla f|^2=\|f\|^2_{\tau}+\|\nabla f\|_{\kappa}^2=:D^0[f]+D^1[f], \\
    D_{h_n}[v]:=& \langle v,L^{\tau,\kappa}_nv \rangle_2 =\frac{1}{n} \sum_{i=1}^n \tau(x_i)v(x_i)^2 + \frac{1}{n} \sum_{i=1}^n\sum_{j=1}^n  W_{ij}\left|v(x_i)-v(x_j)\right|^2
=:D_{h_n}^{0}[v] +D_{h_n}^1[v]. 
\end{align*}
For the lemmas and theorems below, we shall denote by $C_{\M}, C_{\M,\tau,\kappa}$ etc.  constants that  depend on the corresponding subscripts.  

\section{Convergence of Spectrum with Rate} \label{sec:AevalRate}
In this section we prove Theorem \ref{thm:evalRate} by establishing a lower and an upper bound on the eigenvalues of $L^{\tau,\kappa}_n$ in terms of those of $\L^{\tau,\kappa}$. With the above definitions of the Dirichlet energies, the eigenvalues have the characterizations 
\begin{align*}
	\lambda^{(k)}= \underset{\mathcal{V}:\operatorname{dim}(\mathcal{V})=k}{\operatorname{min}} \,\,\underset{f\in \mathcal{V}\backslash 0} {\operatorname{max}} \frac{D[f]}{\|f\|^2_{L^2}}, \quad \quad 
	\lambda_{n}^{(k)}= \underset{V:\operatorname{dim}(V)=k}{\operatorname{min}}\,\, \underset{v\in V\backslash 0} {\operatorname{max}} 
	\frac{D_{h_n}[v]}{ \|v\|_2^2}.
\end{align*}
 In order to compare the Dirichlet energies $D$ and $D_{h_n}$, we need an intermediate quantity defined by
\begin{align*}
    E_{r}[f]:=\int_{\M} \int_{B_{r}(x)} |f(x)-f(y)|^2 \sqrt{\kappa(x)\kappa(y)} d\gamma(y)d\gamma(x).
\end{align*}
Notice that $D_{h_n}^1$ can be seen as a finite sample approximation of $E_{h_n}$ up to a multiplicative constant. The following lemma, which can be proved with the same argument as  \cite{trillos2019error}[Lemma 5],  connects $E_{h_n}$ with $D^1$.   
\begin{lemma} \label{lemma:En}
For $f\in L^2(\gamma)$ and $r<2h_n,$  
\begin{align*}
    E_{r}[f] \leq (1+C_{\M,\kappa}h_n)\frac{\nu_mr^{m+2}}{m+2} D^1[f].
\end{align*}
\end{lemma}

\subsection{Upper Bound}
To start with, define the projection map $P: L^2(\gamma)\rightarrow L^2(\gamma_n)$ by 
\begin{align*}
    Pf(x_i):=n\int_{U_i} f(x) d\gamma(x),
\end{align*}
where $U_i=T_n^{-1}(\{x_i\})$ and $\{T_n\}_{n=1}^{\infty}$ is a sequence of transportation maps as in Proposition \ref{prop:transmap}. The $U_i$'s are called transportation cells. 

\begin{lemma}[Discrete Dirichlet Energy Upper Bound] \label{lemma:evalUB}
Let $f\in H^1(\gamma)$. 
\begin{enumerate}
    \item $\Big|\|Pf\|_{2}-\|f\|_{L^2} \Big| \leq C_{\M}\eps_n \|\nabla f\|_{L^2}$. 
    \item $D_{h_n}[Pf] \le  \left[1+C_{\M,\tau,\kappa}\left(h_n+\frac{\epsilon_n}{h_n}\right) \right]D[f].$
\end{enumerate}
\end{lemma}
\begin{proof}
The first statement is proved in \cite{burago2015graph}[Lemma 4.3(1)]. The second statement will be proved by combining upper bounds for $D^0_{h_n}[Pf]$ and $D^1_{h_n}[Pf].$ First, by H\"older's inequality and the fact that $\gamma(U_i)=\gamma_n(\{x_i\})=1/n,$ 
\begin{align}
	D^0_{h_n}[Pf] 
	= \frac{1}{n} \sum_{i=1}^n \tau(x_i) Pf(x_i)^2 &= \frac{1}{n} \sum_{i=1}^n \tau(x_i) n^2\left|\int_{U_i} f(x)d\gamma(x)\right|^2 \nonumber \\
	& \leq \sum_{i=1}^n  \tau(x_i)\int_{U_i} f(x)^2 d\gamma(x)\nonumber\\
	& \leq \left[1 + \text{Lip}(\tau)\beta \eps_n  \right]\sum_{i=1}^n \int_{U_i} \tau(x)f(x)^2d\gamma(x) \nonumber\\
	& = \left[1 + \text{Lip}(\tau)\beta \eps_n  \right]D^0[f].   \label{eq:U-D1}
\end{align}
For the upper bound on $D^1_{h_n}[Pf]$, observe that 
\begin{align*}
    Pf(x_i)-Pf(x_j)=n^2 \int_{U_i} \int_{U_j} f(y)-f(x) d\gamma(y)d\gamma(x),
\end{align*}
which implies
\begin{align*}
    |Pf(x_i)-Pf(x_j)|^2\leq n^2 \int_{U_i} \int_{U_j} |f(y)-f(x)|^2 d\gamma(y)d\gamma(x).
\end{align*}
By Lipschitz continuity of $\kappa$,  
\begin{align*}
    D_{h_n}^1[Pf]&=\frac{m+2}{n^2\nu_m h_n^{m+2}} \sum_{i=1}^n \sum_{j=1}^n \sqrt{\kappa(x_i)\kappa(x_j)}\mathbf{1}\{d_{\M}(x_i,x_j)<h_n\}|Pf(x_i)-Pf(x_j)|^2\\
    &\leq \frac{m+2}{\nu_m h_n^{m+2}} \sum_{i=1}^n \sum_{j=1}^n \sqrt{\kappa(x_i)\kappa(x_j)}\mathbf{1}\{d_{\M}(x_i,x_j)<h_n\} \int_{U_i} \int_{U_j} |f(y)-f(x)|^2 d\gamma(y)d\gamma(x) \\
	& \leq \left[1+\text{Lip}(\kappa)\alpha\eps_n\right] \frac{m+2}{\nu_m h_n^{m+2}} \sum_{i=1}^n \sum_{j=1}^n \mathbf{1}\{d_{\M}(x_i,x_j)<h_n\} \int_{U_i} \int_{U_j}|f(y)-f(x)|^2 \sqrt{\kappa(x)\kappa(y)} d\gamma(y)d\gamma(x) \\
    & \leq \left[1+\text{Lip}(\kappa)\alpha\eps_n\right] \frac{m+2}{\nu_m h_n^{m+2}}  \int_{\mathcal{M}} \int_{V(x)} |f(y)-f(x)|^2 \sqrt{\kappa(x)\kappa(y)} d\gamma(y)d\gamma(x),
\end{align*}
where if $x\in U_i$ then $V(x)=\bigcup _{j:j\sim i} U_j$, with $j\sim i$ meaning $d_{\M}(x_i,x_j)<h_n$.  In the second to last step we have used that 
\begin{align*}
    \left|\sqrt{\kappa(x_i)\kappa(x_j)}-\sqrt{\kappa(x)\kappa(y)}\right| &\leq \left|\sqrt{\kappa(x_i)\kappa(x_j)}-\sqrt{\kappa(x_i)\kappa(y)}\right| + \left|\sqrt{\kappa(x_i)\kappa(y)}-\sqrt{\kappa(x)\kappa(y)}\right| \\
    &\leq \sqrt{\alpha} \frac{|\kappa(\smash{x_j})-\kappa(y)|}{\sqrt{\kappa(\smash{x_j})}+\sqrt{\kappa(y)}} +\sqrt{\alpha} \frac{|\kappa(x_i)-\kappa(x)|}{\sqrt{\kappa(x_i)}+\sqrt{\kappa(x)}} 
    \leq \text{Lip}(\kappa)\alpha \eps_n.
\end{align*}  Notice that $V(x)\subset B_{h_n+2\epsilon_n}(x)$ and hence  
\begin{align}
    D_{h_n}^1[Pf]&\leq  \left[1+\text{Lip}(\kappa)\alpha\eps_n\right] \frac{m+2}{\nu_m h_n^{m+2}} \int_{\mathcal{M}} \int_{B_{h_n+2\epsilon_n}(x)} |f(y)-f(x)|^2 \sqrt{\kappa(x)\kappa(y)} d\gamma(y)d\gamma(x) \nonumber\\
    &= \left[1+\text{Lip}(\kappa)\alpha\eps_n\right]\frac{m+2}{\nu_m h_n^{m+2}}  E_{h_n+2\epsilon_n} [f] \nonumber\\ 
    & \leq \Big[1+\text{Lip}(\kappa)\alpha\eps_n\Big]\Big[1+C_{\M,\kappa}(h_n+2\eps_n)\Big] \left(\frac{h_n+2\epsilon_n}{h_n}\right)^{m+2}D^1[f]\nonumber\\
    & \leq \left[1+C_{\M,\kappa}\left(h_n+\frac{\epsilon_n}{h_n}\right) \right]D^1[f] \label{eq:U-D2},
\end{align}
where we have used Lemma \ref{lemma:En} and the assumption that $\eps_n\ll h_n$. The result follows by combining \eqref{eq:U-D1} and \eqref{eq:U-D2}. 
\end{proof}

\begin{corollary}[Upper Bound] \label{cor:evalUB}
Suppose $k:=k_n$ is such that $ \eps_n \sqrt{\lambda^{(k_n)}}\ll 1$ for $n$ large. Then  
\begin{align*}
    \lambda^{(k)}_{n} \leq \left[ 1+C_{\M,\tau,\kappa}\left(h_n+ \frac{\epsilon_n}{h_n} +\sqrt{\lambda^{(k)}} \epsilon_n\right) \right]\lambda^{(k)}.
\end{align*}
\end{corollary}

\begin{proof} 
Let $\mathcal{V}$ be the span of eigenfunctions $f_1,\ldots,f_k$ of $\L^{\tau,\kappa}$ associated with eigenvalues $\lambda^{(1)},\ldots,\lambda^{(k)}$. For $f\in \mathcal{V}$, we have 
\begin{align*}
\|\nabla f\|_{L^2} \leq \sqrt{\alpha} \|\nabla f\|_{\kappa} \leq \sqrt{\alpha} \sqrt{D[\smash{f}]} \leq \sqrt{\alpha \lambda^{(k)}} \|f\|_{L^2}. 
\end{align*}
Lemma \ref{lemma:evalUB} (1) then implies that 
\begin{align*}
    \|Pf\|_2 \geq \|f\|_{L^2}-C_{\M} \eps_n \|\nabla f\|_{L^2} \geq \|f\|_{L^2}-C_{\M,\kappa} \eps_n \sqrt{\lambda^{(k)}} \|f\|_{L^2}. 
\end{align*}
Therefore, the assumption that $\eps_n\sqrt{\lambda^{(k)}}\ll 1$  implies that $P|_{\mathcal{V}}$ is injective and $V=P(\mathcal{V})$ has dimension $k$. By Lemma \ref{lemma:evalUB} (2) we have
\begin{align*}
    \lambda^{(k)}_{n} \leq \underset{u\in V\backslash 0}{\operatorname{max}}\frac{D_{h_n}[v]}{\|v\|^2_{2}}=\underset{f\in \mathcal{V}\backslash 0}{\operatorname{max}}\frac{D_{h_n}[Pf]}{\|Pf\|^2_{2}}
    &\leq \underset{f\in \mathcal{V}\backslash 0}{\operatorname{max}}\frac{\left[1+C_{\M,\tau,\kappa}\left(h_n+\frac{\epsilon_n}{h_n}\right) \right]D[f]}{\left(1-C_{\M,\tau,\kappa} \eps_n \sqrt{\lambda^{(k)}}\right)\| f\|^2_{L^2}} \\
    & =\frac{\left[1+C_{\M,\tau,\kappa}\left(h_n+\frac{\epsilon_n}{h_n}\right) \right]\lambda^{(k)} }{1-C_{\M,\tau,\kappa} \eps_n \sqrt{\lambda^{(k)}}} \\
    & \leq \left[ 1+C_{\M,\tau,\kappa}\left( h_n+\frac{\epsilon_n}{h_n} +\sqrt{\lambda^{(k)}} \epsilon_n\right) \right]\lambda^{(k)}. 
\end{align*}
\end{proof}

\subsection{Lower Bound}
Define $P^*:L^2(\gamma_n)\rightarrow L^2(\gamma)$ by 
\begin{align*}
    P^* v  := \sum_{i=1}^n v(x_i)\mathbf{1}_{U_i},
\end{align*}
where $U_i$'s are the transportation cells as in the definition of $P$. We note that $P^*$ defines a piecewise constant interpolation map; for the subsequent analysis we need to introduce a smoothing operator  $\Lambda$ so that the map $\mathcal{I}:=\Lambda\circ P^*$ satisfies $\mathcal{I}v\in H^1(\gamma)$, the Sobolev space of order 1. We now detail the construction of the smoothing  operator.  
Let 
\begin{align*}
\psi(t) :=
\begin{cases}
\frac{m+2}{2\nu_m}(1-t^2), \quad & 0\leq t\leq 1, \\
0 , \quad & t>1.
\end{cases}
\end{align*}
Consider for $0< r<2h_n$ the kernel 
\begin{align*}
    k_r(x,y) :=r^{-m}\psi\left(\frac{d_{\M}(x,y)}{r}\right)
\end{align*}
and the associated integral operator
\begin{align*}
    \Lambda^0_r f= \int_{\mathcal{M}} k_r(x,y) f(y) d\gamma(y). 
\end{align*}
Let $\theta(x):=\Lambda_r^0 \mathbf{1}_{\mathcal{M}}=\int_{\mathcal{M}} k_r(x,y)d\gamma(y)$ and then define 
\begin{align*}
    \Lambda_rf := \theta^{-1} \Lambda_r^0 f,
\end{align*}
so that $\Lambda_r$ preserves constant functions. 
Finally we define the interpolation operator $\mathcal{I}: L^2(\gamma_n) \rightarrow L^2(\gamma)$ by 
\begin{align*}
    \mathcal{I}v := \Lambda_{h_n-2\eps_n} P^*v. 
\end{align*}
We next present some auxiliary bounds that will be needed later. 
\begin{lemma}[Auxiliary Bounds] \label{lemma:auxbounds}
For $f\in L^2(\gamma)$, we have  
\begin{align}
    \|\Lambda_r f\|_{\tau}^2 &\leq [1+\operatorname{Lip}(\tau)r\beta][1+CmKr^2]^2 \|f\|^2_{\tau},\label{eq:eq1} \\ 
     \|\Lambda_rf-f\|^2_{\kappa} &\leq [1+\operatorname{Lip}(\kappa)\alpha r]\frac{Cm}{\nu_mr^m} E_r[f], \label{eq:eq2} \\
     \|\nabla(\Lambda_rf)\|^2_{\kappa} &\leq [1+\operatorname{Lip}(\kappa)\alpha r][1+Cm^2Kr^2] \frac{m+2}{\nu_m r^{m+2}} E_r[f]. \label{eq:eq3}
\end{align}
\end{lemma}

\begin{proof}
The above results are proved in the same way as in \cite{burago2015graph}[Lemma 5.3, 5.4, 5.5] with little adjustments and the main differences are the additional factors $1+\text{Lip}(\tau)\beta r$ or $1+\text{Lip}(\kappa)\alpha r$.  To illustrate the idea, we will prove \eqref{eq:eq1} and the generalizations for \eqref{eq:eq2} and \eqref{eq:eq3} are similar. First by \cite{burago2015graph}[Lemma 5.1], we have for each $x\in \mathcal{M}$, 
\begin{align}
    (1+CmKr^2)^{-1} \leq \theta(x) \leq (1+CmKr^2),  \label{eq:boundTheta}
\end{align}
where recall that $K$ is an upper bound on the absolute value of the sectional curvature. 
Then we have 
\begin{align*}
    \left| \Lambda_r f \right|^2 = \theta^{-2}\left| \int_{\mathcal{M}} k_r(x,y)f(y)   \right|^2 &\leq \theta^{-2}\int_{\mathcal{M}} k_r(x,y)d\gamma(y) \int_{\mathcal{M}} k_r(x,y) |f(y)|^2 d\gamma(y) \\
    &=\theta^{-1} \int_{\mathcal{M}} k_r(x,y) |f(y)|^2 d\gamma(y) \\
    &\leq [1+CmKr^2] \int_{\mathcal{M}} k_r(x,y) |f(y)|^2 d\gamma(y).
\end{align*}
Noticing that $k_r(x,y)$ is zero when $d_{\M}(x,y)>r$ and $\tau$ is Lipschitz, we have 
\begin{align*}
    \|\Lambda_rf\|^2_{\tau} &\leq [1+CmKr^2] \int_{\mathcal{M}} \int_{\mathcal{M}} k_r(x,y) |f(y)|^2 \tau(x) d\gamma(x)d\gamma(y) \\
    & = [1+CmKr^2] \int_{\mathcal{M}} |f(y)|^2\left[\int_{B_r(y)} k_r(x,y)  \tau(x) d\gamma(x)\right]d\gamma(y) \\
    &\leq  [1+\text{Lip}(\tau)r\beta ][1+CmKr^2] \int_{\mathcal{M}} |f(y)|^2\left[\int_{B_r(y)} k_r(x,y)  \tau(y) d\gamma(x)\right]d\gamma(y)  \\
    &=[1+\text{Lip}(\tau)r\beta][1+CmKr^2] \int_{B_r(y)} k_r(x,y) d\gamma(x)\int_{\mathcal{M}} |f(y)|^2\ \tau(y)d\gamma(y)\\
    &= [1+\text{Lip}(\tau)r\beta][1+CmKr^2]^2 \|f\|^2_{\tau}. 
\end{align*}
\end{proof}

\begin{lemma}[Discrete Dirichlet Energy Lower Bound] \label{lemma:evalLB}
For each $v\in L^2(\gamma_n)$, 
\begin{enumerate}
    \item $\Big|\|\mathcal{I}v\|_{L^2}-\|v\|_{2} \Big|\leq C_{\kappa}h_n \sqrt{D_{h_n}[v]}$ .
    \item $D[\mathcal{I}v] \leq \left[1+C_{\M,\tau,\kappa}\left(h_n+\frac{\eps_n}{h_n} \right)\right]D_{h_n}[v]$.  
\end{enumerate}
\end{lemma}

\begin{proof}
1. By equation (6.4) in the proof of \cite{burago2015graph}[Lemma 6.2(1)], 
\begin{align}
    \|\mathcal{I}v-P^*v\|_{L^2} \leq Ch_n \|\delta v\|, \label{eq:IuP*u}
\end{align}
where 
\begin{align}
    \|\delta v\|^2= \frac{m+2}{\nu_mn^2h_n^{m+2}} \sum_{i=1}^n\sum_{j=1}^n \mathbf{1}\{d_{\M}(x_i,x_j)<h_n\} |v(x_i)-v(x_j)|^2. \label{eq:deltau}
\end{align}
The result follows by noticing that $\|\delta v\|\leq \sqrt{\alpha D_{h_n}[v]}$ and 
\begin{align*}
    \Big| \|\mathcal{I}v\|_{L^2} -\|v\|_{2} \Big|= \Big|\|\mathcal{I}v\|_{L^2}-\|P^*v\|_{L^2} \Big| \leq \|\mathcal{I}v-P^*v\|_{L^2}.
\end{align*}
\par

2. The second statement will be proved by combining lower bounds for $D^0_{h_n}[\mathcal{I}v]$ and $D^1_{h_n}[\mathcal{I}v].$ 
For the lower bound on $D_{h_n}^0[\mathcal{I}v]$, we have by \eqref{eq:eq1}, 
\begin{align*}
    D^0[\mathcal{I}v]=\|\Lambda_{h_n-2\eps_n}P^*v(x)\|^2_{\tau}  \leq \Big[1+\LT\beta(h_n-2\eps_n)\Big] \left[1+CmK(h_n-2\eps_n)^2\right]^2 \int \tau(x) |P^*v(x)|^2d\gamma(x).
\end{align*}
We also have
\begin{align*}
    \int \tau(x)|P^*v(x)|^2d\gamma(x)& =\sum_{i=1}^n\int_{U_i} \tau(x) |v(x_i)|^2 d\gamma(x)\\
    &\leq [1+\LT\beta \eps_n]\sum_{i=1}^n \int_{U_i} \tau(x_i) v(x_i)^2 d\gamma(x)=[1+\LT\beta\eps_n  ]D^0_{h_n}[v].
\end{align*}
Therefore
\begin{align}
    D^0[\mathcal{I}v]\leq \left(1+C_{\M,\tau}h_n\right)D^0_{h_n}[v]. \label{eq:L-D1}
\end{align}
Next we seek a lower bound for $D^1_{h_n}[\mathcal{I}v].$ 
We have 
\begin{align*}
    D^1_{h_n}[v] &= \frac{m+2}{n^2\nu_m h_n^{m+2}} \sum_{i=1}^n \sum_{j=1}^n \sqrt{\kappa(x_i)\kappa(x_j)} \mathbf{1}\{d_{\M}(x_i,x_j)<h_n\} |v(x_i)-v(x_j)|^2 \\
    &=\frac{m+2}{n^2\nu_m h_n^{m+2}} \sum_{i=1}^n \sum_{j=1}^n \sqrt{\kappa(x_i)\kappa(x_j)} \mathbf{1}\{d_{\M}(x_i,x_j)<h_n\} |P^* v(x_i)-P^*v(x_j)|^2 \\
    &=\frac{m+2}{\nu_m h_n^{m+2}} \sum_{i=1}^n \sum_{j=1}^n \sqrt{\kappa(x_i)\kappa(x_j)} \mathbf{1}\{d_{\M}(x_i,x_j)<h_n\} \int_{U_i}\int_{U_j}|P^*v(x)-P^*v(y)|^2 d\gamma(y)d\gamma(x) \\
    & \geq [1-\text{Lip}(\kappa)\alpha \eps_n]\frac{m+2}{\nu_m h_n^{m+2}} \sum_{i=1}^n \sum_{j=1}^n  \mathbf{1}\{d_{\M}(x_i,x_j)<h_n\} \int_{U_i}\int_{U_j}|P^*v(x)-P^*v(y)|^2 \sqrt{\kappa(x)\kappa(y)}d\gamma(y)d\gamma(x) \\
    &=[1-\text{Lip}(\kappa)\alpha \eps_n]\frac{m+2}{\nu_m h_n^{m+2}}  \int_{\mathcal{M}} \int_{V(x)} |P^*v(x)-P^*v(y)|^2 \sqrt{\kappa(x)\kappa(y)}d\gamma(y)d\gamma(x),
\end{align*}
where if $x\in U_i$ then $V(x)=\bigcup _{j:j\sim i} U_j$. Notice that $V(x)\supset B_{h_n-2\epsilon_n}(x)$ and hence,
\begin{align}
    D^1_{h_n}[v] &\geq [1-\text{Lip}(\kappa)\alpha \eps_n]\frac{m+2}{\nu_m h_n^{m+2}}  \int_{\mathcal{M}} \int_{B_{h_n-2\epsilon_n}} |P^*v(x)-P^*v(y)|^2 \sqrt{\kappa(x)\kappa(y)}d\gamma(y)d\gamma(x)\nonumber \\
     &= [1-\text{Lip}(\kappa)\alpha \eps_n]\frac{m+2}{\nu_m h_n^{m+2}}  E_{h_n-2\epsilon_n}[P^*v]. \label{eq:boundDh&E} 
\end{align}
Combining inequality \eqref{eq:eq3} with \eqref{eq:boundDh&E} gives  
\begin{align}
    D^1[\mathcal{I}v]&=\|\nabla (\mathcal{I}v)\|^2_{\kappa}\nonumber \\ &=\|\nabla(\Lambda_{h_n-2\epsilon_n}P^*v)\|^2_{\kappa}\nonumber\\
    & \le   \left[1+\LK\alpha(h_n-2\epsilon_n)\right]\left[1+Cm^2K(h_n-2\epsilon_n)^2\right]^2\frac{m+2}{\nu_m(h_n-\eps_n)^{m+2}} E_{h_n-2\epsilon_n}[P^*v]    \nonumber  \\
    &\leq \left[1+\LK\alpha(h_n-2\epsilon_n)\right]\left[1+Cm^2K(h_n-2\epsilon_n)^2\right]^2\left(\frac{h_n}{h_n-2\epsilon_n}\right)^{m+2} [1-\text{Lip}(\kappa)\alpha \eps_n]^{-1}D^1_h[v] \nonumber\\
    & \leq \left[1+C_{\M,\kappa}\left(h_n+\frac{\epsilon_n}{h_n}\right)\right]D^1_h[v] \label{eq:L-D2},
\end{align}
where we have used Lemma \ref{lemma:En} and the fact that $\eps_n\ll h_n$.
The second statement follows by combining \eqref{eq:L-D1} and \eqref{eq:L-D2}. 
\end{proof}

\begin{corollary}[Lower Bound] \label{cor:evalLB}
Suppose $k:=k_n$ is such that $h_n\sqrt{\lambda^{(k_n)}}\ll 1$ for $n$ large. Then 
\begin{align*}
    \lambda^{(k)}_{n} \geq \left[1-C_{\M,\tau,\kappa}\left(\frac{\epsilon_n}{h_n}+h_n\sqrt{\lambda^{(k)}}\right)\right] \lambda^{(k)}.
\end{align*}
\end{corollary}
\begin{proof}
Since we are interested in proving lower bounds for $\lambda^{(k)}_{n}$,   the result is trivial if $\lambda_n^{(k)} \geq \lambda^{(k)} $. Therefore we shall assume that $\lambda^{(k)}_n<\lambda^{(k)}$. Let $V$ be the span of eigenvectors $v_1,\ldots,v_k$ of $L^{\tau,\kappa}_n$ associated with eigenvalues $\lambda_n^{(1)},\ldots,\lambda^{(k)}_{n}$. Lemma \ref{lemma:evalLB}(1) implies  for $v\in V$
\begin{align*}
    \|\mathcal{I}v\|_{L^2} \geq \|v\|_2-C_{\M,\kappa}h_n\sqrt{D_{h_n}[v]}  
    \geq \left[1-C_{\M,\kappa}h_n\sqrt{\lambda^{(k)}_{n}}\right] \|v\|_{2}
    \geq \left[1-C_{\M,\kappa}h_n\sqrt{\lambda^{(k)}}\right] \|v\|_{2}. 
\end{align*}
Therefore, the assumption that $h_n\sqrt{\lambda^{(k)}}\ll 1 $ implies that  $\mathcal{I}|_V$ is injective and $\mathcal{V}=\mathcal{I}(V)$ has dimension $k$. Lemma \ref{lemma:evalLB} then gives 
\begin{align*}
    \lambda^{(k)} \leq \underset{f\in \mathcal{V}\backslash0}{\operatorname{max}}\frac{D[f]}{\|f\|^2_{L^2}}=\underset{v\in V\backslash0}{\operatorname{max}}\frac{D[\mathcal{I}v]}{\|\mathcal{I}v\|^2_{L^2}} 
    &\leq \underset{v\in V\backslash0}{\operatorname{max}}\frac{\left[1+C_{\M,\tau,\kappa}\left(h_n+\frac{\eps_n}{h_n} \right)\right]D_{h_n}[u]}{\left(1-C_{\M,\kappa}h_n\sqrt{\lambda^{(k)}}\right)^2 \|u\|_{2}} \\
    & = \underset{v\in V\backslash0}{\operatorname{max}} \frac{\left[1+C_{\M,\tau,\kappa}\left(h_n+\frac{\eps_n}{h_n} \right)\right]}{\left(1-C_{\M,\kappa}h_n\sqrt{\lambda^{(k)}}\right)^2 }\,\,  \lambda^{(k)}_{n}\\
    & \leq \left[1+C_{\M,\tau,\kappa}\left(\frac{\epsilon_n}{h_n}+h_n\sqrt{\lambda^{(k)}}\right)\right] \lambda^{(k)}_{n}.
\end{align*}
Therefore 
\begin{align*}
    \lambda^{(k)}_{n}\geq \left[1+C_{\M,\tau,\kappa}\left(\frac{\epsilon_n}{h_n}+h_n\sqrt{\lambda^{(k)}}\right)\right]^{-1} \lambda^{(k)} \geq \left[1-C_{\M,\tau,\kappa}\left(\frac{\epsilon_n}{h_n}+h_n\sqrt{\lambda^{(k)}}\right)\right] \lambda^{(k)}.
\end{align*}
\end{proof}  
Combining Lemma \ref{lemma:evalUB} and \ref{lemma:evalLB} we have: 
\begin{theorem}\label{thm:evalRate2}
Suppose  $k:=k_{n}$ is such that $h_n\sqrt{\lambda^{(k_n)}}\ll 1$ for $n$ large. Then
\begin{align*}
    \frac{|\lambda^{(k)}_{n}-\lambda^{(k)}|}{\lambda^{(k)}} \leq C_{\M,\tau,\kappa}\left[\frac{\eps_n}{h_n}+h_n\sqrt{\lambda^{(k)}}\right],
\end{align*}
where $C_{\M,\tau,\kappa}$ is a constant depending on $\M,\tau,\kappa$. 
\end{theorem}

\section{Convergence of Eigenfunctions with Rate} \label{sec:AefunRate}
In this section we prove Theorem \ref{thm:efunRate}.  Before we proceed, we introduce some additional notations. For any interval $J$ of $\mathbb{R}$, denote $H_{J}(\gamma)$ the subspace of $H^1(\gamma)$ that is spanned by eigenfunctions of $\L^{\tau,\kappa}$ associated with eigenvalues in $J$ and $\mathbb{P}_J(\gamma)$ the orthogonal projection from $L^2(\gamma)$ onto $H_J(\gamma)$. Similarly we use the notation $H_{J}(\gamma_n)$ and $\mathbb{P}_J(\gamma_n)$ for $L^{\tau,\kappa}_n$. To ease notation, we will denote $H_{(-\infty,\lambda)}(\gamma)$ and $H_{(-\infty,\lambda)}(\gamma_n)$ as  $H_{\lambda}(\gamma)$ and $H_{\lambda}(\gamma_n),$ respectively. We shall also denote both projections as  $\mathbb{P}_J$ when no confusion arises. 

To start with, we need several auxiliary results. 

\begin{lemma} \label{lemma:efun0}
Let $v\in H_{\lambda}(\gamma_n)$. 
\begin{enumerate}
    \item $\|P\mathcal{I}v-v\|_2\leq C_{\kappa} h_n \sqrt{D_{h_n}[v]} .$
    \item $D[\mathcal{I}v] \geq \left[1-C_{\M,\tau,\kappa}\left(\frac{\eps_n}{h_n}+h_n\sqrt{\lambda}\right)\right]D_{h_n}[v].$ 
\end{enumerate}
\end{lemma}
\begin{proof}
1. By \cite{burago2015graph}[Lemma 6.4(2)], 
\begin{align}
    \|P\mathcal{I}v-v\|_2 \leq Ch_n\|\delta v\|, \label{eq:PIu-u}
\end{align}
where  $\|\delta v\| $ is defined in \eqref{eq:deltau}
The result follows by noticing that $\|\delta v\|\leq \sqrt{\alpha D_{h_n}[v]}$.\par

2. We first bound $D_{h_n}[P\mathcal{I}v]$ in terms of $D_{h_n}[v]$. Denoting $P\mathcal{I}v$ as $w$, we have 
\begin{align*}
    D_{h_n}[w]=\big\langle w, L^{\tau,\kappa}_n w\big\rangle_{2}
    &= \big\langle w-\mathbb{P}_{\lambda}w +\mathbb{P}_{\lambda}w, L^{\tau,\kappa}_n (w-\mathbb{P}_{\lambda}w +\mathbb{P}_{\lambda}w) \big\rangle_{2} \\
    &= \big\langle w-\mathbb{P}_{\lambda}w,  L^{\tau,\kappa}_n(w-\mathbb{P}_{\lambda}w) \big\rangle_2+ \big\langle \mathbb{P}_{\lambda}w, L^{\tau,\kappa}_n \mathbb{P}_{\lambda}w \big\rangle_2+\big\langle w-\mathbb{P}_{\lambda}w, L^{\tau,\kappa}_n\mathbb{P}_{\lambda}w\big\rangle_2\\
    &=\big\langle w-\mathbb{P}_{\lambda}w,  L^{\tau,\kappa}_n(w-\mathbb{P}_{\lambda}w) \big\rangle_2+ \big\langle \mathbb{P}_{\lambda}w, L^{\tau,\kappa}_n \mathbb{P}_{\lambda}w \big\rangle_2 \\
    &\geq \big\langle \mathbb{P}_{\lambda}w, L^{\tau,\kappa}_n \mathbb{P}_{\lambda}w \big\rangle_2,
\end{align*}
where we have used that $\mathbb{P}_{\lambda}w$ and  $L^{\tau,\kappa}_n\mathbb{P}_{\lambda}w \in H_{\lambda}(\gamma_n)$ are orthogonal to $w-\mathbb{P}_{\lambda}w$. Since $L^{\tau,\kappa}_n$ is nonsingular, $\langle \cdot, L^{\tau,\kappa}_n \cdot \rangle_{2}$ defines an inner product \ and the triangle inequality implies 
\begin{align*}
    \sqrt{D_{h_n}[w]} \geq \sqrt{\big\langle \mathbb{P}_{\lambda}w, L^{\tau,\kappa}_n \mathbb{P}_{\lambda}w \big\rangle_2} \geq  \sqrt{\big\langle v, L^{\tau,\kappa}_nv \rangle_{2}} -\sqrt{\langle v-\mathbb{P}_{\lambda}w, L^{\tau,\kappa}_n(v-\mathbb{P}_{\lambda}w) \big\rangle_2}. 
\end{align*}
Now we bound the second term above. Since $v\in H_{\lambda}$, we have $v=\mathbb{P}_{\lambda}v$ and
\begin{align*}
    \big\langle v-\mathbb{P}_{\lambda}w, L^{\tau,\kappa}_n(v-\mathbb{P}_{\lambda}w) \big\rangle_2
    &=\big\langle \mathbb{P}_{\lambda}(v-w), L^{\tau,\kappa}_n\mathbb{P}_{\lambda}(v-w) \big\rangle_2  \\
    &\leq \lambda \| \mathbb{P}_{\lambda}(v-w)\|_2^2  
    \leq \lambda \|v-w\|_2^2  \leq C_{\kappa} \lambda h_n^2 D_{h_n}[v],
\end{align*}
where the last step follows from \eqref{eq:PIu-u}. Hence
\begin{align*}
    \sqrt{D_{h_n}[P\mathcal{I}v]} \geq \left[1-C_{\kappa}\sqrt{\lambda }h_n\right] \sqrt{D_{h_n}[v]},
\end{align*}
and the result follows from Lemma \ref{lemma:evalUB}, which says
\begin{align*}
    D_{h_n}[P\mathcal{I}v] \leq \left[1+C_{\M,\tau,\kappa}\left(h_n+\frac{\eps_n}{h_n}\right)\right]D[\mathcal{I}v]. 
\end{align*}
\end{proof}

We fix orthonormal eigenfunctions $\{v_k\}_{k=1}^n$ and $\{f_k\}_{k=1}^{\infty}$ for $L^{\tau,\kappa}_n$ and $\L^{\tau,\kappa}$. The following lemma bounds the projection error when $J$ is a half-interval. 
\begin{lemma} \label{lemma:efun1}
Suppose $k:=k_n$ is such that $h_n\sqrt{\lambda^{(k_n)}}\ll 1$.  Then, for any $a>0,$
\begin{align*}
    \|\mathcal{I}v_k-\mathbb{P}_{\lambda^{(k)}+a}\mathcal{I}v_k\|_{L^2}^2 &\leq C_{\M,\tau,\kappa} a^{-1}k\lambda^{(k)}\left(\frac{\eps_n}{h_n}+h_n\sqrt{\lambda^{(k)}}\right), \\
    D[\mathcal{I}v_k-\mathbb{P}_{\lambda^{(k)}+a}\mathcal{I}v_k] &\leq C_{\M,\tau,\kappa} a^{-1}(\lambda^{(k)}+a)k\lambda^{(k)}\left(\frac{\eps_n}{h_n}+h_n\sqrt{\lambda^{(k)}}\right).  
\end{align*}
\end{lemma}
\begin{proof}
Let $V$ be the span of $v_1,\ldots,v_k$ and $\mathcal{V}=\mathcal{I}(V)$. Since $h_n\sqrt{\lambda^{(k)}}\ll 1$, Theorem \ref{thm:evalRate} implies $\lambda^{(k)}_{n} \leq C\lambda^{(k)}$ and then by Lemma \ref{lemma:evalLB}, for any $v\in V$,
\begin{align*}
    \|\mathcal{I}v\|_{L^2}&\geq \left[1-C_{\M,\tau,\kappa}h_n\sqrt{\lambda^{(k)}_{n}}  \right] \|v\|_2\geq \left[1-C_{\M,\tau,\kappa}h_n\sqrt{\lambda^{(k)}}  \right] \|v\|_2\\
    D[\mathcal{I}v] & \leq \left[1+C_{\M,\tau,\kappa}\left(h_n+\frac{\eps_n}{h_n}\right)\right] D_{h_n}[v].
\end{align*}
The assumption $h_n\sqrt{\lambda^{(k)}}\ll 1$ also implies that $\mathcal{I}|_V$ is injective and $\mathcal{V}$ is $k$-dimensional. 
Let $\lambda_{\mathcal{V}}^{(1)},\ldots,\lambda_{\mathcal{V}}^{(k)}$ be the eigenvalues of $A:=\L^{\tau,\kappa}|_{\mathcal{V}}$.  The minimax principle implies that, for $j\leq k$, 
\begin{align}
    \lambda_{\mathcal{V}}^{(j)} \leq \frac{\left[1+C_{\M,\tau,\kappa}\left(h_n+\frac{\eps_n}{h_n}\right)\right]}{\left[1-C_{\M,\tau,\kappa}h_n\sqrt{\lambda^{(k)}}  \right]} \lambda_n^{(j)} 
    \leq \lambda_n^{(j)} + C_{\M,\tau,\kappa}\lambda^{(k)} \left(\frac{\eps_n}{h_n}+h_n\sqrt{\lambda^{(k)}}\right).  \label{eq:lambdaH}
\end{align}
Define another operator $\tilde{\L}^{\tau,\kappa}$  by  
\begin{align*}
    \tilde{\L}^{\tau,\kappa}f=\L^{\tau,\kappa} \mathbb{P}_{\lambda^{(k)}+a} f +\lambda^{(k)}(f-\mathbb{P}_{\lambda^{(k)}+a} f).
\end{align*}
Let $\{f_i\}$'s be the eigenvectors of $\L^{\tau,\kappa}$ associated with eigenvalues $\{\lambda^{(i)}\}$. We observe that $\tilde{\L}_{\tau,\kappa}$ is self-adjoint with respect to the $L^2(\gamma)$ inner product and shares the same eigenvectors with corresponding eigenvalues $\lambda^{(1)},\ldots,\lambda^{(k)},\lambda^{(k)},\ldots$. Let $\tilde{A}:=\tilde{\L}^{\tau,\kappa}|_{\mathcal{V}}$ and $\tilde{\lambda}_{\mathcal{V}}^{(1)},\ldots,\tilde{\lambda}_{\mathcal{V}}^{(k)}$ be its eigenvalues. Let $f\in \mathcal{V}$ and $g=f-\mathbb{P}_{\lambda^{(k)}+a}f$. Since $\L^{\tau,\kappa}\mathbb{P}_{\lambda^{(k)}+a}f =\tilde{\L}_{\tau,\kappa}\mathbb{P}_{\lambda^{(k)}+a}f$, we have by orthogonality
\begin{align}
    \langle f,\L^{\tau,\kappa}f \rangle_{L^2}- \langle f,\tilde{\L}^{\tau,\kappa}f\rangle_{L^2}
	&=\langle g,\L^{\tau,\kappa}g \rangle_{L^2}- \langle g,\tilde{\L}^{\tau,\kappa}g\rangle_{L^2}\nonumber\\
	&=\langle g,\L^{\tau,\kappa}g \rangle_{L^2}- \lambda^{(k)} \|g\|^2_{L^2} \geq \frac{a}{\lambda^{(k)}+a}\langle g,\L^{\tau,\kappa}g \rangle_{L^2},   \label{eq:A>A'}
\end{align}
where the last inequality follows from the fact that $\langle g,\L^{\tau,\kappa}g \rangle_{L^2} \geq (\lambda^{(k)}+a)\|g\|^2_{L^2}$. By the minimax principle, we have $\tilde{\lambda}_{\mathcal{V}}^{(j)} \geq \lambda^{(j)}$ for $j\leq k$ and by Theorem \ref{thm:evalRate} we have 
\begin{align*}
    \tilde{\lambda}_{\mathcal{V}}^{(j)}\geq  \lambda^{(j)}_{n}-C_{\M,\tau,\kappa}\lambda^{(k)}\left(\frac{\eps_n}{h_n}+h_n\sqrt{\lambda^{(k)}}\right) .
\end{align*}
Together with \eqref{eq:lambdaH}, we get
\begin{align*}
    \lambda_{\mathcal{V}}^{(j)}-\tilde{\lambda}_{\mathcal{V}}^{(j)} \leq C_{\M,\tau,\kappa}\lambda^{(k)}\left(\frac{\eps_n}{h_n}+h_n\sqrt{\lambda^{(k)}}\right), 
\end{align*}
and by \cite{burago2015graph}[Lemma 7.2],  for any $f\in \mathcal{V},$ 
\begin{align*}
    \langle f,\L^{\tau,\kappa}f \rangle_{L^2}-\langle f,\tilde{\L}^{\tau,\kappa}f \rangle_{L^2}=\langle f,Af \rangle_{L^2}- \langle f,\tilde{A}f\rangle_{L^2} \leq  k\underset{1\leq j\leq k}{\operatorname{max}} \{\lambda_{\mathcal{V}}^{(j)}-\tilde{\lambda}_{\mathcal{V}}^{(j)}\}\leq C_{\M,\tau,\kappa}k\lambda^{(k)}\left(\frac{\eps_n}{h_n}+h_n\sqrt{\lambda^{(k)}}\right),
\end{align*}
where we have used the fact that $A\geq \tilde{A}$ from \eqref{eq:A>A'}. 
Hence \eqref{eq:A>A'} implies
\begin{align*}
    D[g]&=\langle g,\L^{\tau,\kappa}g \rangle_{L^2} \leq C_{\M,\tau,\kappa}a^{-1}(\lambda^{(k)}+a)k \lambda^{(k)}\left(\frac{\eps_n}{h_n}+h_n\sqrt{\lambda^{(k)}}\right),\\
    \|g\|_{L^2}^2 &\leq C_{\M,\tau,\kappa} a^{-1}k\lambda^{(k)}\left(\frac{\eps_n}{h_n}+h_n\sqrt{\lambda^{(k)}}\right).  
\end{align*}
\end{proof}

The next lemma bounds the projection error when $J$ is a finite interval. 
\begin{lemma} \label{lemma:efun2}
Suppose $k:=k_n$ is such that $h_n\sqrt{\lambda^{(k_n)}}\ll 1$. Let $a\leq b \leq c \leq \lambda^{(k)}$ be constants so that the interval $(\lambda^{(k)}+a,\lambda^{(k)}+b)$ does not contain any eigenvalue of $\L^{\tau,\kappa}$. Then  
\begin{align*}
    \|\mathcal{I}v_k-\mathbb{P}_{(\lambda^{(k)}-c,\lambda^{(k)}+a]}\mathcal{I}v_k\|^2_{L^2} \leq C_{\M,\tau,\kappa}c^{-1}b^{-1}k\left[\lambda^{(k)}\right]^2 \left(\frac{\eps_n}{h_n}+h_n\sqrt{\lambda^{(k)}}\right)+c^{-1}a.
\end{align*}
\end{lemma}
\begin{proof}
Let $f=\mathcal{I}v_k$ and decompose it as 
\begin{align*}
    f=\mathbb{P}_{(\lambda^{(k)}-c,\lambda^{(k)}+a]} f + \mathbb{P}_{(-\infty,\lambda^{(k)}-c]} f + \mathbb{P}_{(\lambda^{(k)}+a,\infty)} f=:f_0+f_-+f_+\, .
\end{align*}
Orthogonality implies 
\begin{align*}
    \langle f,\L^{\tau,\kappa}f \rangle_{L^2}=\langle f_0,\L^{\tau,\kappa}f_0 \rangle_{L^2}+\langle f_-,\L^{\tau,\kappa}f_- \rangle_{L^2}+\langle f_+,\L^{\tau,\kappa}f_+ \rangle_{L^2},
\end{align*}
and we have by assumption that $f_+=\mathbb{P}_{[\lambda^{(k)}+b,\infty)}f$. 
By Lemma \ref{lemma:efun1}, we have 
\begin{align}
    \|f_+\|_{L^2}^2& \leq C_{\M,\tau,\kappa}b^{-1} k \lambda^{(k)}\left(\frac{\eps_n}{h_n}+h_n\sqrt{\lambda^{(k)}}\right), \label{eq:f+}\\
    \langle f_+,\L^{\tau,\kappa}f_+ \rangle_{L^2}&\leq C_{\M,\tau,\kappa} b^{-1}(\lambda^{(k)}+b) k \lambda^{(k)}\left(\frac{\eps_n}{h_n}+h_n\sqrt{\lambda^{(k)}}\right). \nonumber 
\end{align}
By Lemma \ref{lemma:efun0} (2), we have 
\begin{align*}
    \langle f,\L^{\tau,\kappa}f \rangle_{L^2}=D[\mathcal{I}v_k] \geq \left[1-C_{\M,\tau,\kappa} \left(\frac{\eps_n}{h_n}+h_n\sqrt{\lambda^{(k)}}\right)\right]D_{h_n}[v_k] 
    =\left[1-C_{\M,\tau,\kappa} \left(\frac{\eps_n}{h_n}+h_n\sqrt{\lambda^{(k)}}\right)\right] \lambda^{(k)}_{n}.
\end{align*}
Then 
\begin{align}
    \langle f_0,\L^{\tau,\kappa}f_0 \rangle_{L^2}+\langle f_-,\L^{\tau,\kappa}f_- \rangle_{L^2}
    &=\langle f,\L^{\tau,\kappa}f\rangle_{L^2}-\langle f_+,\L^{\tau,\kappa}f_+  \rangle_{L^2} \nonumber\\
    &\geq  \lambda^{(k)}_{n}- C_{\M,\tau,\kappa}b^{-1}(\lambda^{(k)}+b)k\lambda^{(k)} \left(\frac{\eps_n}{h_n}+h_n\sqrt{\lambda^{(k)}}\right). \label{eq:f0+f-} 
\end{align}
We also have 
\begin{align*}
    \langle f_0,\L^{\tau,\kappa}f_0\rangle_{L^2} &\leq (\lambda^{(k)}+a) \|f_0\|_{L^2}, \\
    \langle f_-,\L^{\tau,\kappa}f_-\rangle_{L^2} &\leq (\lambda^{(k)}-c)\|f_-\|_{L^2},
\end{align*}
which implies 
\begin{align*}
    \langle f_0,\L^{\tau,\kappa}f_0 \rangle_{L^2}+\langle f_-,\L^{\tau,\kappa}f_- \rangle_{L^2} 
    &\leq \lambda^{(k)}(\|f_0\|_{L^2}^2+\|f_-\|_{L^2}^2) + a\|f_0\|_{L^2}^2-c \|f_-\|^2_{L^2} \\
    &\leq \lambda^{(k)} \|f\|_{L^2}^2+a\|f\|_{L^2}^2-c \|f_-\|_{L^2}^2.
\end{align*}
By Lemma \ref{lemma:evalLB}(1), we have 
\begin{align*}
    \|f\|_{L^2} = \|\mathcal{I}v_k\|_{L^2} \leq \left[1+C_{\kappa}h_n\sqrt{\lambda^{(k)}_{n}}\right]\|v_k\|_2 \leq 1+C_{\kappa}h_n\sqrt{\lambda^{(k)}},
\end{align*}
which gives 
\begin{align*}
    \langle f_0,\L^{\tau,\kappa}f_0 \rangle_{L^2}+\langle f_-,\L^{\tau,\kappa}f_- \rangle_{L^2} \leq (\lambda^{(k)}+a) (1+C_{\kappa}h_n^2\lambda^{(k)})-c \|f_-\|_{L^2}^2. 
\end{align*}
Combining with  \eqref{eq:f0+f-} we have
\begin{align*}
    (\lambda^{(k)}+a) (1+C_{\kappa}h_n^2\lambda^{(k)})-c \|f_-\|_{L^2}^2 \geq \lambda^{(k)}_{n}- C_{\M,\tau,\kappa}b^{-1}(\lambda^{(k)}+b)k\lambda^{(k)} \left(\frac{\eps_n}{h_n}+h_n\sqrt{\lambda^{(k)}}\right),
\end{align*}
and then 
\begin{align}
    \|f_-\|_{L^2}^2 &\leq   C_{\M,\tau,\kappa}c^{-1}b^{-1}(\lambda^{(k)}+b)k\lambda^{(k)} \left(\frac{\eps_n}{h_n}+h_n\sqrt{\lambda^{(k)}}\right)+c^{-1}a+c^{-1}\left|\lambda^{(k)}_{n}-\lambda^{(k)}\right|\nonumber\\
    &\leq C_{\M,\tau,\kappa}c^{-1}b^{-1}k\left[\lambda^{(k)}\right]^2 \left(\frac{\eps_n}{h_n}+h_n\sqrt{\lambda^{(k)}}\right)+c^{-1}a\label{eq:f-},
\end{align}
where the assumption $b\leq \lambda^{(k)}$ is used in the last step. 
The result then follows by combining \eqref{eq:f+} and \eqref{eq:f-} and noticing that $\|\mathcal{I}v_k-\mathbb{P}_{(\lambda^{(k)}-c,\lambda^{(k)}+a]}\mathcal{I}v_k\|^2_{L2} =\|f_+\|_{L^2}^2+\|f_-\|_{L^2}^2$.
\end{proof}

Now we are ready to prove Theorem \ref{thm:efunRate}.
\begin{theorem}[Eigenfunction Approximation] \label{thm:efunRate2}
Let $\lambda$ be an eigenvalue of $\L^{\tau,\kappa}$ with multiplicity $\ell$, i.e., 
\begin{align*}
    \lambda^{(k_n-1)}<\lambda^{(k_n)}=\lambda=\ldots=\lambda^{(k_n+\ell-1)}<\lambda^{(k_n+\ell)}.
\end{align*}
Suppose that $h_n\sqrt{\lambda^{(k_n)}}\ll 1$ and $\eps_n\ll h_n$ for $n$ large.  Let $\psi^{(k_n)}_{n},\ldots,\psi^{(k_n+\ell-1)}_{n}$ be orthonormal eigenvectors of $L^{\tau,\kappa}_n$ associated with eigenvalues 
$\lambda^{(k_n)}_{n},\ldots,\lambda^{(k_n+\ell-1)}_{n}$. Then there exists orthonormal eigenfunctions $\psi^{(k_n)},\ldots,\psi^{(k_n+\ell-1)}$ of $\L^{\tau,\kappa}$ so that for
$j=k_n,\ldots,k_n+\ell-1$
\begin{align*}
    \|\psi^{(j)}_{n}\circ T_n-\psi^{(j)}\|_{L^2}^2 &\leq C_{\M,\tau,\kappa}  j^3  \left(\frac{\eps_n}{h_n}+h_n\sqrt{\lambda^{(j)}}\right),\\
    \|\psi^{(j)}_{n}\circ \mathcal{T}_n-\psi^{(j)}\|_{L^2}^2 &\leq C_{\M,\tau,\kappa} (\log n)^{mc_m}  j^3  \left(\frac{\eps_n}{h_n}+h_n\sqrt{\lambda^{(j)}}\right).
\end{align*}
\end{theorem}

\begin{proof}
For each $j=k,\ldots,k+\ell-1$, let $a=(\eps_n/h_n+h_n\sqrt{\lambda^{(j)}})$ and $b=c=\frac{\delta_{\lambda}}{2}$, where 
\begin{align*}
\delta_{\lambda}=\operatorname{min}\{\lambda^{(k_n)}-\lambda^{(k_n-1)}, \lambda^{(k_n+\ell)}-\lambda^{(k_n-\ell-1)}\}
\end{align*}
so that the assumptions of Lemma \ref{lemma:efun2} are satisfied. Indeed, $a\leq b\leq c \leq \lambda^{(k)}$ and the interval $(\lambda^{(j)}+a,\lambda^{(j)}+b)$ does not contain any eigenvalue of $\L^{\tau,\kappa}$ and $\mathbb{P}_{(\lambda^{(j)}-c,\lambda^{(j)}+a)}=\mathbb{P}_{\{\lambda^{(j)}\}}$. Hence we obtain
\begin{align*}
    \|\mathcal{I}\psi_n^{(j)}-\tilde{\psi}^{(j)}\|_{L^2}^2 \leq C_{\M,\tau,\kappa} \delta_{\lambda}^{-2} j \left[\lambda^{(j)}\right]^2 \left(\frac{\eps_n}{h_n}+h_n\sqrt{\lambda^{(j)}}\right),
\end{align*}
where $\tilde{\psi}^{(j)}=\mathbb{P}_{\{\lambda\}}\mathcal{I}\psi_n^{(j)}$  is a $\lambda$-eigenfunction of $\L^{\tau,\kappa}$. 
Lemma \ref{lemma:evalLB}(1) implies that $\mathcal{I}$ is almost an isometry  on the span of $\psi_n^{(k)},\ldots,\psi^{(k+\ell-1)}_n$ and by the polarization identity we get that the $\mathcal{I}\psi_n^{(j)}$'s are almost orthonormal up to  $C_{\M,\tau,\kappa}h_n\sqrt{\lambda^{(j)}}$. This implies the $\tilde{\psi}^{(j)}$'s are almost orthogonal up to $C_{\M,\tau,\kappa}\delta_{\lambda}^{-2}j[\lambda^{(j)}]^2(\eps_n/h_n+h_n\sqrt{\lambda^{(j)}})$. Hence letting $\{\psi^{(j)}\}_{j=k}^{k+\ell-1}$ be the Gram-Schmidt orthogonalization of $\{\tilde{\psi}^{(j)}\}_{j=k}^{k+\ell-1}$, we get
\begin{align*}
    \|\mathcal{I}\psi_n^{(j)}-\psi^{(j)}\|_{L^2}^2 \leq C_{\M,\tau,\kappa} \delta_{\lambda}^{-2} j \left[\lambda^{(j)}\right]^2 \left(\frac{\eps_n}{h_n}+h_n\sqrt{\lambda^{(j)}}\right).
\end{align*}
Using  \eqref{eq:IuP*u}  that $\|\mathcal{I}v-P^*v \|^2_{L^2} \leq Ch_n^2D_{h_n}[v]$ gives 
\begin{align*}
    \|P^*\psi_n^{(j)}-\psi^{(j)}\|_{L^2}^2 \leq C_{\M,\tau,\kappa} \delta_{\lambda}^{-2} j \left[\lambda^{(j)}\right]^2 \left(\frac{\eps_n}{h_n}+h_n\sqrt{\lambda^{(j)}}\right).
\end{align*}
By Weyl's law that $\lambda^{(j)} \asymp j^{\frac{2}{m}}$ and hence $\delta_{\lambda} \asymp j^{\frac{2}{m}-1} \asymp j^{-1}\lambda^{(j)}$, we conclude that 
\begin{align*}
    \|P^*\psi_n^{(j)}-\psi^{(j)}\|_{L^2}^2 \leq C_{\M,\tau,\kappa}  j^3  \left(\frac{\eps_n}{h_n}+h_n\sqrt{\lambda^{(j)}}\right),
\end{align*}
which is the first assertion of the theorem by noticing that $P^{*}\psi_n^{(j)}=\psi_n^{(j)}\circ T_n$. 
Now by Lemma 17 and the proof of Theorem 6 in \cite{trillos2019error}, we have  
\begin{align*}
    \|\psi_n^{(j)}\circ \mathcal{T}_n-\psi^{(j)} \|_{L^2} &\leq C_{\M}\left[\lambda_j^{\frac{m+1}{4}} \eps_n + (\log n)^{\frac{mc_m}{2}} \|P^*\psi_n^{(j)}-\psi^{(j)}\|_{L^2} \right]\\
    &\leq C_{\M,\tau,\kappa}\left[j^{\frac{m+1}{2m}}\eps_n+(\log n)^{\frac{mc_m}{2}}j^{\frac{3}{2}}\sqrt{\frac{\eps_n}{h_n}+h_n\sqrt{\lambda^{(j)}}}\right]\\
    &\leq C_{\M,\tau,\kappa}(\log n)^{\frac{mc_m}{2}}j^{\frac{3}{2}}\sqrt{\frac{\eps_n}{h_n}+h_n\sqrt{\lambda^{(j)}}},
\end{align*}
where we have used that $\eps_n\ll h_n$ in the last step. 
\end{proof}

\section{Convergence of Gaussian Mat\'ern Field} \label{sec:APTL2}

Now we are ready to prove Theorem \ref{thm:Rate}. Theorem \ref{thm:Rate2} can be proved in the same fashion using the second assertion of Theorem \ref{thm:efunRate2}.
\begin{theorem} \label{thm:rate}
Suppose $\tau$ is Lipschitz, $\kappa\in C^1(\M)$ and both are bounded below by positive constants. 
Let $s>m$ and 
\begin{align*}
    \frac{(\log n)^{c_m}}{n^{1/m}} \ll h_n \ll \frac{1}{n^{1/2s}},
\end{align*}
where $c_m=3/4$ if $m=2$ and $c_m=1/m$ otherwise. 
Then, with probability one, 
\begin{align*}
    \mathbb{E}\|u_n\circ T_n-u\|_{L^2} \xrightarrow{n\rightarrow{\infty}} 0.
\end{align*}
If further $s>(5m+1)/2$ and
\begin{align}
    h_n\asymp \sqrt{\frac{(\log n)^{c_m}}{n^{1/m}}}. \label{eq:hnScaling}
\end{align}
Then, with probability one, 
\begin{align*}
    \mathbb{E}\|u_n\circ T_n-u\|_{L^2} =O\left(\sqrt{h_n}\right)= O\left(\frac{(\log n)^{c_m/4}}{n^{1/4m}}\right). 
\end{align*}
\end{theorem}

\begin{proof}
Suppose we are in a realization where the conclusion of Proposition \ref{prop:transmap} holds.
Suppose $k_n$ is chosen so that $n^{m/2s}\ll k_n \ll h_n^{-m}$. Note that this is possible given the scaling of $h_n$. Then by Theorem \ref{thm:efunRate2} we can fix orthonormal eigenfunctions $\{\psi^{(i)}\}_{i=1}^{\infty}$ of $\L^{\tau,\kappa}$ and  $\{\psi^{(n)}_i\}_{i=1}^n$ of $L^{\tau,\kappa}_n$ for each $n$ so that 
\begin{align}
     \|\psi_n^{(i)}\circ T_n-\psi^{(i)}\|_{L^2} \leq C_{\M,\tau,\kappa}i^{\frac{3}{2}}  \sqrt{\frac{\eps_n}{h_n}+h_n\sqrt{\lambda^{(i)}}},\label{eq:efunBound}
\end{align}
for $i=1,\ldots,k_n$. Recall $u_n$ and $u$  have the following representations 
\begin{align*}
    u_n&:= \tau_n^{\frac{s}{2}-\frac{m}{4}}\kappa_n^{\frac{m}{2}}\sum_{i=1}^n \left[\lambda_n^{(i)}\right]^{-\frac{s}{2}} \xi^{(i)} \psi_n^{(i)}, \\
    u& =\tau^{\frac{s}{2}-\frac{m}{4}}\kappa^{\frac{m}{2}}\sum_{i=1}^{\infty} \left[\lambda^{(i)}\right]^{-\frac{s}{2}} \xi^{(i)} \psi^{(i)}.
\end{align*}
Since $\tau_n$ is the restriction of $\tau$ to $\M_n$, we have $\tau_n\circ T_n=\tau\circ T_n$ and similarly for $\kappa_n$. Therefore $u_n\circ T_n$ has the expression
\begin{align*}
    u_n\circ T_n=[\tau\circ T_n]^{\frac{s}{2}-\frac{m}{4}}[\kappa\circ T_n]^{\frac{m}{2}}\sum_{i=1}^n\left[\lambda_n^{(i)}\right]^{-\frac{s}{2}}\xi^{(i)}\psi_n^{(i)}\circ T_n.
\end{align*}
To bound the expected $L^2$ distance between $u_n\circ T_n$ and $u$, we introduce four intermediate functions. 
\begin{align*}
   u_n^{k_n}&:= [\tau\circ T_n]^{\frac{s}{2}-\frac{m}{4}}[\kappa\circ T_n]^{\frac{m}{2}} \sum_{i=1}^{k_n} \left[\lambda_n^{(i)}\right]^{-\frac{s}{2}} \xi^{(i)} \psi_n^{(i)}\circ T_n, \\
   \tilde{u}_n^{k_n}& := [\tau\circ T_n]^{\frac{s}{2}-\frac{m}{4}}[\kappa\circ T_n]^{\frac{m}{2}} \sum_{i=1}^{k_n} \left[\lambda^{(i)}\right]^{-\frac{s}{2}} \xi^{(i)} \psi_n^{(i)}\circ T_n, \\
   \tilde{u}^{k_n}&:=[\tau\circ T_n]^{\frac{s}{2}-\frac{m}{4}}[\kappa\circ T_n]^{\frac{m}{2}} \sum_{i=1}^{k_n} \left[\lambda^{(i)}\right]^{-\frac{s}{2}} \xi^{(i)} \psi^{(i)}\\
   u^{k_n}&:= \tau^{\frac{s}{2}-\frac{m}{4}}\kappa^{\frac{m}{2}} \sum_{i=1}^{k_n} \left[\lambda^{(i)}\right]^{-\frac{s}{2}} \xi^{(i)} \psi^{(i)}.
\end{align*}
It then suffices to bound the difference between any two consecutive functions. By Theorem \ref{thm:evalRate} and  Weyl's law we have that  $\lambda^{(k_n)}_{n} \gtrsim \lambda^{(k_n)} \gtrsim k_n^{2/m} $, which gives
\begin{align}
    \mathbb{E}\|u_n\circ T_n-u_n^{k_n}\|_{L^2}
    &\leq \beta^{\frac{s}{2}-\frac{m}{4}}\alpha^{\frac{m}{2}}\left(\sum_{i=k_n+1}^n  \left[\lambda_n^{(i)}\right]^{-s} \right) ^{\frac12} \lesssim \left( n\left[\lambda^{(k_n)}_{n}\right]^{-s}\right)^{\frac12} 
    \lesssim  \sqrt{n}k_n^{-\frac{s}{m}}. \label{eq:1}
\end{align}
Similarly, 
\begin{align}
    \mathbb{E}\|u^{k_n}-u\|_{L^2}
    &\lesssim   \left( \sum_{i=k_n+1}^{\infty} \left[\lambda^{(i)}\right]^{-s}\right)^{\frac12} 
    \lesssim \left( \sum_{i=k_n+1}^{\infty} i^{-\frac{2s}{m}} \right)^{\frac12} 
    \lesssim \left( \int _{k_n}^{\infty} x^{-\frac{2s}{m}} \right)^{\frac12} 
    \lesssim  k_n^{\frac{1}{2}-\frac{s}{m}}.  \label{eq:2}
\end{align}
Both \eqref{eq:1} and \eqref{eq:2} converges to 0 by the choice of $k_n$. 

Next, since both $\lambda_n^{(i)}$ and $\lambda^{(i)}$ are bounded below by $\operatorname{min}\tau>0$, by Lipschitz continuity of $x^{-s/2}$ away from 0 we have, for $i=1,\ldots,k_n,$
\begin{align*}
    \left|\left[\lambda_n^{(i)}\right]^{-\frac{s}{2}}-\left[\lambda^{(i)}\right]^{-\frac{s}{2}}\right| \lesssim \left(\left[\lambda_n^{(i)}\right] \wedge \left[\lambda^{(i)}\right]\right)^{-\frac{s}{2}-1} \left|\lambda_n^{(i)}-\lambda^{(i)} \right| 
    \lesssim \left[\lambda^{(i)}\right]^{-\frac{s}{2}} \left(\frac{\eps_n}{h_n}+h_n\sqrt{\lambda^{(i)}}\right)
\end{align*}
Hence 
\begin{align}
    \mathbb{E}\|u_n^{k_n}-\tilde{u}_n^{k_n}\|_{L^2}
    &\lesssim  \left(\sum_{i=1}^{k_n} \left(\left[\lambda_n^{(i)}\right]^{-\frac{s}{2}}-\left[\lambda^{(i)}\right]^{-\frac{s}{2}}\right)^2\right)^{\frac12} \nonumber\\
    &\lesssim  \left(\sum_{i=1}^{k_n} \left[\lambda^{(i)}\right]^{-s}\left(\frac{\eps_n}{h_n}+h_n\sqrt{\lambda^{(i)}}\right)^2\right)^{\frac12}  \label{eq:3}\\
	&\lesssim \left(\frac{\eps_n}{h_n} +h_n k_n^{\frac{1}{m}} \right)\left(\sum_{i=1}^{k_n} i^{-\frac{2s}{m}}\right)^{\frac12}. \nonumber
\end{align}
The last expression goes to 0 since $s>m$, $\eps_n\ll h_n$ and $k_n\ll h_n^{-m}$.

Now by Lipschitz continuity of $\tau$ and $\kappa$, and Proposition \ref{prop:transmap} that $d(x,T_n(x))\leq \eps_n$, we have for all $x\in\M$
\begin{align*}
    \left|\tau(T_n(x))^{\frac{s}{2}-\frac{m}{2}}\kappa(T_n(x))^{\frac{m}{2}} - \tau(x)^{\frac{s}{2}-\frac{m}{2}}\kappa(x)^{\frac{m}{2}} \right| \lesssim \eps_n. 
\end{align*}
Therefore 
\begin{align}
    \mathbb{E}\|\tilde{u}^{k_n}-u^{k_n}\|_{L^2}\lesssim \left\|(\tau\circ T_n)^{\frac{s}{2}-\frac{m}{4}}(\kappa\circ T_n)^{\frac{m}{2}}-\tau^{\frac{s}{2}-\frac{m}{4}}\kappa^{\frac{m}{2}}\right\|_{\infty}  \left(\sum_{i=1}^{k_n}\left[\lambda^{(i)}\right]^{-s}\right) \lesssim \eps_n, \label{eq:4}
\end{align}
which converges to zero. 

Finally, for fixed $\ell \in \mathbb{N}$, we have by using the fact that $\|\psi_n^{(i)}\circ T_n\|_{L^2}=\|\psi_n^{(i)}\|_2=1$   
\begin{align}
    \mathbb{E}\|\tilde{u}_n^{k_n}-\tilde{u}^{k_n} \|_{L^2}
    &\lesssim   \sum_{i=1}^{k_n} \left[\lambda^{(i)}\right]^{-\frac{s}{2}} \|\psi_n^{(i)}\circ T_n-\psi^{(i)}\|_{L^2} \label{eq:5} \\
    &\lesssim \sum_{i=1}^{\ell} \left[\lambda^{(i)}\right]^{-\frac{s}{2}} \|\psi_n^{(i)}\circ T_n-\psi^{(i)}\|_{L^2} + \sum_{i=\ell+1}^{k_n}\left[\lambda^{(i)}\right]^{-\frac{s}{2}}.   \nonumber 
\end{align}
By \eqref{eq:efunBound} we have $\|\psi_n^{(i)}\circ T_n-\psi^{(i)}\|_{L^2}\xrightarrow{n\rightarrow \infty} 0$ for $i=1,\dots, \ell$ since $\ell$ is fixed. Therefore we have 
\begin{align*}
    \underset{n\rightarrow \infty}{\operatorname{lim\, sup}}\,\mathbb{E}\|\tilde{u}_n^{k_n}-\tilde{u}^{k_n} \|_{L^2} \lesssim  \sum_{i=\ell+1}^{\infty}\left[\lambda^{(i)}\right]^{-\frac{s}{2}}.
\end{align*}
Since $\ell$ is arbitrary, the last expression goes to 0 as $\ell\rightarrow \infty$ under the assumption $s>m$. 
Hence by combining all the pieces we get 
$\mathbb{E}\|u_n\circ T_n-u\|_{L^2} \xrightarrow{n \rightarrow \infty} 0.$

Now in order to obtain rates of convergence, we need the additional assumption that $s>\frac{5}{2}m+\frac12$ and to refine in particular the estimates for $\mathbb{E}\|u_n^{k_n}-\tilde{u}_n^{k_n}\|_{L^2}$ and $\mathbb{E}\|\tilde{u}_n^{k_n}-\tilde{u}^{k_n} \|_{L^2}$. By \eqref{eq:3}, we have 
\begin{align}
    \mathbb{E}\|u_n^{k_n}-\tilde{u}_n^{k_n}\|_{L^2} 
	\lesssim \left(\frac{\eps_n}{h_n} +h_n \right)\left(\sum_{i=1}^{k_n} \left[\lambda^{(i)}\right]^{-s+1}\right)^{\frac12}\lesssim \left(\frac{\eps_n}{h_n} +h_n \right), \label{eq:6}
\end{align}
where the last step follows by the assumption $s>\frac{5}{2}m+\frac12$. 
By \eqref{eq:efunBound}, we can further bound \eqref{eq:5} by
\begin{align}
    \mathbb{E}\|\tilde{u}_n^{k_n}-\tilde{u}^{k_n} \|_{L^2}
    &\lesssim   \sum_{i=1}^{k_n} \left[\lambda^{(i)}\right]^{-\frac{s}{2}} i^{\frac32} \sqrt{\frac{\eps_n}{h_n}+h_n\sqrt{\lambda^{(i)}}} \nonumber \\
    & \lesssim \sqrt{\frac{\eps_n}{h_n} +h_n}\sum_{i=1}^{k_n}i^{\frac{3}{2}} \left[\lambda^{(i)}\right]^{-\frac{s}{2}+\frac14}\nonumber\\
    &\lesssim  \sqrt{\frac{\eps_n}{h_n} +h_n}\sum_{i=1}^{k_n}i^{\frac{3}{2}-\frac{s}{m}+\frac{1}{2m}} \lesssim  \sqrt{\frac{\eps_n}{h_n} +h_n}, \label{eq:7}
\end{align}
where the last step follows from that $s>\frac{5}{2}m+\frac12$. Now combining \eqref{eq:1}, \eqref{eq:2}, \eqref{eq:4}, \eqref{eq:6}, \eqref{eq:7}, we see that the error is dominated by     
\begin{align*}
    \mathbb{E} \|u_n\circ T_n-u\|_{L^2}  &\lesssim \sqrt{n}k_n^{-\frac{s}{m}}+\sqrt{\frac{\eps_n}{h_n} +h_n} .
\end{align*}
Therefore by setting
\begin{align*}
    h_n\asymp (\log n)^{\frac{c_m}{2}}n^{-\frac{1}{2m}}, \quad k_n \asymp n^{\frac{2m+1}{4s}}
\end{align*}
we have 
\begin{align*}
    \mathbb{E} \|u_n\circ T_n-u\|_{L^2} \lesssim (\log n)^{\frac{c_m}{4}} n^{-\frac{1}{4m}}. 
\end{align*}
\end{proof}

We end this section with a remark that Theorem \ref{thm:rate} can be stated in terms of the $TL^2$ metric proposed in \cite{trillos2016continuum}. Let $\mathcal{P}(\M)$ be the space of Borel probability measures on $\M$. Define the $TL^2$ space as 
\begin{align*}
    TL^2:= \left\{(\mu,f): \mu \in \mathcal{P}(\M), f \in L^2(\mu)\right\},
\end{align*}
endowed with the metric 
\begin{align*}
    d_{TL^2}\Big((\mu_1,f_1),(\mu_2,f_2) \Big):= \underset{\omega \in \mathscr{C}(\mu_1,\mu_2) }{\operatorname{inf}} \, \left[ \int_{\M}\int_{\M} \left(d_{\M}(x,y)^2+|f_1(x)-f_2(y)|^2\right) d\omega(x,y) \right]^{\frac12},   
\end{align*}
where $\mathscr{C}$ is the set of couplings between $\mu_1$ and $\mu_2$ and $d_{\M}$ is the geodesic distance on $\M$. The $d_{TL^2}$ metric is a natural generalization of $L^2$ convergence of functions and weak convergence of  measures \cite{trillos2016continuum}, which allows comparison of functions defined over the point cloud  with functions defined on $\M$. It bypasses the need to  consider  a specific transport map and thus may be of independent interest. 

The assertions of Theorem \ref{thm:rate} continue to hold if $\mathbb{E}\|u_n\circ T_n-u\|_{L^2}$ is replaced by $\mathbb{E}\left[d_{TL^2}\Big((\gamma_n,u_n),(\gamma,u) \Big)\right]$. The proof follows immediately from the definition since the transport maps $T_n$ induce a coupling defined as $\omega_{T_n}:=(I\times T_n)_{\sharp} \gamma$, the push-forward of $\gamma$ under $I\times T_n:\M\rightarrow \M\times \M_n$, where $(I\times T_n)(x)= \bigl(x,T_n(x) \bigr)$. Hence we see that 
\begin{align*}
    d_{TL^2}\Big((\gamma_n,u_n),(\gamma,u)\Big)&\leq  \left[\int_{\M} \left(d_{\M} \bigl(x,T_n(x)\bigr)^2+|u_n \bigl(T_n(x)\bigr)-u(x)|^2\right) d\gamma(x)\right]^{\frac{1}{2}}\\
    &\lesssim  \eps_n+\|u_n\circ T_n-u\|_{L^2}.
\end{align*}

\end{appendix}

\end{document}